\documentclass[acmsmall,screen]{acmart}
\renewcommand\footnotetextcopyrightpermission[1]{} %
\usepackage{caption}
\usepackage{subcaption}
\usepackage{fancyvrb}
\usepackage{setspace}
\usepackage{url}
\usepackage{color}
\usepackage{setspace}
\usepackage{paralist}
\usepackage{amsmath,amssymb,amsfonts,amsthm}
\usepackage[T1]{fontenc}
\usepackage[latin9]{inputenc}
\usepackage{listings}
\lstdefinestyle{nonumbering}{
        basicstyle=\ttfamily\scriptsize,
        language=Java,
        emph={},
        emphstyle={\underbar},
        numbers=none,
        numberfirstline=false,
        numberstyle=\tiny,
        escapeinside={(*@}{@*)}
}
\lstdefinestyle{small}{
        basicstyle=\tiny,
        language=Java,
        emph={},
        emphstyle={\underbar},
        numbers=none,
        numberfirstline=false,
        numberstyle=\tiny,
        escapeinside={(*@}{@*)}
}

\usepackage{hyperref}
\usepackage{cleveref}
\usepackage{pdfpages}

\usepackage{float}
\usepackage{graphics}
\usepackage{graphicx}
\usepackage{wrapfig}
\usepackage{epsfig}
\usepackage[linesnumbered,ruled,vlined]{algorithm2e}
\usepackage{rotating}
\usepackage{mathpartir}
\usepackage{multirow}
\usepackage{mathtools}
\usepackage{hhline}

\Crefname{figure}{Fig.}{Fig.}

\newcommand{\Iaux}{\Inv_{\text{aux}}}

\newcommand{\cassume}{\bf{assume}}

\renewcommand{\phi}{\varphi}

\newcommand\substitute[3]{{{#1} \left[ #3 ~/~ #2 \right]}}
\newcommand\substitutemany[5]{{{#1} \left[ #3 ~/~ #2,\ldots, #5 ~/~ #4 \right]}}

\newcommand{\true}{{\textit{true}}}

\newcommand{\vocabulary}{{\ensuremath{\Sigma}}}

\newcommand{\srts}{\mathcal{S}}

\newcommand{\func}{f}
\newcommand{\srt}{s}

\newcommand{\s}{s}
\newcommand{\shat}{\hat{\s}}
\newcommand{\Reach}{R}
\newcommand{\Reachhat}{\hat{R}}

\newcommand{\init}{\textit{INIT}}
\newcommand{\initu}{\textit{INIT}_u}
\newcommand{\inithat}{\widehat{\init}}

\newcommand{\TR}{\textit{TR}}
\newcommand{\smallTR}{g}
\newcommand{\TRhat}{\widehat{\textit{TR}}}
\newcommand{\Sigmahat}{\widehat{\Sigma}}
\newcommand{\varphihat}{\widehat{\varphi}}
\newcommand{\thetahat}{\widehat{\theta}}
\newcommand{\Inv}{\textit{INV}}
\newcommand{\Invhat}{\widehat{\Inv}}
\newcommand{\safe}{P}
\newcommand{\TRu}{{\textit{TR}_u}}

\newcommand{\sort}[1]{{\text{#1}}}
\newcommand{\snode}{\sort{node}}
\newcommand{\squorum}{\sort{quorum}}
\newcommand{\sround}{\sort{round}}
\newcommand{\svalue}{\sort{value}}
\newcommand{\sinst}{\sort{instance}}
\newcommand{\svotemap}{\sort{votemap}}
\newcommand{\snodeset}{\sort{nodeset}}
\newcommand{\smap}{\sort{map}}
\newcommand{\sfquorum}{\sort{f\_quorum}}
\newcommand{\scquorum}{\sort{c\_quorum}}
\newcommand{\sconfig}{\sort{config}}

\newcommand{\relation}[1]{{\textit{#1}}}
\newcommand{\rmember}{\relation{member}}
\newcommand{\ronea}{\relation{start\_round\_msg}}
\newcommand{\ronebmaxvote}{\relation{join\_ack\_msg}}
\newcommand{\ronebproj}{\relation{joined\_round}}
\newcommand{\rleftround}{\relation{left\_round}}
\newcommand{\rproposal}{\relation{propose\_msg}}
\newcommand{\rvote}{\relation{vote\_msg}}
\newcommand{\rdecision}{\relation{decision}}
\newcommand{\rroundof}{\relation{roundof}}
\newcommand{\rvalueof}{\relation{valueof}}
\newcommand{\rapply}{\relation{apply}}
\newcommand{\ravailable}{\relation{available}}
\newcommand{\ractive}{\relation{active}}
\newcommand{\rfast}{\relation{fast}}
\newcommand{\rany}{\relation{any\_msg}}
\newcommand{\rfmember}{\relation{f\_member}}
\newcommand{\rcmember}{\relation{c\_member}}
\newcommand{\rquorumin}{\relation{quorum\_in}}
\newcommand{\rcompleteof}{\relation{complete\_of}}
\newcommand{\rcompletemsg}{\relation{complete\_msg}}
\newcommand{\rconfig}{\relation{configure\_round\_msg}}
\newcommand{\rmastercomp}{\relation{master\_complete}}
\newcommand{\rquorumof}{\relation{quorum\_of\_round}}
\newcommand{\rstop}{\relation{stop}}

\newcommand{\action}[1]{{\textsc{#1}}}
\newcommand{\asendonea}{\action{start\_round}}
\newcommand{\ajoinround}{\action{join\_round}}
\newcommand{\apropose}{\action{propose}}
\newcommand{\acastvote}{\action{vote}}
\newcommand{\adecide}{\action{learn}}
\newcommand{\ainstate}{\action{instate\_round}}
\newcommand{\aproposenew}{\action{propose\_new\_value}}

\newcommand{\acdecide}{\action{c\_decide}}
\newcommand{\afdecide}{\action{f\_decide}}
\newcommand{\aconfig}{\action{configure\_round}}
\newcommand{\amark}{\action{mark\_complete}}

\newcommand{\none}{{\bot}}

\DeclarePairedDelimiter{\floor}{\lfloor}{\rfloor}

\setcopyright{none}

\begin{document}

\title{Paxos Made EPR: Decidable Reasoning about Distributed Protocols}
\titlenote{This is the full version of a paper presented in OOPLSA 2017 \cite{oopsla17-epr}.
This version includes appendices that do not appear in the conference proceedings, and a slightly improved model of Multi-Paxos.}

\author{Oded Padon}
\affiliation{
  \institution{Tel Aviv University}
  \country{Israel}
}
\email{odedp@mail.tau.ac.il}

\author{Giuliano Losa}
\affiliation{
  \institution{University of California, Los Angeles}
  \country{USA}
}
\email{giuliano@cs.ucla.edu}

\author{Mooly Sagiv}
\affiliation{
  \institution{Tel Aviv University}
  \country{Israel}
}
\email{msagiv@post.tau.ac.il}

\author{Sharon Shoham}
\affiliation{
  \institution{Tel Aviv University}
  \country{Israel}
}
\email{sharon.shoham@gmail.com}

\begin{abstract}
Distributed protocols such as Paxos play an important role in many
computer systems. Therefore, a bug in a distributed protocol may have
tremendous effects. Accordingly, a lot of effort has been invested in
verifying such protocols. However, checking invariants of such
protocols is undecidable and hard in practice, as it requires
reasoning about an unbounded number of nodes and messages. Moreover,
protocol actions and invariants involve both quantifier alternations
and higher-order concepts such as set cardinalities and arithmetic.

This paper makes a step towards automatic verification of such
protocols. We aim at a technique that can verify correct protocols and
identify bugs in incorrect protocols. To this end, we develop a
methodology for deductive verification based on effectively
propositional logic (EPR)---a decidable fragment of first-order logic
(also known as the Bernays-Sch\"onfinkel-Ramsey class). In addition to
decidability, EPR also enjoys the finite model property, allowing to
display violations as finite structures which are intuitive for
users. Our methodology involves modeling protocols using general
(uninterpreted) first-order logic, and then systematically
transforming the model to obtain a model and an inductive invariant
that are decidable to check. The steps of the transformations are also
mechanically checked, ensuring the soundness of the method. We have
used our methodology to verify the safety of Paxos, and several of its
variants, including Multi-Paxos, Vertical Paxos, Fast Paxos,
Flexible Paxos and Stoppable Paxos. To the best of our knowledge, this work is the first
to verify these protocols using a decidable logic, and the first
formal verification of Vertical Paxos, Fast Paxos and Stoppable Paxos.
\end{abstract}



\maketitle
\thispagestyle{plain} %
\section{Introduction}

Paxos is a family of protocols for solving consensus in a network of
unreliable processors with unreliable communication.  Consensus is the
process of deciding on one result among a group of participants. Paxos
protocols play an important role in our daily life.  For example, Google
uses the Paxos algorithm in their Chubby distributed lock service in
order to keep replicas consistent in case of
failure~\cite{Chuby}. VMware uses a Paxos-based protocol within the
NSX Controller. Amazon Web Services uses Paxos-like algorithms
extensively to power its platform \cite{amazon}.  The key safety property of Paxos
is consistency: processors cannot decide on different values.

Due to its importance, verifying the safety of distributed protocols
like Paxos is an ongoing research challenge.  The systems and
programming languages communities have had several recent success
stories in verifying the safety of Paxos-like protocols in projects
such as IronFleet~\cite{IronFleet},
Verdi~\cite{DBLP:conf/pldi/WilcoxWPTWEA15}, and
PSync~\cite{dragoi_psync:_2016}\footnote{IronFleet and PSync also
  verify certain liveness properties.}.

\subsection{Main Results}
This work aims to increase the level of automation in verification of
distributed protocols, hoping that it will eventually lead to wider
adoption of formal verification in this domain. We follow IronFleet,
Verdi, and PSync, in requiring that the user supplies inductive
invariants for the protocols.
We aim to automate the process of checking the inductiveness of the
user-supplied invariants.
The goal is that the system can reliably produce in finite time either a proof that the invariant is inductive
or display a comprehensible counterexample to induction (CTI), i.e., a
concrete transition of the protocol from state $s$ to state $s'$ such
that $s$ satisfies the given invariant and $s'$ does not\footnote{Such
  a CTI indicates that there is a bug in the protocol itself, or that
  the provided invariant is inadequate (e.g., too weak or too
  strong).}.
Such a task seems very difficult since these protocols are usually
expressed in rich programming languages in which automatically
checking inductive invariants is both undecidable and very hard in
practice.  In fact, in the IronFleet project, it was observed that
undecidability of the reasoning performed by Z3~\cite{z3} is a major
hurdle in their verification
process.

\subsubsection{Criteria for Automatic Deductive Verification}
We aim for an automated deductive verification technique that achieves three goals:
\begin{description}
\item [Natural] Making the invariants readable even for users who are not expert in the tools.
\item [Completeness] Making sure that if the invariant is inductive then the solver is guaranteed to
prove it.
\item [Finite Counterexamples] Guaranteeing that if the invariant is not inductive then the
solver can display a concrete counterexample to induction with a finite number of nodes
which can be diagnosed by users.
\end{description}
These goals are highly ambitious. Expressing the verification
conditions in a decidable logic with a small model property (e.g.,
EPR~\cite{JAR:PiskacMB10}) will guarantee Completeness and Finite
Counterexamples. However, it is not clear how to model complex protocols
like Paxos in such logics.  Consensus
protocols such as Paxos often require higher-order reasoning about
sets of nodes (majority sets or quorums), combined with complex
quantification. In fact, some researchers conjectured that decidable
logics are too restrictive to be useful.

Furthermore, we are aiming to obtain natural invariants.
We decided to verify the designs of the protocols and not their
implementations since the invariants are more natural and since we
wanted to avoid dealing with low level implementation issues.
In the future we plan to use refinement to synthesize efficient low level
implementations. Systems such as Alloy~\cite{alloy} and
TLA~\cite{tla} have already been used for finding bugs in
protocols and inductive invariants (e.g., by Amazon~\cite{amazon}).
Again they verify and identify faults in
the designs and not the actual implementation. However, in contrast to our approach, they cannot automatically
produce proofs for inductiveness (Completeness).

\subsubsection{A Reusable Verification Methodology}
In this work, we develop a novel reusable verification methodology based on Effectively Propositional logic (EPR) for achieving the above goals.
Our methodology allows the expression of complex protocols and systems, while guaranteeing that the verification conditions are expressed in EPR. EPR provides both decidability and finite counterexamples, and is supported by existing solvers (e.g., Z3~\cite{z3}, iProver~\cite{DBLP:conf/cade/Korovin08}, VAMPIRE~\cite{vampire}, CVC4~\cite{DBLP:conf/cav/BarrettCDHJKRT11}).
We have used our methodology to verify the safety of Paxos, and
several of its variants, including Multi-Paxos, Vertical Paxos, Fast
Paxos, Flexible Paxos and Stoppable Paxos. To the best of our knowledge, this work is
the first to verify these protocols using a decidable logic, and,
in the case of Vertical Paxos, Fast Paxos, and Stoppable Paxos, it is also the first mechanized safety proof.

We have also compared our methodology to a traditional approach based
on a state-of-the-art interactive theorem
prover---Isabelle/HOL~\cite{nipkow_isabelle/hol:_2002}. Our comparison
shows that the inductive invariants used are very similar in both
approaches (Natural), and that our methodology allows more reliable
and predictable automation: an interactive theorem prover can
discharge proof obligations to theorem provers using undecidable
theories, but these often fail due to the undecidability. In such
cases, it requires an experienced expert user to prove the inductive
invariant. In contrast, with our methodology all the verification
conditions are decidable and therefore checking them is fully
automated.

\paragraph{First-order uninterpreted abstraction}
The first phase in our verification process is expressing the system and invariant in (undecidable) many-sorted first-order logic
over uninterpreted structures.
This is in contrast to SMT which allows the use of interpreted theories such as arithmetic and the theory of arrays.
The use of theories is natural specifically for handling low level aspects such as machine arithmetic and low level storage.
However, SMT leads to inherent undecidability with quantifiers which are used to model unbounded systems.
In contrast to SMT, we handle concepts, such as arithmetic and set cardinalities, using abstraction
expressible in first-order logic, e.g., a totally ordered set instead of the natural numbers.
This involves coming up with domain knowledge encoded as first-order axioms (e.g. a first-order formula expressing transitivity of a total order).

We are encouraged by the simplicity of our abstractions and the fact that they are precise enough to prove complex protocols.
We also note that using first-order logic has led us to axioms and invariants that elegantly capture the essence of the protocols.
This is also enabled by the fact that we are modeling high-level protocols and not their low level implementations.

At the end of this phase, the verification conditions are in general first-order logic. This is already useful as it allows to use resolution-based theorem provers
(e.g., SPASS~\cite{spass} and VAMPIRE~\cite{vampire}). Yet, at this stage the verification conditions are still undecidable, and solvers are not guaranteed to terminate.

One way to obtain decidability is to restrict quantifier alternations.
We examine the \emph{quantifier alternation graph} of the verification condition, which connects sorts that alternate in
$\forall\exists$ quantification. When this graph contains cycles, solvers such as Z3 often diverge into infinite loops while instantiating quantifiers.
This issue is avoided when the graph is acyclic, in which case the verification condition is essentially in EPR.
Therefore, the second phase of our methodology provides a systematic way
to soundly eliminate the cycles.

\paragraph{Eliminating quantifier alternations using derived relations}
The most creative part in our methodology is adding derived
relations and rewriting the code to break the cycles in
the quantifier alternation graph.
The main idea is to capture an existential formula by a derived relation,
and then to use the derived relation as a substitute for the formula,
both in the code and in the invariant, thus eliminating some quantifier alternations.
The user is responsible for defining the derived relations and performing the rewrites.
The system automatically generates update code for the derived relations,
and automatically checks the soundness of the rewrites.
For the generation of update code, we exploit the locality of updates,
as relations (used for defining the derived relations) are updated by inserting a single entry at a time.
We identify a class of formulas for which this automatic procedure is always possible
and use this class for verifying the Paxos protocols.

We are encouraged by the fact that the transformations needed in this step are reusable across all Paxos variants
we consider. Furthermore, the transformations maintain the simplicity and readability of both the code and the inductive invariants.

\subsection{Summary of the rest of the paper}

In \Cref{sec:background} we present the technical background on using
first-order logic to express transition systems, and on the EPR
fragment. We then develop our general methodology for EPR-based
verification in \Cref{sec:transformations}.  \Cref{sec:intro-to-paxos}
reviews the Paxos consensus algorithm, which is the basis for all
Paxos-like protocols. We present our model of the Paxos consensus
algorithm as a transition system in first-order logic in
\Cref{sec:paxos-fol}, and continue to verify it using EPR by applying
our methodology in \Cref{sec:paxos-epr}.  In \Cref{sec:multi-paxos},
we describe our verification of Multi-Paxos using EPR.  We briefly
discuss the verification of Vertical Paxos, Fast Paxos, Flexible
Paxos, and Stoppable Paxos in \Cref{sec:paxos-variants}. In
\Cref{sec:exp} we report on our implementation and experimental
evaluation. We discuss related work in \Cref{sec:related}, and
\Cref{sec:conclusion} concludes the paper.  More details about the
verification of Vertical Paxos, Fast Paxos, Flexible Paxos, and
Stoppable Paxos appear in \Cref{sec:paxos-variants-long}.
\Cref{sec:isabelle} contains a worked out comparison of the proof of
Paxos using our methodology to a proof using the Isabelle/HOL
interactive proof assistant.

\section{Background: Verification using EPR}
\label{sec:background}

In this section we present the necessary background on the formalization of transition systems using first-order logic, as well as on the EPR fragment of first-order logic.

\subsection{Transition Systems in First Order Logic}

We model transition systems using many-sorted first-order logic. We
use a vocabulary $\Sigma$ which consists of sorted constant symbols,
function symbols and relation symbols to capture the state of the
system, and formulas to capture sets of states and transitions.
Formally, given a vocabulary $\Sigma$, a \emph{state} is a first-order
structure over $\Sigma$. We sometimes use \emph{axioms} in the form of
closed first-order formulas over $\Sigma$, to restrict the set of
states to those that satisfy all the axioms.  A \emph{transition
  system} is a pair $(\init, \TR)$, where $\init$ is the \emph{initial
  condition} given by a closed formula over $\Sigma$, and $\TR$ is the
\emph{transition relation} given by a closed formula over $\Sigma
\uplus \Sigma'$ where $\Sigma$ is used to describe the source state of
the transition and $\Sigma' = \{a' \mid a \in \Sigma\}$ is used to
describe the target state. The set of initial states and the set of
transitions of the system consist of the states, respectively, pairs
of states, that satisfy $\init$, respectively, $\TR$.  We define the
set of reachable states of a transition system in the usual way.  A
\emph{safety property} is expressed by a closed formula $\safe$
over $\Sigma$. The system is \emph{safe} if all of its reachable
states satisfy $\safe$.

In the paper, we use the \emph{relational modeling language} (RML)~\cite{ivy} to express transition systems.
An RML program consists of \emph{actions}, each of which consists of a loop-free code that is executed atomically, and corresponds to a single transition.
RML commands include non-deterministic choice, sequential composition,
and updates to constant symbols, function symbols and relation symbols
(representing the system's state), where updates are expressed by
first-order formulas. In addition, conditions in RML are expressed
using {\cassume} commands. RML programs naturally translate to
formulas $(\init,\TR)$, where $\TR$ is a disjunction of the transition
relation formulas associated with each action (see \cite{ivy} for
details of the translation). As such, we will use models, programs and
transition systems interchangeably throughout the paper. We note that
RML is Turing-complete, and remains so when $\init$ and $\TR$ are
restricted to the EPR fragment.

A closed first-order formula $\Inv$ over $\Sigma$ is an
\emph{inductive invariant} for a transition system $(\init, \TR)$ if
$\init \models \Inv$ and $\Inv \land \TR \models \Inv'$, where $\Inv'$
results from substituting every symbol in $\Inv$ by its primed
version. These requirements ensure that an inductive invariant
represents a superset of the reachable states. Given a safety property
$\safe$, an inductive invariant $\Inv$ proves that the transition
system is safe if $\Inv \models \safe$. Equivalently, $\Inv$ proves
safety of $(\init,\TR)$ for $\safe$ if the following formulas are
unsatisfiable: (i) $\init \wedge \neg \Inv$, (ii) $\Inv \wedge \TR
\wedge \neg \Inv'$, and (iii) $\Inv \wedge \neg \safe$. We refer to
these formulas as the \emph{verification condition} of $\Inv$.
When $\Inv \wedge \TR \wedge \neg \Inv'$ is satisfiable, and $(s,s') \models \Inv \wedge \TR \wedge \neg \Inv'$, we say that the transition $(s,s')$ is a \emph{counterexample to induction} (CTI).

\subsection{Extended Effectively Propositional Logic (EPR)}

The effectively-propositional (EPR) fragment of first-order logic,
also known as the Bernays-Sch\"onfinkel-Ramsey class is restricted to
relational first-order formulas (i.e., formulas over a vocabulary that
contains constant symbols and relation symbols but no function
symbols) with a quantifier prefix $\exists^* \forall^*$ in prenex
normal form.  Satisfiability of EPR formulas is
decidable~\cite{LEWIS1980317}.  Moreover, formulas in this fragment
enjoy the \emph{finite model property}, meaning that a satisfiable
formula is guaranteed to have a finite model. The size of this model
is bounded by the total number of existential quantifiers and
constants in the formula. The reason for this is that given an
$\exists^* \forall^*$-formula, we can obtain an equi-satisfiable
quantifier-free formula by Skolemization, i.e., replacing the
existentially quantified variables by constants, and then
instantiating the universal quantifiers for all constants.
While EPR does not allow any function symbols nor quantifier alternation except $\exists^* \forall^*$,
it can be easily extended to allow \emph{stratified} function symbols and quantifier alternation
(as formalized below).
The extension maintains both the finite model property and the decidability of the satisfiability
problem.

\paragraph{The quantifier alternation graph}
Let $\varphi$ be a formula in negation normal form over a many-sorted signature $\vocabulary$ with a set of sorts $\srts$.
We define the \emph{quantifier alternation graph} of $\varphi$
as a directed graph %
where the set of vertices is the set of sorts, $\srts$, %
and the set of directed edges, %
called $\forall\exists$ edges, is defined as follows.
\begin{itemize}
\item {\bf Function edges:} let $\func$ be a function in $\varphi$ from sorts $\srt_1,\ldots,\srt_k$ to sort $\srt$. Then there is a $\forall\exists$ edge from $\srt_i$ to $\srt$ for every $1 \leq i \leq k$.
\item {\bf Quantifier edges:} let $\exists x: \srt$ be an existential quantifier that resides in the scope of the universal quantifiers $\forall x_1: \srt_1, \ldots, \forall x_k: \srt_k$ in $\varphi$.
Then there is a $\forall\exists$ edge from $\srt_i$ to $\srt$ for every $1 \leq i \leq k$.
\end{itemize}
Intuitively, the quantifier edges correspond to the edges that would arise as function edges if Skolemization is applied.

\paragraph{Extended EPR}
A formula $\varphi$ is \emph{stratified} if its quantifier alternation
graph %
is acyclic.  The \emph{extended EPR}
fragment consists of all stratified formulas.  This fragment maintains
the finite model property and the decidability of EPR.  The reason for
this is that, after Skolemization, the vocabulary of a stratified
formula can only generate a finite set of ground terms. This allows
complete instantiation of the universal quantifiers in the Skolemized
formula, as in EPR. In the sequel, whenever we say a formula is in
EPR, we refer to the extended EPR fragment.

\section{Methodology for Decidable Verification}
\label{sec:transformations}

In this section we explain the general methodology that we follow in our efforts to verify Paxos using decidable reasoning.
While this paper focuses on Paxos and its variant, the methodology is more general and can be useful for verifying other systems as well.

\subsection{Modeling in Uninterpreted First-Order Logic}
\label{sec:model-fol}

The first step in our verification methodology is to express the
protocol as a transition system in many-sorted uninterpreted first-order
logic. This step involves some abstraction,
since protocols usually employ
concepts that are not directly expressible in uninterpreted
first-order logic.

\subsubsection{Axiomatizing Interpreted Domains}

One of the challenges we face is modeling an interpreted domain using \emph{uninterpreted} first-order logic.
Distributed algorithms often use values from interpreted domains, the
most common example being the natural numbers. These domains are usually not precisely expressible in uninterpreted first order logic.

To express an interpreted domain, such as the natural numbers, in uninterpreted first-order logic, we add a sort that represents elements of the interpreted domain, and
uninterpreted symbols to represent the interpreted symbols (e.g. a $\leq$ binary relation). We capture \emph{part} of the intended interpretation of the symbols by introducing \emph{axioms} to the model.
The axioms are a finite set of
first-order logic formulas that are valid in the interpreted
domain. By adding them to the model, we allow the proof of
verification conditions to rely on these axioms. By using only axioms
that are valid in the interpreted domain, we guarantee that any
invariant proved for the first-order model is also valid for the
actual system.

\begin{wrapfigure}{r}{0.35\textwidth}
\vspace{-20pt}
\begin{center}
\begin{tabular}{l}
$\forall x. \; x \leq x$ \\
$\forall x,y,z. \; x \leq y \land y \leq z \to x \leq z $\\
$\forall x,y. \; x \leq y \land y \leq x \to x = y $\\
$\forall x,y. \; x \leq y \lor y \leq x$
\end{tabular}
\end{center}
    \caption{Axiomatization of total order.}
    \label{fig:total-order}
  \vspace{-10pt}
\end{wrapfigure}

One important example for axioms expressible in first-order logic is
the axiomatization of total orders. In many cases, natural numbers are used as a way to enforce a total order on a set of elements.
In such cases, %
we can add a binary relation $\leq$, along with the axioms listed in \Cref{fig:total-order}, which precisely capture the properties of a total order.

\subsubsection{Expressing Higher-Order Logic}
\label{sec:higher-order}

Another hurdle to using first-order logic is the fact that algorithms
and their invariants often use sets and functions as first class
values, e.g. by quantifying over them, sending them in a message, etc.
Consider an algorithm in which messages contain a set of nodes as one
of the message fields. Then, the set of messages sent so far (which
may be part of the state of the system) is a set of tuples, where one of the
elements in the tuples is itself a set of nodes. Similarly, messages may
contain maps, which are naturally modeled by functions (e.g., a
message may contain a map from nodes to values). In such cases, the
invariants needed to prove the algorithms will usually include
higher-order quantification.

While higher-order logic cannot be fully reduced to first-order logic,
it is well-known that we can partly express high-order concepts in
first-order logic in the following way.

\paragraph{Sets} Suppose we want to express
quantification over sets of nodes. We add a new sort called
$\snodeset$, and a binary relation $\rmember : \snode, \snodeset$. We
then use $\rmember(n,s)$ instead of $n \in s$, and express
quantification over sets of nodes as quantification over
$\snodeset$. Typically, we will need to add first-order assumptions
or axioms to correctly express the algorithm and to prove its inductive
invariant. For example, the algorithm may set $s$ to the empty set as
part of a transition. We can translate this in the transition relation
using $\forall x:\snode. \, \neg \rmember(x,s')$ (where $s'$ is the
value of $s$ after the transition).

\paragraph{Functions} Functions can be encoded as first-order elements in a similar
way. Suppose messages in the algorithm contain a map from nodes to
values. In this case, we can add a new first-order sort called
$\smap$, and a function symbol $\rapply : \smap,\snode \to
\svalue$. Then, we can use $\rapply(m,n)$ instead of $m(n)$, and
replace quantification over functions with quantification of the
first-order sort $\smap$. As before, we may need to add axioms that
capture some of the intended second-order meaning of the sort $\smap$.

\medskip
While this encoding is sound (as long as we only use axioms that are valid
in the higher-order interpretation), it cannot be made complete due to the limitation of first-order logic.
However, we did not experience this incompleteness to
be a practical hurdle for verification in first-order logic.

\subsubsection{Semi-Bounded Verification}  Given a transition system in first-order logic with a candidate
inductive invariant, it may still be
undecidable to check the resulting verification condition.
However, bounded verification is decidable, and extremely useful
for debugging the model before continuing with the efforts of unbounded verification.
Contrary to the usual practice of bounding the number of elements in each sort for bounded verification,
we use the quantifier alternation graph to determine only a \emph{subset} of the sorts to bound in order to make verification decidable.
We call this procedure \emph{semi-bounded verification},
and it follows from the observation that whenever we make a sort bounded,
we can remove its node from the quantifier alternation graph.
When the resulting graph becomes acyclic, satisfiability is decidable without bounding the sizes of the remaining sorts.

\subsection{Transformation to EPR Using Derived Relations}
\label{sec:transformation-epr}

\begin{figure}
\centering
\includegraphics[scale=0.4]{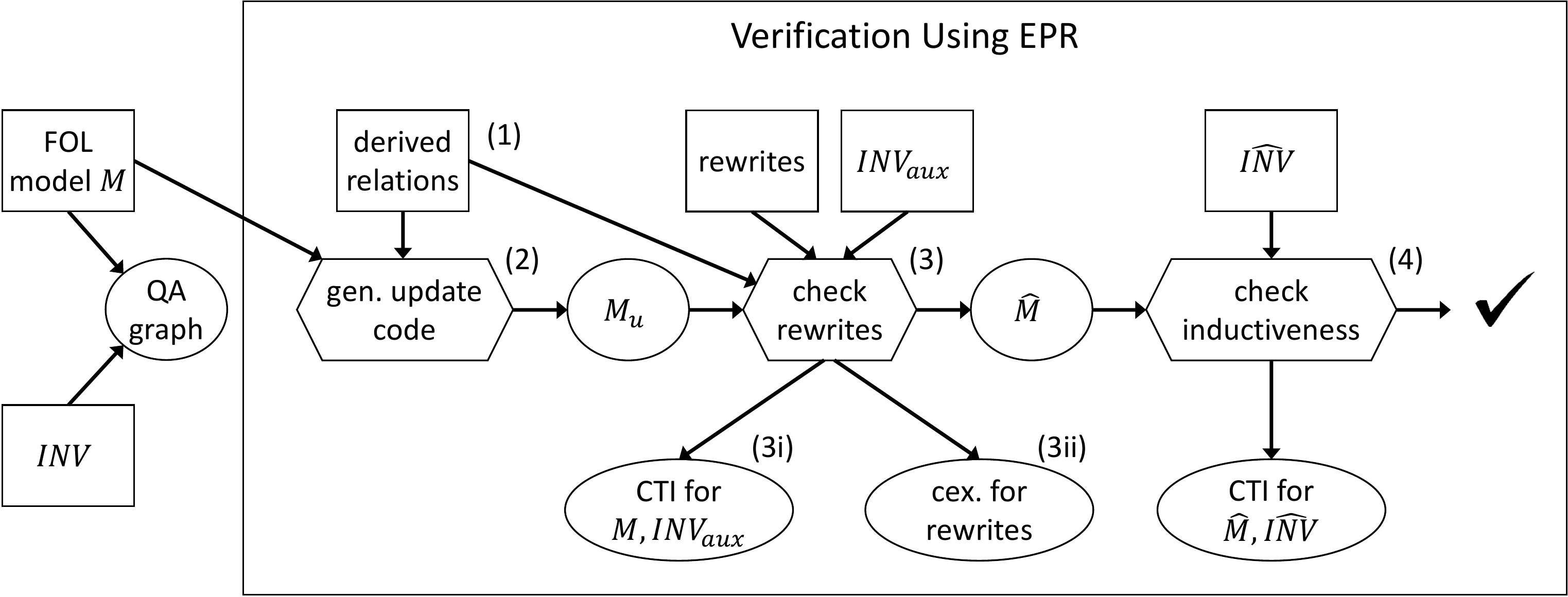}
\caption{
\label{fig:methodology-flow}%
Flow chart of the methodology for verification using EPR.
User provided inputs are depicted as rectangles,
automated procedures are depicted as hexagons, and
automatically generated outputs are depicted as ellipses.
The original first-order logic model $M = (\init,\TR)$,
and the original (first-order logic) inductive invariant $\Inv$
result in a quantifier alternation graph, which guides the process.
The transformation to EPR is carried in steps 1-4, detailed in \Cref{sec:transformation-epr}.
In step 1, the user provides definitions for derived relations.
In step 2, update code is automatically generated, resulting in $M_u = (\initu, \TRu)$.
In step 3, the user provides rewrites that use the derived relations, as well as an auxiliary inductive invariant
to prove the soundness of the rewrites. An automated procedure first checks the soundness of the rewrites
(and provides a counterexample in case they are unsound, or a counterexample to induction if the auxiliary inductive invariant is not inductive), and then outputs the transformed model $\widehat{M} = (\inithat, \TRhat)$.
In step 4, the user provides an inductive invariant $\Invhat$ that proves safety of the transformed model, and an automated check either
verifies the inductiveness or provides a counterexample to induction.
}
\end{figure}

The second step in our methodology for decidable \emph{unbounded}
verification is to transform the model expressed in first-order logic
to a model that has an inductive invariant whose verification
condition is in EPR, and is therefore decidable to check.
The methodology is manual, but following it ensures soundness of the
verification process. The key idea is to use \emph{derived relations}
to simplify the transition relation and the inductive invariant.
Derived relations extend the state of the system and are updated in
its transitions.  Derived relations are somewhat analogous to ghost
variables. However there are two key differences. First, derived
relations are typically not used to record the history of an
execution. Instead, they capture properties of the current state in a
way that facilitates verification using EPR. Second, derived relations
are not only updated in the transitions, but can also affect them.

The transformation of the model using derived relations is conducted
in steps, as detailed below. The various steps are depicted in
\Cref{fig:methodology-flow}.  The inputs provided by the user are
depicted by rectangles, while the automated procedures are depicted as
hexagons, and their outputs are depicted as ellipses.  As illustrated
by the figure, the user is guided by the quantifier alternation graph
of the verification conditions.

In the sequel, we fix a model over a vocabulary $\Sigma$ and let
$\init$ and $\TR$ denote its initial condition and transition
relation, respectively.

\paragraph{(1) Defining a derived relation}
In the first, and most creative part of the process, the user
identifies an existentially quantified formula $\psi(\bar{x})$
that will be captured by a derived relation $r(\bar{x})$. The
selection of $\psi$ is guided by the quantifier alternation graph of the verification condition,
with the purpose of eliminating cycles it contains. Quantifier
alternations in the verification condition originate both from the
model and the inductive invariant. As we shall see, using $r$ will
allow us to eliminate some quantifier alternations.
As an example for demonstrating the next steps, consider a program
defined with a binary relation $p$, and let $r(x)$ be a derived
relation capturing the formula $\psi(x) = \exists y.\, p(x,y)$.

\paragraph{(2) Tracking $\psi$ by $r$} This step automatically extends the model into a model $(\initu, \TRu)$ over vocabulary $\Sigmahat = \Sigma \cup \{r\}$ which makes the same transitions as before, but also updates $r$ to capture $\psi$.
Formally, the transformed model $(\initu, \TRu)$ over $\Sigmahat$ is obtained by adding:
\begin{inparaenum}[(i)]
\item an initial condition that initializes $r$, and
\item \emph{update code} that modifies $r$ whenever the relations mentioned
in $\psi$ are modified.
\end{inparaenum}
The initial condition and update code are automatically generated in a way that guarantees that the following formula is an invariant %
of $(\initu, \TRu)$:
\begin{equation} \label{eq:inst-inv}
\forall \bar{x}.\, r(\bar{x}) \leftrightarrow\psi(\bar{x}).
\end{equation}
We call this invariant the \emph{representation invariant} of $r$.
Our scheme for automatically obtaining $(\initu, \TRu)$ and the class
of formulas $\psi$ that it supports, are discussed in
\Cref{sec:update-code}. In our example, suppose that initially $p$ is
empty. Then, the resulting model would initialize $r$ to be empty as
well. For an action that inserts a pair $(a,b)$ to $p$, the resulting
model would contain update code that inserts $a$ to $r$.

\paragraph{(3) Rewriting the transitions using $r$}
In this step, the user exploits $r$ to eliminate quantifier
alternations in the verification condition by rewriting the system's
transitions, obtaining a model $(\inithat, \TRhat)$ defined over
$\Sigmahat$. The idea is to rewrite the transitions in a way that
ensures that the reachable states are unchanged, while eliminating
quantifier alternations. This is done by rewriting some program
conditions used in {\cassume} commands in the code (e.g., to use $r$
instead of $\psi$, but other rewrites are also possible). The
vocabulary of the model does not change further in this step, nor does
the initial condition (i.e., $\inithat = \initu$).

While the rewrites are performed by the user, we automatically check
that the effect of the modified commands on the reachable states
remains the same (under the assumption of the representation
invariant).  Suppose the user rewrites {\cassume} $\varphi$ to
{\cassume} $\varphihat$. The simplest way to ensure this has the same
effect on the reachable states is to check that the following
\emph{rewrite condition} is valid: $\varphi \leftrightarrow
\substitute{\varphihat}{r(\bar{x})}{\psi(\bar{x})}$.  This condition
guarantees that the two formulas $\varphi$ and $\varphihat$ are
equivalent in any reachable state, due to the representation
invariant. In some cases, the rewrite is such that
$\substitute{\varphihat}{r(\bar{x})}{\psi(\bar{x})}$ is syntactically
identical to $\varphi$, which makes the rewrite condition
trivial.

However, to allow greater flexibility in rewriting the code,
we allow using an EPR check to verify the rewrite condition, and also
relax the condition given above in two ways.
First, we observe that it suffices to verify the equivalence of subformulas of $\varphi$ that were modified by the rewrite. Formally, if $\varphi$ is syntactically identical to $\substitutemany{\varphihat}{\thetahat_1(\bar{y}_1)}{\theta_1(\bar{y}_1)}{\thetahat_k(\bar{y}_k)}{\theta_k(\bar{y}_k)}$, then to establish the rewrite condition, it suffices to prove that for every $1 \leq i \leq k$ the following equivalence is valid:
$\theta_i \leftrightarrow \substitute{\thetahat_i}{r(\bar{x})}{\psi(\bar{x})}$. (The case where $\varphi$ was completely modified is captured by the case where $k=1$, $\varphi = \theta_1$ and $\varphihat = \thetahat_1$.)
Second, and more importantly, recall that we are only interested in
preserving the transitions from reachable states of the system. Thus,
we allow the user to provide an auxiliary invariant $\Iaux$ (by
default $\Iaux = \true$) %
which is used to prove that the reachable transitions remain unchanged
after the transformation. Technically, this is done by automatically
checking that
\begin{enumerate}[(i)]
\item $\Iaux$ is an inductive invariant of $(\init, \TR)$, and
\item the following rewrite condition holds for every $1 \leq i \leq k$:
\begin{equation}
\label{eq:rewrite-cond}
\Iaux \wedge \smallTR \models \theta_i \leftrightarrow \substitute{\thetahat_i}{r(\bar{x})}{\psi(\bar{x})},
\end{equation}
where $\smallTR$ captures additional conditions that guard the
modified {\cassume} command ($\smallTR$ is automatically computed from the program).
\end{enumerate}
These conditions guarantee that the two formulas $\varphi$ and
$\varphihat$ are equivalent whenever the modified {\cassume} command
is executed. To ensure that these checks can be done automatically, we
require that
the corresponding formulas are in EPR.  We note that verifying $\Iaux$
for $(\init, \TR)$ can be possible in EPR even in cases where
verifying safety of $(\init, \TR)$ is not in EPR, since $\Iaux$ can be
weaker (and contain less quantifier alternations) than an invariant
that proves safety.

In our example, suppose the program contains the command
{\cassume}~$\exists y.\, p(a,y)$. Then we could rewrite it to
{\cassume}~$r(a)$.  For a more sophisticated example, suppose that the
program contains the command {\cassume}~$\exists y.\, p(a,y) \land
q(a,y)$, and suppose this command is guarded by the condition $u(a)$
(i.e., the {\cassume} only happens if $u(a)$ holds).  Suppose further
that we can verify that $\Iaux = \forall x,y.\, u(x) \land p(x,y) \to
q(x,y)$ is an invariant of the original system.  Then we could rewrite
the assume command as {\cassume}~$r(a)$ since
$\left(
\forall x,y.\, u(x) \land p(x,y) \to q(x,y)
\right)
\land u(a)
\models (\exists y.\, p(a,y) \land q(a,y)) \leftrightarrow \exists y.\, p(a,y)$.

\paragraph{(4) Providing an inductive invariant}
Finally, the user proves the safety of the transformed model $(\inithat,\TRhat)$ by providing an inductive invariant $\Invhat$ for it,
whose verification condition will be in EPR. Usually this is composed of:
\begin{inparaenum}[(i)]
\item Using $r$ in the inductive invariant as a substitute to using $\psi$. The point here is that using $\psi$ would introduce quantifier alternations, and using $r$ instead avoids them. In our example, the safety proof might require the property that $\forall x.\, u(x) \to \exists y.\, p(x,y)$, and using $r$ we can express this as $\forall x.\, u(x) \to r(x)$.
\item Letting the inductive invariant express some properties that are implied by the representation invariant.
Note that expressing the full representation invariant would typically introduce quantifier alternations that break stratification. However, some properties implied by it
may still be expressible while keeping the verification condition in EPR.
In our example, we may add $\forall x,y.\, p(x,y) \to r(x)$ to the inductive invariant.
Note that adding $\forall x.\, r(x) \to \exists y.\, p(x,y)$ to the inductive invariant would make the verification condition outside of EPR.
\end{inparaenum}

Given $(\inithat,\TRhat)$ and $\Invhat$, we can now automatically derive the verification conditions in EPR and check that they hold.
The following theorem summarizes the soundness of the approach:

\begin{theorem}[Soundness]
\label{thm:soundness}
Let $(\init,\TR)$ be a model over vocabulary $\Sigma$, and $\safe$ be a safety property over $\Sigma$.
If $(\inithat,\TRhat)$ is a model obtained by the above procedure,
and $\Invhat$ is an inductive invariant for it such that $\Invhat \models \safe$,
then $\safe$ holds in all reachable states of $(\init,\TR)$.
\end{theorem}
\begin{proof}[Proof Sketch]
Let $B = \{(\s,\shat) \mid \s \in \Reach \wedge \shat \in \Reachhat
\wedge \shat|_{\Sigma} = \s \}$, where $\Reach$ and $\Reachhat$ denote
the reachable states of $(\init,\TR)$ and $(\inithat,\TRhat)$
respectively, and $\shat|_\Sigma$ denotes the projection of a state
$\shat$ (defined over $\Sigmahat$) to $\Sigma$. Steps 2 and 3 of the
transformation above ensure that $B$ is a bisimulation relation between
$(\init,\TR)$ and $(\inithat,\TRhat)$, i.e., every transition possible
in the reachable states of one of these systems has a corresponding
transition in the other. This ensures that $(\inithat,\TRhat)$ has the
same reachable states as $(\init,\TR)$, up to the addition of relation
$r$. Therefore, any safety property expressed over $\Sigma$ which is
verified to hold in $(\inithat,\TRhat)$ also holds in $(\init,\TR)$.
\end{proof}

As shown in the proof of \Cref{thm:soundness}, the transformed model
$(\inithat,\TRhat)$ is bisimilar to the original model. While this ensures that both are equivalent w.r.t. to the safety property, note that we check safety by checking inductiveness of a candidate invariant. Unlike safety, inductiveness is not necessarily preserved by the transformation. Namely, given a candidate inductive invariant $\Invhat$ which is not inductive for $(\inithat,\TRhat)$, the counterexample to induction cannot in general be transformed to the original model, as it might depend on the derived relations and the rewritten {\cassume} commands. An example of this phenomenon appears in \Cref{sec:joined-rounds}.

\paragraph{Using the methodology}
Our description above explains a final successful verification using
the proposed methodology.  As always, obtaining this involves a series
of attempts, where in each attempt the user provides the verification
inputs, and gets a counterexample. Each counterexample guides the user
to modify the verification inputs, until eventually verification is
achieved. As depicted in \Cref{fig:methodology-flow}, with the EPR
verification methodology, the user provides 5 inputs, and could obtain
3 kinds of counterexamples.  The inputs are the model, the derived
relations, the rewrites, the auxiliary invariant for proving the
soundness of the rewrites, and finally the inductive invariant for the
resulting model. The possible counterexamples are either a
counterexample to inductiveness (CTI) for the auxiliary invariant and
the original model, or a counterexample to the soundness of the
rewrite itself, or a CTI for the inductive invariant of the transformed
model. After obtaining any of the 3 kinds of counterexamples, the user
can modify any one of the 5 inputs. For example, a CTI for the
inductive invariant of the transformed model may be eliminated by changing the
inductive invariant itself, but it may also be overcome by an additional
rewrite, which in turn requires an auxiliary invariant for its
soundness proof. Indeed, we shall see an example of this in
\Cref{sec:joined-rounds}.

The task of managing the inter-dependence between the 5 verification
inputs may seem daunting, and indeed it requires some expertise and
creativity from the user. This is expected, since the inputs from the
user reduce the undecidable problem of safety verification to
decidable EPR checks. This burden on the user is eased by the fact
that for every input, the user always obtains an answer from the
system, either in the form of successful verification, or in the form
of a \emph{finite} counterexample, which is displayed graphically and
guides the user towards the solution. Furthermore, our experience
shows that most of the creative effort is reusable across similar
protocols.  In the verification of all the variants of Paxos we
consider in this work, we use the same two derived relations and very
similar rewrites (as explained in \Cref{sec:paxos-epr,sec:paxos-variants}).

\paragraph{Incompleteness of EPR verification}
While the transformation using a given set of derived relations and
rewrites results in a bisimilar transition system, the methodology for
EPR verification is not complete.  This is expected, as there can be
no complete proof system for safety in a formalism that is
Turing-complete. For the EPR verification methodology, the
incompleteness can arise from several sources. It may happen that
after applying the transformation, the resulting transition system,
while safe, cannot be verified with an inductive invariant that
results in EPR verification conditions. Another potential source for
incompleteness is our requirement that the rewrites should also be
verified in EPR. It can be the case that a certain (sound) rewrite
leads to a system that can be verified using EPR, but the soundness of
the rewrite itself cannot be verified using EPR.  Another potential
source of incompleteness can be the inability to express sufficiently
powerful axioms about the underlying domain. We note that the three
mentioned issues interact with each other, as it may be the case that
a certain axiom is expressible in first-order logic, but it happens to
introduce a quantifier alternation cycle, when considered together
with either the inductive invariant or the verification conditions for
the rewrites.

We consider developing a proof-theoretic understanding of which systems
can and cannot be verified using EPR to be an intriguing direction for
future investigation. We are encouraged by the fact that in practice,
the proposed methodology has proven itself to be powerful enough to verify Paxos and
its variants considered in this work.

\paragraph{Multiple derived relations}
For simplicity, the description above considered a single derived
relation. In practice, we usually add multiple derived relations,
where each one captures a different formula. The methodology remains
the same, and each derived relation allows us to transform the model
and eliminate more quantifier alternations, until the resulting model
can be verified in EPR. In this case, the resulting inductive
invariant may include properties implied by the representation
invariants of several relations and relate them directly. For example,
suppose we add the following derived relations: $r_1(x)$ defined by
$\psi_1(x) = \exists y.\, p(x,y)$, and $r_2(x)$ defined by $\psi_2(x)
= \exists y,z.\, p(x,y) \land p(y,z)$.  Then, the inductive invariant
may include the property: $\forall x.\, r_2(x) \to r_1(x)$.

\paragraph{Overapproximating the reachable states} Our methodology ensures that the transformed model is bisimilar to the original model.
It is possible to generalize our methodology and only require that the modified model simulates the original model, which maintains soundness.
This may allow more flexibility both in the update code and in the manual rewrites performed by the user.

\subsection{Automatic Generation of Update Code}
\label{sec:update-code}

In this subsection, we describe a rather na\"{\i}ve scheme for automatic
generation of initial condition and update code for derived relations,
which suffices for verification of the Paxos variants considered in
this paper. We refer the reader to, e.g.,
\cite{PaigeK82,TOPLAS:RepsSL10}, for more advanced techniques for
generation of update code for derived relations.

We limit the formula $\psi(\bar{x})$ which defines a derived relation
$r(\bar{x})$ to have the following form:
\begin{equation*}
\psi(x_1,\ldots,x_n) = \exists y_1,\ldots,y_m. \; \varphi(x_1,\ldots,x_n,y_1,\ldots,y_m) \land p(x_{i_1},\ldots,x_{i_k},y_1,\ldots,y_m)
\end{equation*}
where $\varphi$ is a quantifier-free formula, $p$ is a relation
symbol and $x_{i_j} \in \{x_1,\ldots,x_n\}$ for every $1 \leq j \leq k$. Note that $p$ occurs positively, and that it depends on
\emph{some} (possibly none) of the variables $x_i$ and \emph{all} of the variables $y_i$.
Our scheme further requires that the relations appearing in $\varphi$
are never modified, and that $p$ is initially empty and only updated
by inserting a single tuple at a time\footnote{These restrictions can
  be relaxed, e.g., to support removal of a single tuple or
  addition of multiple tuples. However, such updates were not needed
  for verification of the protocols considered in this paper, so for simplicity of the presentation we do not handle them.}.

Since $p$ is initially empty, the initial condition for $r(\bar{x})$
is that it is empty as well, i.e.:
\begin{equation*}
  \forall x_1,\ldots,x_n . \; \neg r(x_1,\ldots,x_n).
\end{equation*}
The only updates allowed for $p$ are insertions of a single tuple by a command of the form:\\
\begin{equation*}
p(x_{i_1},\ldots,x_{i_k},y_1,\ldots,y_m) \text{ := } p(x_{i_1},\ldots,x_{i_k},y_1,\ldots,y_m) \lor \left( \bigwedge_{j=1}^k x_{i_j} = a_j \land \bigwedge_{j=1}^m y_j = b_j \right).
\end{equation*}
For such an update, we generate the following update for $r(\bar{x})$:
\begin{equation*}
r(x_1,\ldots,x_n) \text{ := } r(x_1,\ldots,x_n) \lor \left( \varphi(x_1,\ldots,x_n,b_1,\ldots,b_m) \land \bigwedge_{j=1}^k x_{i_j} = a_j \right).
\end{equation*}
Notice that the update code translates to a purely universally
quantified formula, since $\varphi$ is quantifier-free, so it does not
introduce any quantifier alternations.

\begin{lemma}
The above scheme results in a model which maintains the representation
invariant: $\forall \bar{x}.\, r(\bar{x}) \leftrightarrow\psi(\bar{x})$.
\end{lemma}
\begin{proof}[Proof Sketch]
The representation invariant is an inductive invariant of the
resulting model. Initiation is trivial, since both $p$ and $r$ are
initially empty. Consecution follows from the following, which is valid in first-order logic:
$
\left(\forall \bar{x}.  r(\bar{x}) \leftrightarrow \psi(\bar{x}) \right)\land
\left(\forall \bar{w}, \bar{y}.  p'(\bar{w},\bar{y}) \leftrightarrow
p(\bar{w},\bar{y}) \lor (\bar{w}=\bar{a} \land \bar{y}=\bar{b}) \right) \land
\left(\forall \bar{x}.  r'(\bar{x}) \leftrightarrow
r(\bar{x}) \lor \left( \varphi(\bar{x},\bar{b}) \land \bigwedge_{j=1}^k x_{i_j} = a_j \right) \right)
$ $
\models
\left(\forall \bar{x}.  r'(\bar{x}) \leftrightarrow \psi'(\bar{x}) \right)
$.
\end{proof}

\section{Introduction to Paxos} \label{sec:intro-to-paxos}

A popular approach for implementing distributed systems is
state-machine replication (SMR) \cite{schneider_implementing_1990},
where a (virtual) centralized sequential state machine is replicated
across many nodes (processors), providing fault-tolerance and exposing to its clients the familiar semantics of a centralized state machine.
SMR can be thought of as repeatedly agreeing on a command to be executed next by the state machine, where each time agreement is obtained by solving a \emph{consensus} problem.
In the consensus problem, a set of nodes each propose a value and then reach agreement on a single proposal.

The Paxos family of protocols is widely used in practice for implementing SMR. Its core is the Paxos consensus algorithm~\cite{paxos,paxos-made-simple}.
A Paxos-based SMR implementation executes a sequence of Paxos consensus instances, with various optimizations.
The rest of this section explains the Paxos consensus algorithm (whose verification in EPR we discuss in~\Cref{sec:paxos-fol,sec:paxos-epr}).
We return to the broader context of SMR in~\Cref{sec:paxos-variants}.

\paragraph{Setting}

We consider a fixed set of nodes, which operate asynchronously and communicate by message passing,
where every node can send a message to every node. Messages can be
lost, duplicated, and reordered, but they are never corrupted.
Nodes can fail by stopping, but otherwise faithfully execute their algorithm.
A stop failure of a node can be captured by a loss of all messages to and from this node.
Nodes must solve the consensus problem: each node has a value to propose and all nodes must eventually decide on a unique value among the proposals.

\paragraph{Paxos consensus algorithm}

We assume that nodes in the Paxos consensus algorithm can all propose values, vote for values, and learn about decisions.
The algorithm operates in a sequence of numbered rounds in which nodes can participate.
At any given time, different nodes may operate in different rounds, and a node stops participating in a round once it started participating in a higher round.
Each round is associated with a single node that is the \emph{owner} of that round.
This association from rounds to nodes is static and known to all nodes.

Every round represents a possibility for its owner to propose a value to the other nodes and get it decided on by having a \emph{quorum} of the nodes vote for it in the round.
Quorums are sets of nodes such that any two quorums intersect (e.g., sets consisting of a strict majority of the nodes).
To avoid the algorithm being blocked by the stop failure of a node which made a proposal, any node can start one of its rounds and make a new proposal in it at any time (in particular, when other rounds are still active) by executing the following two phases:
\begin{description}
\item[phase 1.]
The owner $p$ of round $r$ starts the round by communicating with the other nodes
to have a majority of them \emph{join} round $r$, and to determine which values are \emph{choosable} in lower rounds than $r$, i.e., values that might have or can still be decided in rounds lower than $r$.

\item [phase 2.]
If a value $v$ is choosable in $r'< r$, in order not to contradict a potential decision in $r'$, node $p$ proposes $v$ in round $r$.
If no value is choosable in any $r' < r$, then $p$ proposes a value of its choice in round $r$.
If a majority of nodes \emph{vote} in round $r$ for $p$'s proposal, then it becomes \emph{decided}.
\end{description}
Note that it is possible for different values to be proposed in
different rounds, and also for several decisions to be made in different rounds.
Safety is guaranteed by the fact that (by definition of choosable) a value can be decided in a round $r' < r$  only if it is choosable in $r'$, and that if a value $v$ is choosable in round $r' < r$, then a node proposing in $r$ will only propose $v$.
The latter relies on the property that choosable values from prior rounds cannot be missed.
Next, we describe in more detail what messages the nodes exchange and how a node makes sure not to miss any choosable value from prior rounds.

\emph{Phase 1a}: The owner $p$ of round $r$ sends a ``start-round'' message, requesting all nodes to join round $r$.

\emph{Phase 1b}: Upon receiving a request to join round $r$, a
node will only join if it has not yet joined a higher round. If it
agrees to join, it will respond with a ``join-acknowledgment'' message
that will also contain its maximal vote so far, i.e., its vote in the
highest round prior to $r$, or $\bot$ if no such vote exists.  By
sending the join-acknowledgment message, the node promises that it
will not join or vote in any round smaller than $r$.

\emph{Phase 2a}: After $p$ receives join-acknowledgment messages from a quorum of the nodes, it proposes
a value $v$ for round $r$ by sending a ``propose'' message to all
nodes. Node $p$ selects the value $v$ by taking the maximal
vote reported by the nodes in their join-acknowledgment messages,
i.e., the value that was voted for in the highest round prior to $r$ by any of the
nodes whose join-acknowledgment messages formed the quorum.
As we will see, only this value can be choosable in any $r'<r$ out of all proposals from lower rounds.
If all of these nodes report they have not voted in any prior round,
then $p$ may propose any value.

\emph{Phase 2b}: Upon receiving a propose message proposing value $v$
for round $r$, a node will ignore it if it already joined a round
higher than $r$, and otherwise it will vote for it, by sending a vote
message to all nodes. Whenever a quorum of nodes vote for a
value in some round, this value is considered to be decided. Nodes
learn this by observing the vote messages.

Note that a node can successfully start a new round or get a value decided only if at least one quorum of nodes is responsive. When quorums are taken to be sets consisting of a strict majority of the nodes, this means Paxos tolerates the failure of at most $\floor{n/2}$ nodes, where $n$ is the total number of nodes.
Moreover, Paxos may be caught in a live-lock if nodes keep starting new rounds before any value has a chance to be decided on.

\section{Paxos in First-Order Logic}
\label{sec:paxos-fol}

The first step of our verification methodology is to model the Paxos
consensus algorithm as a transition system in many-sorted first-order
logic over uninterpreted domains. This section explains our model, listed in
\Cref{fig:paxos-fol}, as well as its safety proof via an inductive invariant. %

\lstset{ %
  breakatwhitespace=false,         %
  keywordstyle=\bf,       %
  language=C,                 %
  otherkeywords={module,individual,init,action,returns,assert,assume,instantiate,isolate,mixin,before,relation,function,sort,variable,axiom,then,constant,let,*,local},           %
  numbers=left,                    %
  numbersep=5pt,                   %
  numberstyle=\tiny,               %
  rulecolor=\color{black},         %
  tabsize=8,	                   %
   columns=fullflexible,
}

\begin{figure}
\begin{minipage}{.56\textwidth}
\begin{lstlisting}[
    %
    basicstyle=\scriptsize,%
    keepspaces=true,
    numbers=left,
    %
    xleftmargin=2em,
    numberstyle=\tiny,
    emph={
      %
      %
      %
      %
      %
      %
    },
    emphstyle={\bfseries},
    mathescape=true,
  ]
sort $\snode$, $\squorum$, $\sround$, $\svalue$

relation $\leq$ : $\sround,\sround$
axiom total_order($\leq$)
constant $\none$ : $\sround$

relation $\rmember$ : $\snode,\squorum$
axiom $\forall q_1,q_2 : \squorum. \,  \exists n:\snode. \, \rmember(n,q_1) \land \rmember(n, q_2) \label{line:axiom-quorum}$

relation $\ronea$ : $\sround$
relation $\ronebmaxvote$ : $\snode,\sround,\sround,\svalue$
relation $\rproposal$ : $\sround,\svalue$
relation $\rvote$ : $\snode,\sround,\svalue$
relation $\rdecision$ : $\snode,\sround,\svalue$

init $\forall r. \; \neg\ronea(r)$
init $\forall n,r_1,r_2,v. \; \neg\ronebmaxvote(n,r_1,r_2,v)$
init $\forall r,v. \; \neg\rproposal(r,v)$
init $\forall n,r,v. \; \neg\rvote(n,r,v)$
init $\forall n,r,v. \; \neg\rdecision(n,r,v)$

action $\asendonea(\text{r} : \sround)$ {
  assume $\text{r} \neq \none$
  $\ronea(\text{r})$ := true
}
action $\ajoinround(\text{n} : \snode ,\, \text{r} : \sround)$ {
  assume $\text{r} \neq \none$
  assume $\ronea(\text{r})$
  assume $\neg \exists r',r'',v. \; r' > \text{r} \land \ronebmaxvote(\text{n},r',r'',v) \label{line:join-round-if}$
  # find maximal round in which n voted, and the corresponding vote.
  # maxr = $\bot$ and v is arbitrary when n never voted.
  local maxr, v := max $\{ (r',v') \mid \rvote(\text{n},r',v') \land r' < \text{r}\} \label{line:join-round-max}$
  $\ronebmaxvote(\text{n},\text{r},\text{maxr},\text{v})$ := true $\label{line:join-round-send}$
}
action $\apropose(\text{r} : \sround ,\, \text{q} : \squorum)$ {
  assume $\text{r} \neq \none$
  assume $\forall v. \; \neg\rproposal(\text{r},v) \label{line:propose-assume-unique}$
  # 1b from quorum q
  assume $\forall n. \; \rmember(n, \text{q}) \to \exists r',v. \;  \ronebmaxvote(n,\text{r},r',v)$  $\label{line:propose-assume-ae}$
  # find the maximal round in which a node in the quorum reported
  # voting, and the corresponding vote.
  # v is arbitrary $\text{if}$ the nodes reported not voting.
  local maxr, v := max $\{ (r',v') \mid \exists n. \; \rmember(n, \text{q})$
                                             $\land \ronebmaxvote(n,\text{r},r',v') \land r' \neq \none \} \label{line:propose-max}$
  $\rproposal(\text{r}, \text{v})$ := true $\label{line:propose-send}$  # propose value v
}
action $\acastvote(\text{n} : \snode ,\, \text{r} : \sround ,\, \text{v} : \svalue)$ {
  assume $\text{r} \neq \none$
  assume $\rproposal(\text{r}, \text{v})$
  assume $\neg \exists r',r'',v. \;  r' > \text{r} \land \ronebmaxvote(\text{n},r',r'',v) \label{line:vote-round}$
  $\rvote(\text{n}, \text{r}, \text{v})$ := true
}
action $\adecide(\text{n} : \snode, \text{r} : \sround ,\, \text{v} : \svalue ,\, \text{q} : \squorum)$ {
  assume $\text{r} \neq \none$
  # 2b from quorum q
   assume $\forall n. \; \rmember(n, \text{q}) \to \rvote(n, \text{r}, \text{v})$
  $\rdecision(\text{n}, \text{r}, \text{v})$ := true
}
\end{lstlisting}
\captionof{figure}{Model of Paxos consensus algorithm as a\\transition system in many-sorted first-order logic.}
  \label{fig:paxos-fol}
\end{minipage}%
\begin{minipage}{.43\textwidth}
\begin{minipage}{\textwidth}
\centering
\vspace{-1cm}
\includegraphics[scale=0.3]{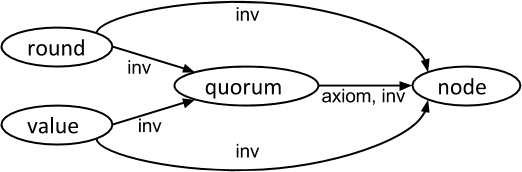}
\captionof{figure}{Quantifier alternation graph for EPR model of Paxos.
The graph is obtained for the model of \Cref{fig:paxos-epr}, and the inductive invariant
given by the conjunction of \cref{eq:safety,eq:proposal-unique,eq:vote-proposal,eq:quorum-of-decision,eq:one_b_1,eq:one_b_2,eq:one_b_3,eq:choosable-epr,eq:proj-msg,eq:proj-left-round}. The $\forall\exists$ edges come from:
(i) the quorum axiom (\Cref{fig:paxos-fol} \cref{line:axiom-quorum}) - edge from $\squorum$ to $\snode$;
(ii) \cref{eq:quorum-of-decision} -- edges from $\sround$ and $\svalue$ to $\squorum$, and an edge from  $\squorum$ to $\snode$ (from the negation of the inductive invariant in the VC); and
(iii) \cref{eq:choosable-epr} -- edges from $\sround$, $\svalue$ and $\squorum$ to $\snode$. }
  \label{fig:paxos-epr-graph}
\end{minipage}

\vspace{0.1cm}

\begin{minipage}{\textwidth}
\begin{lstlisting}[
    %
    basicstyle=\scriptsize,%
    keepspaces=true,
    numbers=left,
    %
    xleftmargin=2em,
    numberstyle=\tiny,
    emph={
      %
      %
      %
      %
      %
      %
    },
    emphstyle={\bfseries},
    mathescape=true,
    escapeinside={(*@}{@*)},
  ]
relation $\rleftround$ : $\snode,\sround$
relation $\ronebproj$ : $\snode,\sround$

init $\forall n,r. \; \neg \rleftround(n,r) \label{line:epr-init-left-round}$
init $\forall n,r. \; \neg \ronebproj(n,r) \label{line:epr-init-joined-round}$

action $\ajoinround(\text{n} : \snode ,\, \text{r} : \sround)$ {
  assume $\text{r} \neq \none$
  assume $\ronea(\text{r})$
  assume $\neg \rleftround(n,r) \label{line:epr-join-round-if}$ # rewritten
  local maxr, v := max $\{ (r',v') \mid$
                              $\rvote(\text{n},r',v') \land r' < \text{r}\} \label{line:epr-join-round-max}$
  $\ronebmaxvote(\text{n},\text{r},\text{maxr},\text{v})$ := true $\label{line:epr-join-round-send}$
  $\text{\# generated update code for derived relations:}$
  $\rleftround(\text{n},R)$ := $\rleftround(\text{n},R) \lor R < \text{r} \label{line:epr-join-round-left}$
  $\ronebproj(\text{n},\text{r})$ := true $\label{line:epr-join-round-proj}$
}
action $\apropose(\text{r} : \sround ,\, \text{q} : \squorum)$ {
  assume $\text{r} \neq \none$
  assume $\forall v. \; \neg\rproposal(\text{r},v) \label{line:epr-propose-assume-unique}$
  # rewritten to aviod quantifier alternation
  assume $\forall n. \; \rmember(n, \text{q}) \to \ronebproj(n,\text{r}) \label{line:epr-propose-assume}$
  # rewritten to use $\rvote$ instead of $\ronebmaxvote$
  local maxr, v := max $\{ (r',v') \mid \exists n. \; \rmember(n, \text{q})$
                              $\land \rvote(n,r',v') \land r' < \text{r} \} \label{line:epr-propose-max}$
  $\rproposal(\text{r}, \text{v})$ := true
}
action $\acastvote(\text{n} : \snode ,\, \text{r} : \sround ,\, \text{v} : \svalue)$ {
  assume $\text{r} \neq \none$
  assume $\rproposal(\text{r}, \text{v})$
  assume $\neg \rleftround(\text{n},\text{r}) \label{line:epr-vote-round}$ # rewritten
  $\rvote(\text{n}, \text{r}, \text{v})$ := true
}
\end{lstlisting}
\captionof{figure}{Changes to the Paxos model that allow verification in EPR.
Declarations that appear in \Cref{fig:paxos-fol} are omitted, as well as the $\adecide$ action which is left unmodified.}
  \label{fig:paxos-epr}
\end{minipage}%
\end{minipage}
\end{figure}

\subsection{Model of the Protocol}

Our model of Paxos involves some abstraction.
Since each round $r$ has a unique owner that will exclusively propose in $r$, we abstract away the owner node and treat the round itself as the
proposer. We also abstract the mechanism by which nodes receive
the values up for proposal, and allow them to propose arbitrary values.

Additional abstractions are needed as some aspects of the protocol cannot
be fully expressed in uninterpreted first-order logic.
One such aspect is the fact that round numbers are integers, as arithmetic cannot be fully captured in first-order logic.
Another aspect which must be abstracted is the use of sets of nodes as quantification over sets is also beyond first-order logic.
We model these aspects %
according to the principles of \Cref{sec:model-fol}:

\paragraph{Sorts and Axioms}
We use the following four uninterpreted sorts:
(i)~$\snode$ - to represent nodes of the system, (ii)~$\svalue$ - to represent the values subject to the consensus algorithm,
(iii)~$\sround$ - to model the rounds of Paxos, and (iv)~$\squorum$ - to model sets of nodes with pairwise intersection in a first-order abstraction. %
While nodes and values are naturally uninterpreted, the rounds and the
quorums are uninterpreted representations of interpreted concepts:
integers and sets of nodes that intersect pairwise, respectively.
We express some features that come from the desired interpretation using
relations and axioms.

For rounds, we include a binary relation $\leq$, and axiomatize it to
be a total order (\Cref{fig:total-order}).
Our model also includes a constant $\none$ of sort $\sround$, which represents %
a special round that is not
considered an actual round of the protocol, and instead serves as a
special value used in the join-acknowledgment (1b) message when a node
has not yet voted for any value. Accordingly, any action assumes that
the round it involves is not $\none$.

The quorum sort is used to represent sets of nodes that contain strictly more
than half of the nodes.  As explained in \Cref{sec:model-fol}, we
introduce a membership relation between nodes and quorums. An
important property for Paxos is that any two quorums intersect. We
capture this with an axiom in first-order logic (\Cref{fig:paxos-fol}
\cref{line:axiom-quorum}).

\paragraph{State}
The state of the protocol consists of the set of messages the nodes
have sent. We represent these using relations, where each tuple in a
relation corresponds to a single message. The relations $\ronea$,
$\ronebmaxvote$, $\rproposal$, $\rvote$ correspond to the 1a, 1b, 2a, 2b
phases of the algorithm, respectively.  In modeling the algorithm, we assume all
messages are sent to all nodes, so the relations do not contain
destination fields. Note that recording messages via relations (i.e., sets) is an abstraction of the network but it is consistent with the messaging model we assume, in which messages may
be lost, duplicated, and reordered.
The $\rdecision$ relation captures the decisions
learned by the nodes.

\paragraph{Actions} The different atomic steps taken by the nodes in the protocol are modeled using actions.  The {\asendonea} action models phase 1a of the protocol, sending a start round message to all nodes.
The {\ajoinround} action models the receipt of a start round message and the
transmission of a join-acknowledgment (1b) message. The {\apropose}
action models the receipt of join-acknowledgment (1b) messages
from a quorum of nodes, and the transmission of a propose (2a) message
which proposes a value for a round. The {\acastvote}
action models the receipt of a propose (2a) message by a node,
and voting for a value by sending a vote (2b) message. Finally, the {\adecide}
action models learning a decision by node $n$, when it is voted
for by a quorum of nodes.

In these actions, sending a message is expressed by inserting the corresponding tuple to the corresponding relation. Different conditions (e.g., not joining a round if already joined  higher round, properly reporting the previous votes, or appropriately selecting the proposed value) are expressed using assume statements.
To prepare a join-acknowledgment message  in {\ajoinround}, as well as to propose a value in {\apropose}, a node needs to compute the maximal vote (performed by it or reported to it, respectively).
This is done by a max operation (\cref{line:join-round-max} and \cref{line:propose-max}) which operates with respect to the order on rounds, and
returns the round $\none$ and an arbitrary value in case the set is
empty. The $r, v := \text{max} \left\{(r',v') \mid \varphi(r',v')
\right\}$ operation is syntactic sugar for an assume of the following formula:
\begin{equation} \label{eq:maxrv}
(r = \none \land \forall r', v'. \; \neg\varphi(r',v')) \lor
(r \neq \none \land \varphi(r, v) \land \forall r', v'. \; \varphi(r',v') \to r' \leq r).
\end{equation}
Note that if $\varphi$ is a purely existentially quantified formula,
then \cref{eq:maxrv} is alternation-free.

\subsection{Inductive Invariant}
\label{sec:paxos-fol-inv}

The key safety property we wish to verify about Paxos is that only a
single value can be decided (it can be decided at multiple rounds, as
long as it is the same value). This is expressed by the following
universally quantified formula:
\begin{small}
\begin{equation} \label{eq:safety}
\forall n_1,n_2:\snode, r_1,r_2:\sround, v_1,v_2:\svalue. \; \rdecision(n_1,r_1,v_1) \land \rdecision(n_2,r_2,v_2) \to v_1 = v_2
\end{equation}
\end{small}

While the safety property holds in all the reachable states of the protocol, it is not inductive.
That is, assuming that it holds is not sufficient to prove that it still holds
after an action is taken. For example, consider a state $s$ in which
$\rdecision(n_1,r_1,v_1)$ holds and there is a quorum $q$ of nodes
such that, for every node $n$ in $q$, $\rvote(n,r_2,v_2)$ holds, with
$v_2\neq v_1$. Note that the safety property holds in $s$. However, a
$\adecide$ action introduces a transition from state $s$ to a state $s'$ in which both
$\rdecision(n_2,r_1,v_1)$ and $\rdecision(n_2,r_2,v_2)$ hold,
violating the safety property.
This counterexample to induction does not indicate a violation of safety,
but it indicates that the safety property needs to be strengthened in order to obtain an inductive invariant.
We now describe such an inductive invariant.

Our inductive invariant contains, in addition to the safety property,
the following rather simple statements that are maintained by the
protocol and are required for inductiveness:
\begin{small}
\begin{align}
& \forall r:\sround, v_1,v_2:\svalue. \; \rproposal(r,v_1) \land \rproposal(r,v_2) \to v_1 = v_2 \label{eq:proposal-unique} \\
& \forall n:\snode, r:\sround, v:\svalue. \; \rvote(n,r,v) \to \rproposal(r,v) \label{eq:vote-proposal}\\
&\!\begin{multlined}
\forall r:\sround, v:\svalue. \\ (\exists n:\snode. \; \rdecision(n,r,v)) \to
\exists q:\squorum. \forall n:\snode. \;\rmember(n, q) \to \rvote(n,r,v)
\end{multlined}
\label{eq:quorum-of-decision}
\end{align}
\end{small}
\Cref{eq:proposal-unique} states that there is a unique proposal per round.
\Cref{eq:vote-proposal} states that a vote for $v$ in round $r$ is cast only when a proposal for $v$ has been made in round $r$.
\Cref{eq:quorum-of-decision} states that a decision for $v$ is made in round $r$ only if a quorum of nodes have voted for $v$ in round $r$.
In addition, the inductive invariant restricts the join-acknowledgment
messages so that they faithfully represent the maximal vote (up to the
joined round), or $\none$ if there are no votes so far, and also
asserts that there are no actual votes at round $\none$:
\begin{small}
\begin{align}
& \forall n:\snode,\, r,r':\sround,\, v,v':\svalue. \; \ronebmaxvote(n,r,\none,v) \land r' < r \to \neg \rvote(n,r',v') \label{eq:one_b_1} \\
& \forall n:\snode,\, r,r':\sround,\, v:\svalue. \; \ronebmaxvote(n,r,r',v) \land r' \neq \none \to  r' < r \land \rvote(n,r',v) \label{eq:one_b_2}\\
&\!\begin{multlined}
\forall n:\snode,\, r,r',r'':\sround,\, v,v':\svalue. \\ \ronebmaxvote(n,r,r',v) \land r' \neq \none \land r' < r'' < r \to
 \neg \rvote(n,r'',v')
\end{multlined}
\label{eq:one_b_3}
\\
& \forall n:\snode, v:\svalue. \; \neg \rvote(n,\none,v) \label{eq:votenone}
\end{align}
\end{small}
The properties stated so far are rather straightforward, and are
usually not even mentioned in paper proofs or explanations of the
protocol. The key to the correctness argument of the protocol is the
observation that when the owner of round $r_2$ proposes a value in $r_2$, it cannot miss any value that is choosable at a lower round:
whenever a value $v_2$ is proposed at round $r_2$, then in all rounds $r_1$ prior to $r_2$, no other value $v_1 \neq v_2$ is choosable.
The property that no $v_1 \neq v_2$ is choosable at $r_1$ is captured in the inductive invariant by the
requirement that in any quorum of nodes, there must be at least one
node that has already left round $r_1$ (i.e., joined a higher round),
and did not vote for $v_1$ at $r_1$ (and hence will also not vote for it in the future). Formally, this is:
\begin{small}
\begin{equation} \label{eq:choosable}
\begin{split}
& \forall r_1,r_2:\sround,v_1,v_2:\svalue, q:\squorum. \;
\rproposal(r_2,v_2) \land r_1 < r_2 \land v_1 \neq v_2 \to \\
& \; \exists n:\snode, r', r'':\sround, v:\svalue. \; \rmember(n,q) \land \neg \rvote(n,r_1,v_1) \land
r' > r_1 \land \ronebmaxvote(n,r',r'',v)
\end{split}
\end{equation}
\end{small}

The fact that this property is maintained by the protocol is obtained
by the proposal mechanism and the interaction between phase 1 and
phase 2 (see \Cref{sec:isabelle} for a detailed explanation). %

\Cref{eq:safety,eq:proposal-unique,eq:vote-proposal,eq:quorum-of-decision,eq:one_b_1,eq:one_b_2,eq:one_b_3,eq:votenone,eq:choosable}
define an inductive invariant that proves the safety of the Paxos
model of \Cref{fig:paxos-fol}. However, the verification condition for
this inductive invariant contains cyclic quantifier alternations, and
is therefore outside of EPR. We now review the quantifier alternations
in the verification condition, which originate both from the model and
from the inductive invariant.

In the model, the axiomatization of quorums (\Cref{fig:paxos-fol}
\cref{line:axiom-quorum}) introduces a $\forall\exists$-edge from
quorum to node. In addition, the assumption in
the $\apropose$ action that join-acknowledgment messages were received from a quorum of
nodes (\cref{line:propose-assume-ae})
introduces $\forall\exists$-edges from node to round and from node to
value.

In the inductive invariant, only \cref{eq:quorum-of-decision,eq:choosable} include quantifier alternations (the rest are universally quantified).
\Cref{eq:quorum-of-decision} has quantifier structure $\forall \sround,
\svalue \, \exists \squorum \, \forall \snode$\footnote{The local existential quantifier in $(\exists n:\snode. \rdecision(n,r,v))$ does not affect the quantifier alternation graph.}. Note that the inductive
invariant appears both positively and negatively in the verification
condition, so \cref{eq:quorum-of-decision} adds $\forall\exists$-edges
from round to quorum and from value to quorum (from the positive
occurrence), as well as an edge from quorum to node (from the negative
occurrence). While the latter coincides with the edge that comes from
the quorum axiomatization (\cref{line:axiom-quorum}), the former edges
closes a cycle in the quantifier alternation graph.
\Cref{eq:choosable} has quantifier prefix $\forall \sround, \svalue,
\squorum \, \exists \snode, \sround, \svalue$. Thus, it introduces 9
edges in the quantifier alternation graph, including self-loops at
round and value.
In conclusion, while the presented model in first-order logic has an
inductive invariant in first-order logic, the resulting verification condition is outside of EPR.

\section{Paxos in EPR}
\label{sec:paxos-epr}

The quantifier alternation graph of the model of Paxos described in
\Cref{sec:paxos-fol} contains cycles. To obtain a safety proof of
Paxos in EPR, we apply the methodology described in
\Cref{sec:transformations} to transform this model in a way that
eliminates the cycles from the quantifier alternation graph. The
resulting changes to the model are presented in \Cref{fig:paxos-epr},
and the rest of this section explains them step by step.

\subsection{Derived Relation for Left Rounds}

We start by addressing the quantifier alternation that appears in
\cref{eq:choosable} as part of the inductive invariant.  We observe
that the following existentially quantified formula appears both as a subformula
there, and in the conditions of the $\ajoinround$ and the $\acastvote$
actions (\Cref{fig:paxos-fol} \cref{line:join-round-if,line:vote-round}):
\begin{equation*} %
\psi_1(n,r) = \exists r',r'':\sround,v:\svalue. \; r' > r \land \ronebmaxvote(n,r',r'',v)
\end{equation*}
This formula captures the fact that node $n$ has joined a round higher
than $r$, which means it promises to never participate in round $r$ in
any way, i.e., it will neither join nor vote in round $r$.
We add a derived relation called $\rleftround$ to capture $\psi_1$, so that
$\rleftround(n,r)$ captures the fact that node $n$ has left round $r$.
The formula $\psi_1$ is in the class of formulas handled by the scheme
described in \Cref{sec:update-code}, and thus we obtain the initial condition
and update code for $\rleftround$. The result appears in \Cref{fig:paxos-epr} \cref{line:epr-init-left-round,line:epr-join-round-left}.

\paragraph{Rewriting (steps 3+4)}
Using the $\rleftround$ relation, we rewrite the conditions of the
$\ajoinround$ and $\acastvote$ actions (\Cref{fig:paxos-epr}
\cref{line:epr-join-round-if,line:epr-vote-round}). These rewrites are
trivially sound as explained in \Cref{sec:transformation-epr} (with a trivial rewrite condition).
We also rewrite \cref{eq:choosable} as:
\begin{equation} \label{eq:choosable-epr}
\begin{split}
\forall & r_1,r_2:\sround,v_1,v_2:\svalue, q:\squorum. \;
\rproposal(r_2,v_2) \land r_1 < r_2 \land v_1 \neq v_2 \to \\
&\exists n:\snode, \; \rmember(n,q) \land \neg \rvote(n,r_1,v_1) \land \rleftround(n,r_1)
\end{split}
\end{equation}
\Cref{eq:choosable-epr} contains less quantifer alternations than
\cref{eq:choosable}, and it will be part of the inductive invariant
that will eventually be used to prove safety.

\subsection{Derived Relation for Joined Rounds} \label{sec:joined-rounds}

After the previous transformation, the verification condition is still
not stratified. The reason is the combination of
\cref{eq:choosable-epr} and the condition of the $\apropose$ action (\Cref{fig:paxos-fol} \cref{line:propose-assume-ae}):
\[
\forall n : \snode. \; \rmember(n, q) \to \exists r' : \sround, \, v:\svalue. \;  \ronebmaxvote(n,r,r',v)
.
\]
While each of these introduces quantifier alternations that are
stratified when viewed separately, together they form
cycles. \Cref{eq:choosable-epr} introduces $\forall\exists$-edges from
$\sround$ and $\svalue$ to $\snode$, while the $\apropose$ condition
introduces edges from $\snode$ to $\sround$ and $\svalue$. The
$\apropose$ condition expresses the fact that every node in the quorum
q has joined round r by sending a join-acknowledgment (1b) message to
round r. However, because the join-acknowledgment message contains two
more fields (representing the node's maximal vote so far), the
condition existentially quantifies over them. To eliminate this
existential quantification and remove the cycles, we add a derived
relation, called $\ronebproj$, that captures the projection of
$\ronebmaxvote$ over its first two components, given by the formula:
\begin{equation*}%
\psi_2(n,r) = \exists r':\sround,v:\svalue. \; \ronebmaxvote(n,r,r',v)
\end{equation*}
This binary relation over nodes and rounds records the sending of
join-acknowledgment messages, ignoring the maximal vote so far
reported in the message. Thus, $\ronebproj(n,r)$ captures the fact
that node $n$ has agreed to join round $r$.
The formula $\psi_2$ is in the class of formulas handled by the scheme
of \Cref{sec:update-code}, and thus we obtain the initial condition
and update code for $\ronebproj$, as it appears in
\Cref{fig:paxos-epr}
\cref{line:epr-init-joined-round,line:epr-join-round-proj}.

\paragraph{Rewriting (steps 3+4): first attempt}
We rewrite the condition of the $\apropose$ action to use $\ronebproj$ instead of $\psi_2$. (The rewrite condition is again trivial, ensuring soundness.)
The result appears in \Cref{fig:paxos-epr}
\cref{line:epr-propose-assume},
and is purely universally quantified.
When considering the transformed model, and the candidate inductive
invariant given by the conjunction of
\cref{eq:safety,eq:proposal-unique,eq:vote-proposal,eq:quorum-of-decision,eq:one_b_1,eq:one_b_2,eq:one_b_3,eq:votenone,eq:choosable-epr},
the resulting quantifier alternation graph is acyclic. %
This means that the verification
condition is in EPR and hence decidable to check. However, it turns
out that this candidate invariant is not inductive,
and the check yields a counterexample to induction.

\begin{figure}
\centering
\includegraphics[scale=0.333]{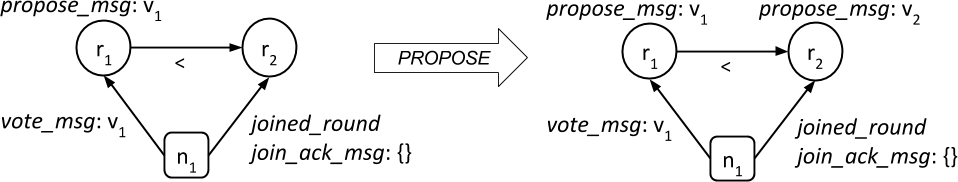}
\caption{
\label{fig:paxos-epr-cti}%
Counterexample to induction of EPR model of Paxos after the first
attempt. The counterexample contains one node $n_1$ (depicted as square),
two rounds $r_1 < r_2$ (depicted as circles), two values $v_1, v_2$, and a single quorum (that contains $n_1$).
The figure displays the the
$\ronebmaxvote,\,\ronebproj,\,\rproposal,\,\rvote$ relations, as well
as the $<$ relation (derived from $\leq$).  The action occurring in the
counterexample is $\apropose$, where an arbitrary value is proposed
for $r_2$, even though node $n_1$ voted for $v_1$ in $r_1$. This
erroneous behavior occurs due to the fact that the representation
invariant of $\ronebproj$ is violated in the pre-state of the counterexample:
$\ronebproj(n_1,r_2)$ holds, eventhough there is no corresponding $\ronebmaxvote$ entry.
This allows a $\apropose$ action to erroneously propose $v_2$ for $r_2$, in spite of the fact that $v_1$ is choosable at $r_1$. 
}
\end{figure}

The counterexample is depicted in \Cref{fig:paxos-epr-cti}.
The counterexample shows a $\apropose$ action that leads to a violation of \cref{eq:choosable-epr}.
The example contains a single node $n_1$ (which forms a quorum), that has voted for value $v_1$ in round $r_1$,
and yet a different value $v_2$ is proposed for a later round $r_2$ (based on the quorum composed only of $n_1$) which leads to a violation of \cref{eq:choosable-epr}.
The $\apropose$ action is enabled since $\ronebproj(n_1,r_2)$ holds.
However, an arbitrary value is proposed since $\ronebmaxvote$ is empty.
The root cause for the counterexample is that the inductive invariant does not capture the connection between
$\ronebproj$ and $\ronebmaxvote$, so it allows a state in which a node
$n_1$ has joined round $r_2$ according to $\ronebproj$ (i.e.,
$\ronebproj(n_1,r_2)$ holds), but it has not joined it according to
$\ronebmaxvote$ (i.e., $\exists r',v. \;
\ronebmaxvote(n_1,r_2,r',v)$ does not hold).
Note that the counterexample is spurious,
in the sense that it does not represent a reachable state. However,
for a proof by an inductive invariant, we must eliminate this
counterexample nonetheless.

\paragraph{Rewriting (steps 3+4): second attempt}
One obvious way to eliminate the counterexample discussed above is to
add the representation invariant of $\ronebproj$ to the inductive
invariant. However, this will result in a cyclic quantifier
alternation, causing the verification condition to be outside of
EPR. Instead, we will eliminate this counterexample by rewriting the
code of the $\apropose$ action, relying on an auxiliary invariant to
verify the rewrite, as explained in \Cref{sec:transformation-epr}.  We
observe that the mismatch between $\ronebproj$ and $\ronebmaxvote$ is
only problematic in this example because node $n_1$ voted in $r_1 <
r_2$. While the condition of the $\apropose$ action is supposed to
ensure that the max operation considers past votes of all nodes in the
quorum, such a scenario where the $\ronebproj$ is inconsistent with
$\ronebmaxvote$ makes it possible for the $\apropose$ action to
overlook past votes, which is the case in this counterexample.  Our
remedy is therefore to rewrite the max operation (which is implemented
by an assume command, as explained before) to consider the votes directly
by referring to the \emph{vote} messages instead of the
\emph{join-acknowledgment} messages that report them. We first
formally state the rewrite and then justify its correctness.

As before, we rewrite the condition of the $\apropose$ action to use
$\ronebproj$ (\Cref{fig:paxos-epr} \cref{line:epr-propose-assume}). In
addition, we rewrite the max operation in \Cref{fig:paxos-fol} \cref{line:propose-max},
i.e., max $\{ (r',v') \mid \varphi_1(r',v') \}$ where
\[
\varphi_1(r',v') = \exists n. \, \rmember(n, q) \land \ronebmaxvote(n,r,r',v') \land r' \neq \none
\]
to the max operation in \Cref{fig:paxos-epr} \cref{line:epr-propose-max},
i.e., max $\{ (r',v') \mid \varphi_2(r',v') \}$ where
\[
\varphi_2(r',v') = \exists n.\, \rmember(n, q) \land \rvote(n,r',v')\land r' < r
\]

The key to the correctness of this change is that a
join-acknowledgment message from node $n$ to round $r$ contains its
maximal vote prior to round $r$, and once the node sent this message,
it will never vote in rounds smaller than $r$. Therefore, while the
original $\apropose$ action considers the maximum over votes reflected
by join-acknowledgment messages from a quorum, looking at the actual
\emph{votes} from the quorum in rounds prior to $r$ (as captured by
the $\rvote$ relation) yields the same maximum.

Formally, we establish the rewrite condition of step 3 given by
\cref{eq:rewrite-cond} using an auxiliary invariant $\Iaux$, defined
as the conjunction of
\cref{eq:proposal-unique,eq:vote-proposal,eq:one_b_1,eq:one_b_2,eq:one_b_3,eq:votenone}.
This invariant captures the connection between $\ronebmaxvote$ and
$\rvote$ explained above. The invariant $\Iaux$ is inductive for the
original model, and its verification condition is in EPR (the
resulting quantifier alternation graph is acyclic). Second, we prove
that under the assumption of $\Iaux$ and the condition $\forall
n.\,\rmember(n, q) \to \exists r',v. \ronebmaxvote(n,r,r',v)$
(\Cref{fig:paxos-fol} \cref{line:propose-assume-ae}), the operation
$\max \{ (r',v') \mid \varphi_1(r',v') \}$ is equivalent to the
operation $\max \{ (r',v') \mid \varphi_2(r',v') \}$ (recall that both
translate to assume's according to \cref{eq:maxrv}). This check
is also in EPR.
In conclusion, we are able to establish the rewrite condition using
two EPR checks: one for proving $\Iaux$, and one for proving
\cref{eq:rewrite-cond}.

\paragraph{Inductive Invariant}
After the above rewrite, the conjunction of
\cref{eq:safety,eq:proposal-unique,eq:vote-proposal,eq:quorum-of-decision,eq:one_b_1,eq:one_b_2,eq:one_b_3,eq:votenone,eq:choosable-epr}
is still not an inductive invariant, due to a counterexample to
induction in which a node has joined a higher round according to
$\ronebproj$, but has not left a lower round according to
$\rleftround$. As before, the counterexample is inconsistent with the
representation invariants. However, this time the counterexample (and
another similar one) can be eliminated by strengthening the inductive
invariant with the following facts, which are implied by the
representation invariants of $\ronebproj$ and $\rleftround$:
\begin{align}
& \forall n:\snode, \, r_1,r_2:\sround. \; r_1 < r_2 \land \ronebproj(n,r_2) \to \rleftround(n,r_1) \label{eq:proj-left-round} \\
& \forall n:\snode, \, r,r':\sround, \, v:\svalue. \; \ronebmaxvote(n,r,r',v) \to  \ronebproj(n,r) \label{eq:proj-msg}
\end{align}
Both are purely universally quantified and therefore do not affect
the quantifier alternation graph.

Finally, the invariant given by the conjunction of
\cref{eq:safety,eq:proposal-unique,eq:vote-proposal,eq:quorum-of-decision,eq:one_b_1,eq:one_b_2,eq:one_b_3,eq:votenone,eq:choosable-epr,eq:proj-msg,eq:proj-left-round}
is indeed an inductive invariant for the transformed model. \Cref{fig:paxos-epr-graph} depicts the quantifier alternation graph of the resulting verification condition. This graph is acyclic, and so the invariant can be verified in EPR.
This invariant proves the safety of the
transformed model (\Cref{fig:paxos-epr}).
Using \Cref{thm:soundness}, the safety of the original Paxos model (\Cref{fig:paxos-fol}) follows.

\section{Multi-Paxos}
\label{sec:multi-paxos}

In this section we describe our verification of
Multi-Paxos\footnote{The version described here is slightly improved
  compared to \cite{oopsla17-epr}, by using the unary $\ractive$
  relation instead of the binary $\ravailable$ relation.}. Multi-Paxos
is an implementation of state-machine replication (SMR): nodes run a
sequence of instances of the Paxos algorithm, where the
$i^{\text{th}}$ instance is used to decide on the $i^{\text{th}}$
command in the sequence of commands executed by the machine.  For
efficiency, when starting a round, a node uses a single message to
simultaneously do so in all instances.  In response, each node sends a
join-acknowledgment message that reports its maximal vote in each
instance; this message has finite size as there are only finitely many
instances in which the node ever voted.  The key advantage over using
one isolated incarnation of Paxos per instance is that when a unique
node takes on the responsibility of starting a round and proposing
commands, only phase 2 of Paxos has to be executed for every proposal.

Below we provide a description of the EPR verification of Multi-Paxos.
The main change compared to the Paxos consensus algorithm is in the
modeling of the algorithm in first-order logic, where we use the
technique of \Cref{sec:higher-order} to model higher-order
concepts. The transformation to EPR is essentially the same as in
\Cref{sec:paxos-epr}, using the same derived relations and rewrites.

\subsection{Model of the Protocol}

\lstset{ %
  breakatwhitespace=false,         %
  keywordstyle=\bf,       %
  language=C,                 %
  otherkeywords={module,individual,init,action,returns,assert,assume,instantiate,isolate,mixin,before,relation,function,sort,variable,axiom,then,constant,let,*,local},           %
  numbers=left,                    %
  numbersep=5pt,                   %
  numberstyle=\tiny,               %
  rulecolor=\color{black},         %
  tabsize=8,	                   %
   columns=fullflexible,
}

\begin{figure}
\begin{lstlisting}[
    %
    basicstyle=\scriptsize,%
    keepspaces=true,
    numbers=left,
    %
    xleftmargin=2em,
    numberstyle=\tiny,
    emph={
      %
      %
      %
      %
      %
      %
    },
    emphstyle={\bfseries},
    mathescape=true,
  ]
sort $\snode$
sort $\squorum$
sort $\sround$
sort $\svalue$
sort $\sinst$
sort $\svotemap$

relation $\leq$ : $\sround,\sround$
axiom total_order($\leq$)
constant $\none$ : $\sround$
relation $\rmember$ : $\snode,\squorum$
axiom $\forall q_1,q_2 : \squorum. \;  \exists n:\snode. \;  \rmember(n,q_1) \land \rmember(n, q_2) \label{line:multi-fol-axiom-quorum}$

relation $\ronea$ : $\sround$
relation $\ronebmaxvote$ : $\snode,\sround,\svotemap$
relation $\rproposal$ : $\sinst,\sround,\svalue$
relation $\ractive$ : $\sround$
relation $\rvote$ : $\snode,\sinst,\sround,\svalue$
relation $\rdecision$ : $\snode,\sinst,\sround,\svalue$
function $\rroundof$ : $\svotemap,\sinst \to \sround$
function $\rvalueof$ : $\svotemap,\sinst \to \svalue$

init $\forall r. \; \neg\ronea(r)$
init $\forall n,r,m. \; \neg\ronebmaxvote(n,r,m)$
init $\forall i,r,v. \; \neg\rproposal(i,r,v)$
init $\forall r. \; \neg\ractive(r)$
init $\forall n,i,r,v. \; \neg\rvote(n,i,r,v)$
init $\forall r,v. \; \neg\rdecision(r,v)$

action $\asendonea(\text{r} : \sround)$ { assume $\text{r} \neq \none$ ; $\ronea(\text{r})$ := true } $\label{line:multi-fol-sendonea}$
action $\ajoinround(\text{n} : \snode ,\, \text{r} : \sround)$ { $\label{line:multi-fol-joinround}$
    assume $\text{r} \neq \none \land \ronea(\text{r}) \land \neg \exists r',m. \; r' > \text{r} \land \ronebmaxvote(\text{n},r',m) \label{line:multi-fol-join-round-if}$
    local m : $\svotemap$ := *
    assume $\forall i. \; (\rroundof(\text{m},i),\rvalueof(\text{m},i)) = \text{max }\left\{ (r',v') \mid \rvote(\text{n},i,r',v') \land r' < \text{r}\right\} \label{line:multi-fol-join-round-max}$
    $\ronebmaxvote(\text{n},\text{r},\text{m})$ := true $\label{line:multi-fol-join-round-send}$
}
action $\ainstate(\text{r} : \sround ,\, \text{q} : \squorum)$ {
    assume $\text{r} \neq \none$
    assume $\neg\ractive(\text{r}) \label{line:multi-fol-instate-assume-once}$
    assume $\forall n. \; \rmember(n, \text{q}) \to \exists m. \;  \ronebmaxvote(n,\text{r},m) \label{line:multi-fol-instate-assume-ae}$
    local m : $\svotemap$ := *
    assume $\forall i. \; (\rroundof(\text{m},i),\rvalueof(\text{m},i)) = \text{max }\left\{ (r',v') \mid \exists n,m'. \; \rmember(n, \text{q}) \land \ronebmaxvote(n,\text{r},m') \, \land \right. \label{line:multi-fol-instate-max}$
                                                                                                                $\left.  r' = \rroundof(m',\text{i}) \land v' = \rvalueof(m',\text{i}) \land r' \neq \none \, \right\}$
    $\ractive(\text{r})$ := true $\label{line:multi-fol-instate-active}$
    $\rproposal(I,\text{r},V)$ := $\rproposal(I,\text{r},V) \lor (\rroundof(\text{m},I) \neq \none \land V = \rvalueof(\text{m},I))$ $\label{line:multi-fol-instate-propose}$
}
action $\aproposenew(\text{r} : \sround ,\, \text{i} : \sinst ,\, \text{v} : \svalue)$ {
    assume $\text{r} \neq \none$
    assume $\ractive(\text{r}) \land \forall v. \; \neg\rproposal(\text{i},\text{r},v) \label{line:multi-fol-propose-available}$
    $\rproposal(\text{r}, \text{v})$ := true $\label{line:multi-fol-propose-send}$
}
action $\acastvote(\text{n} : \snode ,\, \text{i} : \sinst ,\, \text{r} : \sround ,\, \text{v} : \svalue)$ { $\label{line:multi-fol-castvote}$
    assume $\text{r} \neq \none \land \rproposal(\text{i}, \text{r}, \text{v}) \land \neg \exists r',m. \;  r' > \text{r} \land \ronebmaxvote(\text{n},r',m) \label{line:multi-fol-vote-round}$
    $\rvote(\text{n}, \text{i}, \text{r}, \text{v})$ := true
}
action $\adecide(\text{n} : \snode, \text{i} : \sinst ,\, \text{r} : \sround ,\, \text{v} : \svalue ,\, \text{q} : \squorum)$ { $\label{line:multi-fol-decide}$
    assume $\text{r} \neq \none \land \forall n'. \; \rmember(n', \text{q}) \to \rvote(n', \text{i}, \text{r}, \text{v})$
    $\rdecision(\text{n}, \text{i}, \text{r}, \text{v})$ := true
}
\end{lstlisting}
\caption{%
\label{fig:multi-paxos-fol}%
Model of Multi-Paxos as a transition system in many-sorted first-order
logic.%
}
\end{figure}

Multi-Paxos uses the same message types as the basic Paxos algorithm,
but when a node joins a round, it sends a join-acknowledgment message
(1a) with its maximal vote for each instance (there will only be a
finite set of instances for which it actually voted). Upon receipt of
join-acknowledgment messages from a quorum of nodes that join round
$r$, the owner of the round $r$ determines the instances for which it
is obliged to propose a value (because it may be choosable in a prior
round) and proposes values accordingly for these instances.  Other
instances are considered \emph{available}, and in these instances the
round owner can propose any value. Next, the owner of round $r$ can
propose commands in available instances (in response to client
requests, which are abstracted in our model). This means that the
$\apropose$ action of Paxos is split into two actions in Multi-Paxos:
one that processes the join-acknowledgment messages from a quorum and
another that proposes new values.

Our model of Multi-Paxos in first-order logic appears in
\Cref{fig:multi-paxos-fol}. We explain the key differences compared to the Paxos model from \Cref{fig:paxos-fol}.

\paragraph{State}
We extend the vocabulary of the Paxos model with two new sorts:
$\sinst$ and $\svotemap$. The $\sinst$ sort represents instances, and
the $\rproposal$, $\rvote$ and $\rdecision$ relations are extended to
include an instance in each tuple. In practice, instances may be
natural numbers that give the order in which commands of the state
machine must be executed by each replica. However, we are only
interested in proving consistency (i.e., that decisions are unique per
instance), and the consistency proof does not depend on the instances
being ordered. Therefore, our model does not include a total order
over instances.

The $\svotemap$ sort models a map from instances to
$(\sround,\svalue)$ pairs, which are passed in the join-acknowledgment
messages (captured by the relation
$\ronebmaxvote:\snode,\sround,\svotemap$).  We use the encoding
explained in \Cref{sec:higher-order}, and add two functions,
$\rroundof : \svotemap,\sinst \to \sround$ and $\rvalueof :
\svotemap,\sinst \to \svalue$, that allow access to the content of a
$\svotemap$.

We also add a new relation $\ractive : \sround$, to support the
splitting of the $\apropose$ action into two actions. A round is
considered $\ractive$ once the round owner has received a quorum of
join-acknowledgment messages, and then it can start proposing new
values for available instances.

\paragraph{Actions}
The $\asendonea$ action is identical to Paxos, and the $\acastvote$
and $\adecide$ actions are identical except they are now parameterized
by an $\sinst$. The $\ajoinround$ action is identical in principle,
except it now must obtain and deliver a $\svotemap$ that maps each
instance to the maximal vote of the node (and $\none$ for instances in
which the node did not vote). To express the computation of the
maximal vote of every instance i,
we use an assume statement in
\cref{line:multi-fol-join-round-max} of the form $\forall i. \;
(\rroundof(m,i),\rvalueof(m,i)) = \text{max} \left\{ (r',v') \mid
\varphi(i,r',v') \right\}$ which follows a
non-deterministic choice of $\text{m}:\svotemap$.
The assume statement is realized by the following formula in the transition relation (which is an adaptation of \cref{eq:maxrv} to account for multiple instances):
\begin{equation} \label{eq:maxvotemap}
\begin{split}
\forall i. \; &
(\rroundof(m,i) = \none \land \forall r', v'. \; \neg\varphi(i,r',v')) \; \lor \\
& (\rroundof(m,i) \neq \none \land \varphi(i,\rroundof(m,i), \rvalueof(m,i)) \land \forall r', v'. \; \varphi(i,r',v') \to r' \leq \rroundof(m,i))
\end{split}
\end{equation}
Note that for \cref{line:multi-fol-join-round-max}, $\varphi$ is quantifier free,
so \cref{eq:maxvotemap} is purely universally quantified.

The most notable difference in the actions is that, in Multi-Paxos,
the $\apropose$ action is split into two actions: $\ainstate$ and
$\aproposenew$. $\ainstate$ processes the join-acknowledgment
messages from a quorum, and $\aproposenew$ proposes new values, %
modeling the fact that only phase 2 is
repeated for every instance.

The $\ainstate$ action takes place when the owner of a round received
join-acknowledgment messages from a quorum of nodes. The owner then
finds the maximal vote reported in the messages for each instance,
which is done in \cref{line:multi-fol-instate-max}. This is realized
using \cref{eq:maxvotemap}. Observe that $\varphi$ here contains
existential quantifiers over $\snode$ and $\svotemap$. This introduces
quantifier alternation, which results in $\forall\exists$ edges from
$\sinst$ to both $\snode$ and $\svotemap$. Fortunately, these edges do
not create cycles in the quantifier alternation graph. Next, in
\cref{line:multi-fol-instate-active}, the round is marked as active,
and in \cref{line:multi-fol-instate-propose}, all obligatory proposals
are made.

The $\aproposenew$ action models the proposal of a new value in an
available instance, after the $\ainstate$ action took place and the
round is made active. This action occurs due to client requests, which
are abstracted in our model. Therefore, the only preconditions of this
action in our model are that the $\ainstate$ action took place
(captured by the $\ractive$ relation), and that the instance is still
available, i.e., that the round owner has not proposed any other value
for it. These preconditions are expressed in
\cref{line:multi-fol-propose-available}. In practice, a round owner
will choose the next available instance (according to some total
order). However, since this is not necessarily for correctness, our
model completely abstracts the total order over instances, and allows
a new value to be proposed in any available instance.

\subsection{Inductive Invariant}
\label{sec:multi-paxos-inv}

The safety property we wish to prove for Multi-Paxos is that each
Paxos instance is safe. Formally:
\begin{equation} \label{eq:multi-safety}
\begin{split}
& \forall i:\sinst, n_1,n_2:\snode, r_1,r_2:\sround, v_1,v_2:\svalue. \; \\
& \qquad \qquad \qquad \qquad \rdecision(n_1,i,r_1,v_1) \land \rdecision(n_2,i,r_2,v_2) \to v_1 = v_2
\end{split}
\end{equation}
\Cref{eq:multi-safety} generalizes \cref{eq:safety} by universally
quantifying over all instances. The inductive invariant that proves
safety contains similarly generalized versions of
\cref{eq:proposal-unique,eq:vote-proposal,eq:quorum-of-decision,eq:one_b_1,eq:one_b_2,eq:one_b_3,eq:votenone},
where the $\ronebmaxvote$ relation is adjusted to contain a
$\svotemap$ element instead of a $\sround,\svalue$ pair as the message
content. In addition, the inductive invariant asserts that proposals
are only made for active rounds:
\begin{equation} \label{eq:active-proposal}
\forall i:\sinst,r:\sround,v:\svalue. \;
\rproposal(i,r,v) \to \ractive(r)
\end{equation}

Finally, the inductive invariant for Multi-Paxos contains a
generalized version of \cref{eq:choosable}, that states that if round
$r$ is active, any value which is not proposed in it is also not
choosable for lower rounds:
\begin{equation} \label{eq:multi-choosable}
\begin{split}
\forall & i:\sinst,r_1,r_2:\sround,v:\svalue, q:\squorum. \;
\ractive(r_2) \land r_1 < r_2  \land \neg \rproposal(i,r_2,v) \to \\
&\exists n:\snode, r':\sround, m:\svotemap. \; \rmember(n,q) \land \neg \rvote(n,r_1,v) \land
r' > r_1 \land \\ & \qquad \qquad \qquad \qquad \qquad \qquad \qquad \qquad \qquad \qquad \qquad \qquad \qquad \ronebmaxvote(n,r',m)
\end{split}
\end{equation}
This generalizes \cref{eq:choosable}, in which the condition is that
another value is proposed, since proposals are unique by
\cref{eq:proposal-unique}.
\Cref{eq:multi-safety,eq:active-proposal,eq:multi-choosable} together
with
\cref{eq:proposal-unique,eq:vote-proposal,eq:quorum-of-decision,eq:one_b_1,eq:one_b_2,eq:one_b_3,eq:votenone}
(generalized by universally quantifying over all instances) provide an
inductive invariant for the model of \Cref{fig:multi-paxos-fol}.

\subsection{Transformation to EPR}
\label{sec:multi-paxos-epr}

As with the Paxos model of \Cref{fig:paxos-fol}, the resulting
verification condition for the Multi-Paxos model is outside of EPR,
and must be transformed to allow EPR verification. The required
transformations are essentially identical to the Paxos model (the new sorts do not appear in any cycles in the quantifier alternation graph), where we
define the $\rleftround$ relation by:
\begin{equation*}
\psi_1(n,r) = \exists r',m. \; r' > r \land \ronebmaxvote(n,r',m)
\end{equation*}
And the $\ronebproj$ relation by:
\begin{equation*}
\psi_2(n,r) = \exists m. \; \ronebmaxvote(n,r,m)
\end{equation*}
The $\rleftround$ relation is used to rewrite
\cref{eq:multi-choosable} in precisely the same way it was used to
rewrite \cref{eq:choosable}. In addition, we rewrite the max operation
in \cref{line:multi-fol-instate-max} of \Cref{fig:multi-paxos-fol} to
use $\rvote$ instead of $\ronebmaxvote$, which is exactly the same
change that was required to verify the Paxos model. Thus, the
transformations to EPR are in this case completely reusable, and allow
EPR verification of Multi-Paxos.

\section{Paxos Variants}
\label{sec:paxos-variants}

In this section we briefly describe our verification in EPR of several
variants of Paxos. More elaborate explanations are provided in
\Cref{sec:paxos-variants-long}. In all cases, the transformations to
EPR of \Cref{sec:paxos-epr} were employed (with slight modifications),
demonstrating the reusability of the derived relations and rewrites
across different Paxos variants.

\subsection{Vertical Paxos}
\label{sec:vertical-paxos-short}

Vertical Paxos~\cite{lamport_vertical_2009} is a variant of Paxos
whose set of participating nodes and quorums (called the
configuration) can be changed dynamically by an external
reconfiguration master.  By using reconfiguration to replace failed
nodes, Vertical Paxos makes Paxos reliable in the long-term.  The
reconfiguration master dynamically assigns configurations to rounds,
which means that each round uses a different set of quorums.  A
significant algorithmic complication is that old configurations must
be eventually retired in practice.  This is achieved by having the
nodes inform the master when a round $r$ becomes \emph{complete},
meaning that $r$ holds all the necessary information about choosable
values in lower rounds.  The master tracks the highest complete round
and passes it on to each new configuration to indicate that lower
rounds need not be accessed. Rounds below the highest complete round
can then be retired safely.

We model configurations in first-order logic by introducing a new sort
$\sconfig$ that represents a set of quorums, with a suitable member
relation called $\rquorumin$.  Moreover, we change the axiomatization
of quorums to only require that quorums of the same configuration
intersect:
\begin{align*}
&
\forall c:\sconfig, q_1,q_2:\squorum. \;
\rquorumin(q_1,c) \land \rquorumin(q_2,c) \rightarrow
\\
&
\qquad \qquad \qquad \qquad
\exists n:\snode. \rmember(n,q_1) \land \rmember(n, q_2)
\end{align*}
To model the complete round associated to each configuration, we introduce a
function symbol $\rcompleteof: \sconfig \to \sround$.  This function symbol
introduces additional cycles to the quantifier alternation graph. The
transformation to EPR replaces the $\rcompleteof$ function by a derived relation
defined by the formula $\rcompleteof(c) = r$.  With this derived relation, we can
rewrite the model and invariant so that the function no longer appears in the
verification condition, and hence the cycles that it introduced are eliminated.
Other than that, the transformation to EPR uses the same derived relations and
rewrites of \Cref{sec:paxos-epr} (in fact, it only requires the $\rleftround$
derived relation). A full description appears in \Cref{sec:paxos-vertical}.

\subsection{Fast Paxos}

Fast Paxos~\cite{lamport_fast_2006} is a variant of Multi-Paxos that
improves its latency. The key idea is to mark some of the rounds as
\emph{fast} and allow any node to directly propose values in these
rounds without going through the round owner. As a result, multiple
values can be proposed, as well as voted for, in the same (fast) round.
In order to maintain consistency, Fast Paxos uses two kinds of quorums: \emph{classic}
quorums and \emph{fast} quorums. The quorums have the property that
any two classic quorums intersect, and any classic quorum and
\emph{two} fast quorums intersect.
Now, in a propose action receiving join-acknowledgment messages from the
classic quorum $q$ with a maximal vote reported in a fast round $maxr$,
multiple different values may be reported by nodes in $q$ in $maxr$.  To
determine which one may be choosable in $maxr$, a node will check whether there
exists a \emph{fast} quorum $f$ such that all the nodes in $q\cap f$ reported
voting $v$ in $maxr$.  If yes, by the intersection property of quorums, only
this value $v$ may be choosable in $maxr$, and hence must be proposed.

We model fast quorums with an
additional sort $\sfquorum$ (and relation $\rfmember$), and axiomatize its intersection property as:
\[
\forall q:\squorum, f_1,f_2:\sfquorum. \; \exists n:\snode. \;
\rmember(n,q) \land \rfmember(n, f_1) \land \rfmember(n, f_2)
\]
The rest of the details of the model appear in \Cref{sec:fast-paxos}.
The transformation to EPR is similar to \Cref{sec:paxos-epr}.  This
includes the rewrite of the new condition for proposing a value.
Interestingly, in this case, the verification of the latter rewrite is
not in EPR when considering the formulas as a whole, but it is in EPR
when we consider only the subformulas that change (see
\Cref{sec:fast-paxos}).

\subsection{Flexible Paxos}

Flexible Paxos~\cite{howard_flexible_2016} extends Paxos based on the
observation that it is only necessary for safety that every phase 1 quorum intersects
with every phase 2 quorum (quorums of the same phase do not have to intersect).
This allows greater flexibility and introduces an adjustable trade-off between the cost of deciding on new values and
the cost of starting a new round. For example, in a system with 10
nodes, one may use sets of 8 nodes as phase 1 quorums,
and sets of 3 nodes phase 2 quorums.
EPR verification of Flexible Paxos is essentially the same as for
normal Paxos, except we introduce two quorum sorts (for phase 1 and phase 2),
and adapt the intersection axiom to:
\[
\forall q_1:\squorum_1, q_2:\squorum_2. \; \exists n:\snode. \;
\rmember_1(n,q_1) \land \rmember_2(n, q_2)
\]
The detailed model and the adjusted invariant appear in \Cref{sec:flexible-paxos}.

\subsection{Stoppable Paxos}

Stoppable Paxos~\cite{lamport_stoppable_2008} extends Multi-Paxos with the ability for a node to propose a special stop command in order to stop the algorithm, with the guarantee that if the stop command is decided in instance $i$, then no command is ever decided at an instance $j>i$.
Stoppable Paxos therefore enables Virtually Synchronous system reconfiguration~\cite{birman_virtually_2010,chockler_group_2001}: Stoppable Paxos stops in a state known to all participants, which can then start a new instance of Stoppable Paxos in a new configuration (e.g., in which participants have been added or removed); moreover, no pending commands can leak from a configuration to the next, as only the final state of the command sequence is transfered from one configuration to the next.

Stoppable Paxos may be the most intricate algorithm in the Paxos family: as acknowledged by Lamport et al.~\cite{lamport_stoppable_2008}, ``getting the details right was not easy''.
The main algorithmic difficulty in Stoppable Paxos is to ensure that no command may be decided after a stop command while at the same time allowing a node to propose new commands without waiting, when earlier commands are still in flight (which is important for performance).
In contrast, in the traditional approach to reconfigurable SMR~\cite{lamport_reconfiguring_2010}, a node that has $c$ outstanding command proposals may cause up to $c$ commands to be decided after a stop command is decided;
Those commands needs to be passed-on to the next configuration and may contain other stop commands, adding to the complexity of the reconfiguration system.

Before proposing a command in an instance in Stoppable Paxos, a node
must check if other instances have seen stop commands proposed and in
which round.  This creates a non-trivial dependency between rounds and
instances, which are mostly orthogonal concepts in other variants of
Paxos.  This manifest as the instance sort having no incoming edge in
the quantifier alternation graph in other variants, while such edges
appear in Stoppable Paxos.  Interestingly, the rule given by Lamport
et al. to propose commands results in verification conditions that are
not in EPR, and rewriting seems difficult.  However, we found an
alternative rule which results in EPR verification conditions.  This
alternative rule soundly overapproximates the original rule (and
introduces new behaviors), and, as we have verified (in EPR), it also
maintains safety. The details of the modified rule and its
verification appear in \Cref{sec:stoppable-paxos}.

\section{Experimental Evaluation}
\label{sec:exp}

\newcommand{\tepr}{\textbf{EPR}}
\newcommand{\tiaux}{{\mathbf{\Iaux}}}
\newcommand{\trewrite}{\textbf{RW}}
\newcommand{\tfol}{\textbf{FOL}}
\newcommand{\tna}{-}

\begin{figure}
  \begin{footnotesize}
  \arraycolsep=4pt
  \[
  \begin{array}{||l|c|c|c|c|c|c|c|c|c|c|c|c|c|c|c|c|c|c||}
    \hhline{|t:===================:t|}
    &
    \multicolumn{2}{c|}{ \tepr} &
    \multicolumn{2}{c|}{\tiaux} &
    \multicolumn{2}{c|}{\trewrite} &
    \multicolumn{3}{c|}{\tfol-2} &
    \multicolumn{3}{c|}{\tfol-4} &
    \multicolumn{3}{c|}{\tfol-8} &
    \multicolumn{3}{c||}{\tfol-16}
    \\
    \textbf{Protocol} &
    \mu & \sigma &
    \mu & \sigma &
    \mu & \sigma &
    \mu & \sigma & \tau &
    \mu & \sigma & \tau &
    \mu & \sigma & \tau &
    \mu & \sigma & \tau
    \\
    \hhline{||-------------------||}
    \mbox{Paxos} &
    1.0 & 0.1 &
    0.7 & 0 &
    0.5 & 0 &
    1.2 & 0.2 & 0 &
    2.5 & 1.2 & 0 &
    86 & 112 & 2 &
    278 & 65 & 9
    \\
    \hhline{||-------------------||}
    \mbox{Multi-Paxos} &
    1.2 & 0.1 &
    0.8 & 0 &
    0.6 & 0 &
    1.2 & 0.1 & 0 &
    1.8 & 0.4 & 0 &
    107 & 129 & 3 &
    229 & 110 & 7
    \\
    \hhline{||-------------------||}
    \mbox{Vertical Paxos} &
    2.2 & 0.2 &
    \tna & \tna &
    \tna & \tna &
    27 & 10 & 0 &
    47 & 32 & 0 &
    209 & 104 & 5 &
    274 & 78 & 9
    \\
    \hhline{||-------------------||}
    \mbox{Fast Paxos} &
    4.7 & 1.6 &
    0.9 & 0 &
    0.6 & 0 &
    3.7 & 0.9 & 0 &
    19 & 25 & 0 &
    127 & 97 & 2 &
    300 & 0 & 10
    \\
    \hhline{||-------------------||}
    \mbox{Flexible Paxos} &
    1.0 & 0 &
    0.7 & 0 &
    0.5 & 0 &
    1.1 & 0.1 & 0 &
    2.7 & 2.1 & 0 &
    100 & 120 & 2 &
    275 & 75 & 9
    \\
    \hhline{||-------------------||}
    \mbox{Stoppable Paxos} &
    3.8 & 0.9 &
    1.0 & 0 &
    0.6 & 0 &
    186 & 123 & 5 &
    300 & 0 & 10 &
    300 & 0 & 10 &
    300 & 0 & 10
    \\
    \hhline{|t:===================:t|}
  \end{array}
  \]
  \end{footnotesize}
  \caption{
    \label{fig:exp-results}
    Run times (in seconds) of checking verification conditions using IVy and Z3.
    Each experiment was repeated 10 times (with random seeds used for Z3's heuristics).
    $\mu$ reports the mean time, $\sigma$ reports the standard deviation, and $\tau$ reports the number of runs
    that timed out at 300 seconds (where this occurred).
    $\tepr$ is the verification of the EPR model.
    $\tiaux$ is the verification of the auxiliary invariant.
    $\trewrite$ is the verification of the rewrite condition.
    ${\tfol}-N$ is the run time of semi-bounded verification of the first-order model, with bound 2 for values and bound $N$ for rounds (in all variants, bounding the number of values and rounds eliminates cycles from the quantifier alternation graph).
  }
\end{figure}

We have implemented our methodology using the IVy
tool~\cite{ivy,ken_fmcad16}, which uses the Z3 theorem
prover~\cite{z3} for checking verification
conditions. \Cref{fig:exp-results} lists the run times for the
automated checks performed when verifying the different Paxos
variants. The experiments were performed on a laptop running Linux,
with a Core-i7 1.8 GHz CPU. Z3 version 4.5.0 was used, along with the
latest version of IVy (commit
$\texttt{7ce6738}$)\footnote{\Cref{fig:exp-results} varies from
  \cite{oopsla17-epr}, since the Multi-Paxos and the Stoppable Paxos
  models are slightly improved (see \Cref{sec:multi-paxos}), and a
  more recent version of IVy was used.}. Z3 uses heuristics which
employ randomness. Therefore, each experiment was repeated 10 times
using random seeds. We report the mean times, as well as the standard
deviation and the number of experiments which timed out at 300 seconds
(these are included in the mean).  The IVy files used for these
experiments are available at the supplementary web page of this
paper\footnote{\url{http://www.cs.tau.ac.il/~odedp/paxos-made-epr.html}}.

For each variant, \Cref{fig:exp-results} reports the time for checking
the inductive invariant that proves the safety of the EPR model, as
well as the times required to verify auxiliary invariants and rewrite
conditions (see \Cref{sec:transformation-epr}). We also report on the
times required to check the inductive invariant for the original
first-order logic models, using semi-bounded verification. In all
variants, quantifier alternation cycles can be eliminated by bounding
the number of values and rounds. We bound the number of values to 2,
and vary the bound on the number of rounds.

As \Cref{fig:exp-results} demonstrates, using our methodology for EPR
verification results in verification conditions that are solved by Z3
in a few seconds, with no timeouts, and with negligible variance among
runs. In contrast, when using semi-bounded verification, the run time
quickly increases as we attempt to increase the number of
rounds. Moreover, the variance in run time increases significantly,
causing an unpredictable experience for verification users. We have
also attempted to use unbounded verification for the first-order logic
models, but Z3 diverged in this case for all variants. This shows the
practical value of our methodology, as it allows to transform models
whose verification condition cannot be handled by Z3 (and demonstrate
poor scalability and predictability for bounded verification), into
models that can be verified by Z3 in a few seconds.

\section{Related Work}
\label{sec:related}

\paragraph{Automated verification of distributed protocols}

Here we review several works that developed techniques for automated
verification of distributed protocols, and compare them with our
approach.

The Consensus Verification Logic
$\mathbb{CL}$~\cite{dragoi_logic-based_2014} is a logic tailored to
verify consensus algorithms in the partially synchronous Heard-Of
Model~\cite{charron-bost_heard-model:_2009}, with a decidable fragment
that can be used for verification. PSync~\cite{dragoi_psync:_2016} is
a domain-specific programming language and runtime for developing formally
verified implementations of consensus algorithms based on
$\mathbb{CL}$ and the Heard-Of Model. Once the user provides inductive
invariants and ranking functions in $\mathbb{CL}$, safety and liveness
can be automatically verified. PSync's verified implementation of
LastVoting (Paxos in the Heard-Of Model) is comparable in
performance with state-of-the-art unverified systems.

Many interesting theoretical decidability results, as well as the ByMC
verification tool, have been developed based on the formalism of
Threshold Automata
\cite{vienna_book,konnov_smt_2015,konnov_short_2017,DBLP:conf/ershov/KonnovVW15}. This
formalism allows to express a restricted class of distributed
algorithms operating in a partially synchronous communication
mode. This restriction allows decidability results based on cutoff
theorems, for both safety and liveness.

\cite{alberti_counting_2016} present a decidable fragment of
Presburger arithmetic with function symbols and cardinality constraint
over interpreted sets.  Their work is motivated by applications to the
verification of fault-tolerant distributed algorithms, and they
demonstrate automatic safety verification of some fault-tolerant
distributed algorithms expressed in a partially synchronous
round-by-round model similar to PSync.

\newcommand{\sharpie}{\#\Pi} $\sharpie$
\cite{gleissenthall_cardinalities_2016} present a logic that combines
reasoning about set cardinalities and universal quantifiers, along with an
invariant synthesis method. The logic is not decidable, so a sound and
incomplete reasoning method is used to check inductive
invariants. Inductive invariants are automatically synthesized by
method of Horn constraint solving. The technique is applied to
automatically verify a collection of parameterized systems, including
mutual exclusion, consensus, and garbage collection. However,
Paxos-like algorithms are beyond the reach of this verification
methodology since they require more complicated inductive invariants.

\newcommand{\consl}{\mathit{ConsL}} Recently, \cite{MSB17} presented a
cutoff result for consensus algorithms.  They define $\consl$, a
domain specific language for consensus algorithms,
whose semantics is based on the Heard-Of Model. $\consl$ admits a cutoff
theorem, i.e., a parameterized algorithm expressed in $\consl$ is
correct (for any number of processors) if and only if it is correct
up to a some finite bounded number of processors (e.g., for Paxos the bound is 5) .
This theoretical result shows that for algorithms expressible in
$\consl$, verification is decidable. However, $\consl$ is focused on
algorithms for the core consensus problem, and does not support the
infinite-state per process that is needed, e.g., to model Multi-Paxos
and SMR.

The above mentioned works obtain automation (and some decidability) by
restricting the programming model. We note that our approach takes a
different path to decidability compared to these works. We axiomatize
arithmetic, set cardinalities, and other higher-order concepts in an
uninterpreted first-order abstraction. This is in contrast to the
above works, in which these concepts are baked into specially designed
logics and formalisms. Furthermore, we start with a Turing-complete
modeling language and invariants with unrestricted quantifier
alternation, and provide a methodology to reduce quantifier
alternation to obtain decidability.  This allows us to employ a
general-purpose decidable logic to verify asynchronous Paxos,
Multi-Paxos, and their variants, which are beyond the reach of all of
the above works.

\paragraph{Deductive verification in undecidable logic}
IronFleet~\cite{IronFleet} is a verified implementation of SMR, using
the Dafny~\cite{dafny} program verifier, with verified safety and
liveness properties. Compared to our work, this system implementation
is considerably more detailed. The verification using Dafny employs Z3
to check verification conditions expressed in undecidable logics that
combine multiple theories and quantifier alternations. This leads to
great difficulties due to the unpredictability of the solver.  To
mitigate some of this unpredictability, IronFleet adopted a style they
call \emph{invariant quantifier hiding}.  This style attempts to
specify the invariants in a way that reduces the quantifiers that are
explicitly exposed to the solver. Our work is motivated by the
IronFleet experience and observations. The methodology we develop
provides a more systematic treatment of quantifier alternations, and
reduces the verification conditions to a decidable logic.

\paragraph{Verification using interactive theorem provers}

Recently, the Coq~\cite{coq} proof assistant has been used to develop
verified implementations of systems, such as a file system
\cite{fscq}, and shared memory data structures
\cite{DBLP:conf/pldi/SergeyNB15}. Closer to our work is
Verdi~\cite{DBLP:conf/pldi/WilcoxWPTWEA15}, which presents a verified
implementation of an SMR system based on Raft~\cite{raft}, a
Paxos-like algorithm.
This approach requires great effort, due
to the manual process of the proof; developing a
verified SMR system requires many months of work by verification
experts, and proofs measure in thousands of lines.

\cite{rahli_15th_2015, schiper_developing_2014} verify the safety of
implementations of consensus and SMR algorithms in the EventML
programming language. EventML interfaces with the Nuprl theorem
prover, in which proofs are conducted, and uses Nuprl's code
generation facilities.

Other works applied interactive theorem proving to verify Paxos
protocols at the algorithmic level, without an executable
implementation. \cite{DiskPaxos-AFP} proved the correctness of the
Disk Paxos algorithm in Isabelle/HOL~\cite{nipkow_isabelle/hol:_2002},
in about 6,500 lines of proof script. Recently,
\cite{chand_formal_2016} presented safety proofs of Paxos and
Multi-Paxos using the TLA+~\cite{lamport_temporal_1994} specification
language and the TLA Proof System
TLAPS~\cite{chaudhuri_tla+proof_2010}. TLA+ has also been used in
Amazon to model check distributed algorithms~\cite{amazon}. However,
they did not spend the effort required to obtain formal proofs, and
only used the TLA+ models for bug finding via the TLA+ model
checker~\cite{tlc}.

Compared to our approach, using interactive theorem provers requires
more user expertise and is more labor intensive. We note that
part of the difficulty in using an interactive theorem prover lies in
the unpredictability of the automated proof methods available and the
considerable experience needed to write proofs in an idiomatic style
that facilitates automation. An interesting direction of research is
to integrate our methodology in an interactive theorem prover to
achieve predictable automation in a style that is natural to systems
designers.

\paragraph{Works based on EPR}

\cite{ivy} and \cite{ken_fmcad16} have also used EPR to verify
distributed protocols and cache coherence protocols. \cite{ivy}
develops an interactive technique for assisting the user to find
universally quantified invariants (without quantifier alternations).
In contrast, here we use invariants that contain quantifier alternations,
as used in proofs of Paxos protocols.
\cite{ken_fmcad16} goes beyond our work by extracting executable code from the modeling language.
In the future, we plan to apply a similar extraction methodology to Paxos protocols.

In \cite{CAV:IBINS13,POPL:IBILNS14,shachar_thesis} it was shown that
EPR can express a limited form of transitive closure, in the context
of linked lists manipulations. We notice that in the context of our methodology,
their treatment of transitive closure can be considered
as adding a derived relation. This work and our work both show that EPR is
surprisingly powerful, when augmented with derived relations.

In \cite{tacas17-bh}, bounded quantifier instantiation is explored as
a possible solution to the undecidability caused by %
quantifier alternations. This work shares some of the motivation and
challenges with our work, but proposes an alternative solution. The
context we consider here is also wider, since we deal not only with
quantifier alternations in the inductive invariant, but also with
quantifier alternations in the transition relation. \cite{tacas17-bh}
also shows an interesting connection between derived relations and quantifier instantiation, and these insights may
apply to our methodology as well. An appealing future research direction is to 
combine user provided derived relations and rewrites
together with heuristically generated quantifier instantiation.

\section{Conclusion}
\label{sec:conclusion}

In this paper we have shown how to verify interesting distributed
protocols using EPR---a decidable fragment of first-order logic, which is supported
by existing solvers (e.g.,~\cite{z3,DBLP:conf/cade/Korovin08,vampire,DBLP:conf/cav/BarrettCDHJKRT11}). %
To mitigate the gap between the complexity of Paxos-like protocols
and the restrictions of EPR, we developed a methodology for gradually eliminating complications.
While this process requires assistance from the user, its steps are also mechanically checked (in EPR) to guarantee soundness.

We believe that our methodology can be applied to other distributed
protocols as well, as our setting is very general. The generality of
our approach is rooted in the use of first-order logic, with arbitrary
relations and functions, and a Turing-complete imperative language.

While EPR has shown to be surprisingly powerful, we note that it
is not a panacea, as some cyclic quantifier alternations may not be
avoided. Still, we have successfully used our methodology to
eliminate the cycles from several variants of Paxos, which is
considered a rich and complex protocol of great practical importance.

\begin{acks}
We thank
Yotam M. Y. Feldman,
Ken McMillan,
Yuri Meshman,
James R. Wilcox,
and the anonymous referees
for insightful comments which improved this paper.
Padon and Sagiv were supported by the European Research
Council under the European Union's Seventh Framework Program
(FP7/2007--2013) / ERC grant agreement no. [321174-VSSC].
This research was partially supported by
Len Blavatnik and the Blavatnik Family foundation.
This material is based upon work supported by the National Science
Foundation under Grant No. 1655166.
\end{acks}

\bibliography{refs}


\begin{thebibliography}{57}


\ifx \showCODEN    \undefined \def \showCODEN     #1{\unskip}     \fi
\ifx \showDOI      \undefined \def \showDOI       #1{#1}\fi
\ifx \showISBNx    \undefined \def \showISBNx     #1{\unskip}     \fi
\ifx \showISBNxiii \undefined \def \showISBNxiii  #1{\unskip}     \fi
\ifx \showISSN     \undefined \def \showISSN      #1{\unskip}     \fi
\ifx \showLCCN     \undefined \def \showLCCN      #1{\unskip}     \fi
\ifx \shownote     \undefined \def \shownote      #1{#1}          \fi
\ifx \showarticletitle \undefined \def \showarticletitle #1{#1}   \fi
\ifx \showURL      \undefined \def \showURL       {\relax}        \fi
\providecommand\bibfield[2]{#2}
\providecommand\bibinfo[2]{#2}
\providecommand\natexlab[1]{#1}
\providecommand\showeprint[2][]{arXiv:#2}

\bibitem[\protect\citeauthoryear{Alberti, Ghilardi, and Pagani}{Alberti
  et~al\mbox{.}}{2016}]%
        {alberti_counting_2016}
\bibfield{author}{\bibinfo{person}{Francesco Alberti}, \bibinfo{person}{Silvio
  Ghilardi}, {and} \bibinfo{person}{Elena Pagani}.}
  \bibinfo{year}{2016}\natexlab{}.
\newblock \showarticletitle{Counting {{Constraints}} in {{Flat Array
  Fragments}}}. In \bibinfo{booktitle}{{\em Automated {{Reasoning}}}}.
  \bibinfo{publisher}{{Springer, Cham}}, \bibinfo{pages}{65--81}.
\newblock


\bibitem[\protect\citeauthoryear{Barrett, Conway, Deters, Hadarean, Jovanovic,
  King, Reynolds, and Tinelli}{Barrett et~al\mbox{.}}{2011}]%
        {DBLP:conf/cav/BarrettCDHJKRT11}
\bibfield{author}{\bibinfo{person}{Clark Barrett},
  \bibinfo{person}{Christopher~L. Conway}, \bibinfo{person}{Morgan Deters},
  \bibinfo{person}{Liana Hadarean}, \bibinfo{person}{Dejan Jovanovic},
  \bibinfo{person}{Tim King}, \bibinfo{person}{Andrew Reynolds}, {and}
  \bibinfo{person}{Cesare Tinelli}.} \bibinfo{year}{2011}\natexlab{}.
\newblock \showarticletitle{{CVC4}}. In \bibinfo{booktitle}{{\em Computer Aided
  Verification - 23rd International Conference, {CAV} 2011, Snowbird, UT, USA,
  July 14-20, 2011. Proceedings}}. \bibinfo{pages}{171--177}.
\newblock


\bibitem[\protect\citeauthoryear{Bertot and Cast{\'{e}}ran}{Bertot and
  Cast{\'{e}}ran}{2004}]%
        {coq}
\bibfield{author}{\bibinfo{person}{Yves Bertot} {and} \bibinfo{person}{Pierre
  Cast{\'{e}}ran}.} \bibinfo{year}{2004}\natexlab{}.
\newblock \bibinfo{booktitle}{{\em Interactive Theorem Proving and Program
  Development - Coq'Art: The Calculus of Inductive Constructions}}.
\newblock \bibinfo{publisher}{Springer}.
\newblock
\showISBNx{978-3-642-05880-6}
\showDOI{%
\url{https://doi.org/10.1007/978-3-662-07964-5}}


\bibitem[\protect\citeauthoryear{Birman}{Birman}{2010}]%
        {birman_virtually_2010}
\bibfield{author}{\bibinfo{person}{Ken Birman}.}
  \bibinfo{year}{2010}\natexlab{}.
\newblock \showarticletitle{A History of the Virtual Synchrony Replication
  Model}. In \bibinfo{booktitle}{{\em Replication: Theory and Practice}} {\em
  (\bibinfo{series}{Lecture Notes in Computer Science})},
  \bibfield{editor}{\bibinfo{person}{Bernadette Charron{-}Bost},
  \bibinfo{person}{Fernando Pedone}, {and} \bibinfo{person}{Andr{\'{e}}
  Schiper}} (Eds.), Vol.~\bibinfo{volume}{5959}. \bibinfo{publisher}{Springer},
  \bibinfo{pages}{91--120}.
\newblock
\showISBNx{978-3-642-11293-5}
\showDOI{%
\url{https://doi.org/10.1007/978-3-642-11294-2_6}}


\bibitem[\protect\citeauthoryear{Blanchette and Nipkow}{Blanchette and
  Nipkow}{2010}]%
        {blanchette_nitpick:_2010}
\bibfield{author}{\bibinfo{person}{Jasmin~Christian Blanchette} {and}
  \bibinfo{person}{Tobias Nipkow}.} \bibinfo{year}{2010}\natexlab{}.
\newblock \showarticletitle{Nitpick: {{A}} Counterexample Generator for
  Higher-Order Logic Based on a Relational Model Finder}. In
  \bibinfo{booktitle}{{\em International Conference on Interactive Theorem
  Proving}}. \bibinfo{publisher}{{Springer}}, \bibinfo{pages}{131--146}.
\newblock


\bibitem[\protect\citeauthoryear{Bloem, Jacobs, Khalimov, Konnov, Rubin, Veith,
  and Widder}{Bloem et~al\mbox{.}}{2015}]%
        {vienna_book}
\bibfield{author}{\bibinfo{person}{Roderick Bloem}, \bibinfo{person}{Swen
  Jacobs}, \bibinfo{person}{Ayrat Khalimov}, \bibinfo{person}{Igor Konnov},
  \bibinfo{person}{Sasha Rubin}, \bibinfo{person}{Helmut Veith}, {and}
  \bibinfo{person}{Josef Widder}.} \bibinfo{year}{2015}\natexlab{}.
\newblock \bibinfo{booktitle}{{\em Decidability of Parameterized
  Verification}}.
\newblock \bibinfo{publisher}{Morgan {\&} Claypool Publishers}.
\newblock
\showDOI{%
\url{https://doi.org/10.2200/S00658ED1V01Y201508DCT013}}


\bibitem[\protect\citeauthoryear{Burrows}{Burrows}{2006}]%
        {Chuby}
\bibfield{author}{\bibinfo{person}{Michael Burrows}.}
  \bibinfo{year}{2006}\natexlab{}.
\newblock \showarticletitle{The Chubby Lock Service for Loosely-Coupled
  Distributed Systems}. In \bibinfo{booktitle}{{\em 7th Symposium on Operating
  Systems Design and Implementation {{OSDI} '06), November 6-8, Seattle, WA,
  {USA}}}}. \bibinfo{publisher}{{USENIX} Association},
  \bibinfo{pages}{335--350}.
\newblock


\bibitem[\protect\citeauthoryear{Chand, Liu, and Stoller}{Chand
  et~al\mbox{.}}{2016}]%
        {chand_formal_2016}
\bibfield{author}{\bibinfo{person}{Saksham Chand}, \bibinfo{person}{Yanhong~A.
  Liu}, {and} \bibinfo{person}{Scott~D. Stoller}.}
  \bibinfo{year}{2016}\natexlab{}.
\newblock \showarticletitle{Formal {{Verification}} of {{Multi}}-{{Paxos}} for
  {{Distributed Consensus}}}. In \bibinfo{booktitle}{{\em {{FM}} 2016: {{Formal
  Methods}}: 21st {{International Symposium}}, {{Limassol}}, {{Cyprus}},
  {{November}} 9-11, 2016, {{Proceedings}} 21}}.
  \bibinfo{publisher}{{Springer}}, \bibinfo{pages}{119--136}.
\newblock


\bibitem[\protect\citeauthoryear{Charron-Bost and Schiper}{Charron-Bost and
  Schiper}{2009}]%
        {charron-bost_heard-model:_2009}
\bibfield{author}{\bibinfo{person}{Bernadette Charron-Bost} {and}
  \bibinfo{person}{André Schiper}.} \bibinfo{year}{2009}\natexlab{}.
\newblock \showarticletitle{The Heard-of Model: Computing in Distributed
  Systems with Benign Faults}.
\newblock \bibinfo{journal}{{\em Distributed Computing\/}}
  \bibinfo{volume}{22}, \bibinfo{number}{1} (\bibinfo{year}{2009}),
  \bibinfo{pages}{49--71}.
\newblock


\bibitem[\protect\citeauthoryear{Chaudhuri, Doligez, Lamport, and
  Merz}{Chaudhuri et~al\mbox{.}}{2010}]%
        {chaudhuri_tla+proof_2010}
\bibfield{author}{\bibinfo{person}{Kaustuv Chaudhuri}, \bibinfo{person}{Damien
  Doligez}, \bibinfo{person}{Leslie Lamport}, {and} \bibinfo{person}{Stephan
  Merz}.} \bibinfo{year}{2010}\natexlab{}.
\newblock \showarticletitle{The {{TLA}}+{{Proof System}}: {{Building}} a
  {{Heterogeneous Verification Platform}}}. In \bibinfo{booktitle}{{\em
  Proceedings of the 7th {{International Colloquium Conference}} on
  {{Theoretical Aspects}} of {{Computing}}}} {\em
  (\bibinfo{series}{ICTAC'10})}. \bibinfo{publisher}{{Springer-Verlag}},
  \bibinfo{pages}{44--44}.
\newblock
\showISBNx{978-3-642-14807-1}


\bibitem[\protect\citeauthoryear{Chen, Ziegler, Chajed, Chlipala, Kaashoek, and
  Zeldovich}{Chen et~al\mbox{.}}{2016}]%
        {fscq}
\bibfield{author}{\bibinfo{person}{Haogang Chen}, \bibinfo{person}{Daniel
  Ziegler}, \bibinfo{person}{Tej Chajed}, \bibinfo{person}{Adam Chlipala},
  \bibinfo{person}{M.~Frans Kaashoek}, {and} \bibinfo{person}{Nickolai
  Zeldovich}.} \bibinfo{year}{2016}\natexlab{}.
\newblock \showarticletitle{Using Crash Hoare Logic for Certifying the {FSCQ}
  File System}. In \bibinfo{booktitle}{{\em 2016 {USENIX} Annual Technical
  Conference, {USENIX} {ATC} 2016, Denver, CO, USA, June 22-24, 2016.}}
\newblock


\bibitem[\protect\citeauthoryear{Chockler, Keidar, and Vitenberg}{Chockler
  et~al\mbox{.}}{2001}]%
        {chockler_group_2001}
\bibfield{author}{\bibinfo{person}{Gregory~V. Chockler}, \bibinfo{person}{Idit
  Keidar}, {and} \bibinfo{person}{Roman Vitenberg}.}
  \bibinfo{year}{2001}\natexlab{}.
\newblock \showarticletitle{Group communication specifications: a comprehensive
  study}.
\newblock \bibinfo{journal}{{\em {ACM} Comput. Surv.\/}} \bibinfo{volume}{33},
  \bibinfo{number}{4} (\bibinfo{year}{2001}), \bibinfo{pages}{427--469}.
\newblock
\showDOI{%
\url{https://doi.org/10.1145/503112.503113}}


\bibitem[\protect\citeauthoryear{de~Moura and Bj{\o}rner}{de~Moura and
  Bj{\o}rner}{2008}]%
        {z3}
\bibfield{author}{\bibinfo{person}{Leonardo de Moura} {and}
  \bibinfo{person}{Nikolaj Bj{\o}rner}.} \bibinfo{year}{2008}\natexlab{}.
\newblock \showarticletitle{{Z3}: An Efficient {SMT} Solver}. In
  \bibinfo{booktitle}{{\em Tools and Algorithms for the Construction and
  Analysis of Systems, 14th International Conference, TACAS 2008, Held as Part
  of the Joint European Conferences on Theory and Practice of Software, ETAPS
  2008, Budapest, Hungary, March 29-April 6, 2008. Proceedings}} {\em
  (\bibinfo{series}{Lecture Notes in Computer Science})},
  Vol.~\bibinfo{volume}{4963}. \bibinfo{publisher}{Springer},
  \bibinfo{pages}{337--340}.
\newblock


\bibitem[\protect\citeauthoryear{Dragoi, Henzinger, Veith, Widder, and
  Zufferey}{Dragoi et~al\mbox{.}}{2014}]%
        {dragoi_logic-based_2014}
\bibfield{author}{\bibinfo{person}{Cezara Dragoi}, \bibinfo{person}{Thomas~A.
  Henzinger}, \bibinfo{person}{Helmut Veith}, \bibinfo{person}{Josef Widder},
  {and} \bibinfo{person}{Damien Zufferey}.} \bibinfo{year}{2014}\natexlab{}.
\newblock \showarticletitle{A Logic-Based Framework for Verifying Consensus
  Algorithms}. In \bibinfo{booktitle}{{\em International {{Conference}} on
  {{Verification}}, {{Model Checking}}, and {{Abstract Interpretation}}}}.
  \bibinfo{publisher}{{Springer}}, \bibinfo{pages}{161--181}.
\newblock


\bibitem[\protect\citeauthoryear{Dragoi, Henzinger, and Zufferey}{Dragoi
  et~al\mbox{.}}{2016}]%
        {dragoi_psync:_2016}
\bibfield{author}{\bibinfo{person}{Cezara Dragoi}, \bibinfo{person}{Thomas~A.
  Henzinger}, {and} \bibinfo{person}{Damien Zufferey}.}
  \bibinfo{year}{2016}\natexlab{}.
\newblock \showarticletitle{{{PSync}}: A Partially Synchronous Language for
  Fault-Tolerant Distributed Algorithms}.
\newblock \bibinfo{journal}{{\em ACM SIGPLAN Notices\/}} \bibinfo{volume}{51},
  \bibinfo{number}{1} (\bibinfo{year}{2016}), \bibinfo{pages}{400--415}.
\newblock


\bibitem[\protect\citeauthoryear{Feldman, Padon, Immerman, Sagiv, and
  Shoham}{Feldman et~al\mbox{.}}{2017}]%
        {tacas17-bh}
\bibfield{author}{\bibinfo{person}{Yotam M.~Y. Feldman}, \bibinfo{person}{Oded
  Padon}, \bibinfo{person}{Neil Immerman}, \bibinfo{person}{Mooly Sagiv}, {and}
  \bibinfo{person}{Sharon Shoham}.} \bibinfo{year}{2017}\natexlab{}.
\newblock \showarticletitle{Bounded Quantifier Instantiation for Checking
  Inductive Invariants}. In \bibinfo{booktitle}{{\em Tools and Algorithms for
  the Construction and Analysis of Systems - 23rd International Conference,
  {TACAS} 2017, Held as Part of the European Joint Conferences on Theory and
  Practice of Software, {ETAPS} 2017, Uppsala, Sweden, April 22-29, 2017,
  Proceedings, Part {I}}} {\em (\bibinfo{series}{Lecture Notes in Computer
  Science})}, \bibfield{editor}{\bibinfo{person}{Axel Legay} {and}
  \bibinfo{person}{Tiziana Margaria}} (Eds.), Vol.~\bibinfo{volume}{10205}.
  \bibinfo{pages}{76--95}.
\newblock
\showISBNx{978-3-662-54576-8}
\showDOI{%
\url{https://doi.org/10.1007/978-3-662-54577-5_5}}


\bibitem[\protect\citeauthoryear{Hawblitzel, Howell, Kapritsos, Lorch, Parno,
  Roberts, Setty, and Zill}{Hawblitzel et~al\mbox{.}}{2015}]%
        {IronFleet}
\bibfield{author}{\bibinfo{person}{Chris Hawblitzel}, \bibinfo{person}{Jon
  Howell}, \bibinfo{person}{Manos Kapritsos}, \bibinfo{person}{Jacob~R. Lorch},
  \bibinfo{person}{Bryan Parno}, \bibinfo{person}{Michael~L. Roberts},
  \bibinfo{person}{Srinath T.~V. Setty}, {and} \bibinfo{person}{Brian Zill}.}
  \bibinfo{year}{2015}\natexlab{}.
\newblock \showarticletitle{IronFleet: proving practical distributed systems
  correct}. In \bibinfo{booktitle}{{\em Proceedings of the 25th Symposium on
  Operating Systems Principles, {SOSP}}}. \bibinfo{pages}{1--17}.
\newblock


\bibitem[\protect\citeauthoryear{Howard, Malkhi, and Spiegelman}{Howard
  et~al\mbox{.}}{2016}]%
        {howard_flexible_2016}
\bibfield{author}{\bibinfo{person}{Heidi Howard}, \bibinfo{person}{Dahlia
  Malkhi}, {and} \bibinfo{person}{Alexander Spiegelman}.}
  \bibinfo{year}{2016}\natexlab{}.
\newblock \showarticletitle{Flexible Paxos: {{Quorum}} Intersection Revisited}.
\newblock \bibinfo{journal}{{\em arXiv preprint arXiv:1608.06696\/}}
  (\bibinfo{year}{2016}).
\newblock


\bibitem[\protect\citeauthoryear{Itzhaky}{Itzhaky}{2014}]%
        {shachar_thesis}
\bibfield{author}{\bibinfo{person}{Shachar Itzhaky}.}
  \bibinfo{year}{2014}\natexlab{}.
\newblock {\em \bibinfo{title}{Automatic Reasoning for Pointer Programs Using
  Decidable Logics}}.
\newblock \bibinfo{thesistype}{Ph.D. Dissertation}. \bibinfo{school}{Tel Aviv
  University}.
\newblock


\bibitem[\protect\citeauthoryear{Itzhaky, Banerjee, Immerman, Lahav, Nanevski,
  and Sagiv}{Itzhaky et~al\mbox{.}}{2014}]%
        {POPL:IBILNS14}
\bibfield{author}{\bibinfo{person}{Shachar Itzhaky}, \bibinfo{person}{Anindya
  Banerjee}, \bibinfo{person}{Neil Immerman}, \bibinfo{person}{Ori Lahav},
  \bibinfo{person}{Aleksandar Nanevski}, {and} \bibinfo{person}{Mooly Sagiv}.}
  \bibinfo{year}{2014}\natexlab{}.
\newblock \showarticletitle{Modular reasoning about heap paths via effectively
  propositional formulas}. In \bibinfo{booktitle}{{\em the 41st Annual {ACM}
  {SIGPLAN-SIGACT} Symposium on Principles of Programming Languages, {POPL}}}.
  \bibinfo{pages}{385--396}.
\newblock


\bibitem[\protect\citeauthoryear{Itzhaky, Banerjee, Immerman, Nanevski, and
  Sagiv}{Itzhaky et~al\mbox{.}}{2013}]%
        {CAV:IBINS13}
\bibfield{author}{\bibinfo{person}{Shachar Itzhaky}, \bibinfo{person}{Anindya
  Banerjee}, \bibinfo{person}{Neil Immerman}, \bibinfo{person}{Aleksandar
  Nanevski}, {and} \bibinfo{person}{Mooly Sagiv}.}
  \bibinfo{year}{2013}\natexlab{}.
\newblock \showarticletitle{Effectively-Propositional Reasoning about
  Reachability in Linked Data Structures}. In \bibinfo{booktitle}{{\em CAV}}
  {\em (\bibinfo{series}{LNCS})}, Vol.~\bibinfo{volume}{8044}.
  \bibinfo{pages}{756--772}.
\newblock


\bibitem[\protect\citeauthoryear{Jackson}{Jackson}{2006}]%
        {alloy}
\bibfield{author}{\bibinfo{person}{Daniel Jackson}.}
  \bibinfo{year}{2006}\natexlab{}.
\newblock \bibinfo{booktitle}{{\em Software Abstractions: Logic, Language, and
  Analysis}}.
\newblock \bibinfo{publisher}{The MIT Press}.
\newblock
\showISBNx{0262101149}


\bibitem[\protect\citeauthoryear{Jaskelioff and Merz}{Jaskelioff and
  Merz}{2005}]%
        {DiskPaxos-AFP}
\bibfield{author}{\bibinfo{person}{Mauro Jaskelioff} {and}
  \bibinfo{person}{Stephan Merz}.} \bibinfo{year}{2005}\natexlab{}.
\newblock \showarticletitle{Proving the Correctness of Disk Paxos}.
\newblock \bibinfo{journal}{{\em Archive of Formal Proofs\/}}
  (\bibinfo{date}{June} \bibinfo{year}{2005}).
\newblock
\showISSN{2150-914x}
\newblock
\shownote{\url{http://isa-afp.org/entries/DiskPaxos.shtml}, Formal proof
  development.}


\bibitem[\protect\citeauthoryear{Konnov, Lazic, Veith, and Widder}{Konnov
  et~al\mbox{.}}{2017}]%
        {konnov_short_2017}
\bibfield{author}{\bibinfo{person}{Igor Konnov}, \bibinfo{person}{Marijana
  Lazic}, \bibinfo{person}{Helmut Veith}, {and} \bibinfo{person}{Josef
  Widder}.} \bibinfo{year}{2017}\natexlab{}.
\newblock \showarticletitle{A {{Short Counterexample Property}} for {{Safety}}
  and {{Liveness Verification}} of {{Fault}}-Tolerant {{Distributed
  Algorithms}}}. In \bibinfo{booktitle}{{\em Proceedings of the 44th {{ACM
  SIGPLAN Symposium}} on {{Principles}} of {{Programming Languages}}}} {\em
  (\bibinfo{series}{POPL 2017})}. \bibinfo{publisher}{{ACM}},
  \bibinfo{pages}{719--734}.
\newblock
\showISBNx{978-1-4503-4660-3}


\bibitem[\protect\citeauthoryear{Konnov, Veith, and Widder}{Konnov
  et~al\mbox{.}}{2015a}]%
        {konnov_smt_2015}
\bibfield{author}{\bibinfo{person}{Igor Konnov}, \bibinfo{person}{Helmut
  Veith}, {and} \bibinfo{person}{Josef Widder}.}
  \bibinfo{year}{2015}\natexlab{a}.
\newblock \showarticletitle{{{SMT}} and {{POR Beat Counter Abstraction}}:
  {{Parameterized Model Checking}} of {{Threshold}}-{{Based Distributed
  Algorithms}}}. In \bibinfo{booktitle}{{\em Computer {{Aided Verification}}}}.
  \bibinfo{publisher}{{Springer, Cham}}, \bibinfo{pages}{85--102}.
\newblock


\bibitem[\protect\citeauthoryear{Konnov, Veith, and Widder}{Konnov
  et~al\mbox{.}}{2015b}]%
        {DBLP:conf/ershov/KonnovVW15}
\bibfield{author}{\bibinfo{person}{Igor~V. Konnov}, \bibinfo{person}{Helmut
  Veith}, {and} \bibinfo{person}{Josef Widder}.}
  \bibinfo{year}{2015}\natexlab{b}.
\newblock \showarticletitle{What You Always Wanted to Know About Model Checking
  of Fault-Tolerant Distributed Algorithms}. In \bibinfo{booktitle}{{\em
  Perspectives of System Informatics - 10th International Andrei Ershov
  Informatics Conference, {PSI} 2015, in Memory of Helmut Veith, Kazan and
  Innopolis, Russia, August 24-27, 2015, Revised Selected Papers}} {\em
  (\bibinfo{series}{Lecture Notes in Computer Science})},
  \bibfield{editor}{\bibinfo{person}{Manuel Mazzara} {and}
  \bibinfo{person}{Andrei Voronkov}} (Eds.), Vol.~\bibinfo{volume}{9609}.
  \bibinfo{publisher}{Springer}, \bibinfo{pages}{6--21}.
\newblock
\showISBNx{978-3-319-41578-9}
\showDOI{%
\url{https://doi.org/10.1007/978-3-319-41579-6_2}}


\bibitem[\protect\citeauthoryear{Korovin}{Korovin}{2008}]%
        {DBLP:conf/cade/Korovin08}
\bibfield{author}{\bibinfo{person}{Konstantin Korovin}.}
  \bibinfo{year}{2008}\natexlab{}.
\newblock \showarticletitle{{iProver} - An Instantiation-Based Theorem Prover
  for First-Order Logic (System Description)}. In \bibinfo{booktitle}{{\em
  Automated Reasoning, 4th International Joint Conference, {IJCAR} 2008,
  Sydney, Australia, August 12-15, 2008, Proceedings}}.
  \bibinfo{pages}{292--298}.
\newblock


\bibitem[\protect\citeauthoryear{Lamport}{Lamport}{1994}]%
        {lamport_temporal_1994}
\bibfield{author}{\bibinfo{person}{Leslie Lamport}.}
  \bibinfo{year}{1994}\natexlab{}.
\newblock \showarticletitle{The Temporal Logic of Actions}.
\newblock \bibinfo{journal}{{\em ACM Transactions on Programming Languages and
  Systems (TOPLAS)\/}} \bibinfo{volume}{16}, \bibinfo{number}{3}
  (\bibinfo{year}{1994}), \bibinfo{pages}{872--923}.
\newblock


\bibitem[\protect\citeauthoryear{Lamport}{Lamport}{1998}]%
        {paxos}
\bibfield{author}{\bibinfo{person}{Leslie Lamport}.}
  \bibinfo{year}{1998}\natexlab{}.
\newblock \showarticletitle{The Part-Time Parliament}.
\newblock \bibinfo{journal}{{\em {ACM} Trans. Comput. Syst.\/}}
  \bibinfo{volume}{16}, \bibinfo{number}{2} (\bibinfo{year}{1998}),
  \bibinfo{pages}{133--169}.
\newblock
\showDOI{%
\url{https://doi.org/10.1145/279227.279229}}


\bibitem[\protect\citeauthoryear{Lamport}{Lamport}{2001}]%
        {paxos-made-simple}
\bibfield{author}{\bibinfo{person}{Leslie Lamport}.}
  \bibinfo{year}{2001}\natexlab{}.
\newblock \showarticletitle{Paxos Made Simple}.
\newblock  (\bibinfo{date}{December} \bibinfo{year}{2001}),
  \bibinfo{pages}{51--58}.
\newblock
\showURL{%
\url{https://www.microsoft.com/en-us/research/publication/paxos-made-simple/}}


\bibitem[\protect\citeauthoryear{Lamport}{Lamport}{2002}]%
        {tla}
\bibfield{author}{\bibinfo{person}{Leslie Lamport}.}
  \bibinfo{year}{2002}\natexlab{}.
\newblock \bibinfo{booktitle}{{\em Specifying Systems: The TLA+ Language and
  Tools for Hardware and Software Engineers}}.
\newblock \bibinfo{publisher}{Addison-Wesley Longman Publishing Co., Inc.},
  \bibinfo{address}{Boston, MA, USA}.
\newblock
\showISBNx{032114306X}


\bibitem[\protect\citeauthoryear{Lamport}{Lamport}{2006}]%
        {lamport_fast_2006}
\bibfield{author}{\bibinfo{person}{Leslie Lamport}.}
  \bibinfo{year}{2006}\natexlab{}.
\newblock \showarticletitle{Fast Paxos}.
\newblock \bibinfo{journal}{{\em Distributed Computing\/}}
  \bibinfo{volume}{19}, \bibinfo{number}{2} (\bibinfo{year}{2006}),
  \bibinfo{pages}{79--103}.
\newblock


\bibitem[\protect\citeauthoryear{Lamport, Malkhi, and Zhou}{Lamport
  et~al\mbox{.}}{2008}]%
        {lamport_stoppable_2008}
\bibfield{author}{\bibinfo{person}{Leslie Lamport}, \bibinfo{person}{Dahlia
  Malkhi}, {and} \bibinfo{person}{Lidong Zhou}.}
  \bibinfo{year}{2008}\natexlab{}.
\newblock \bibinfo{booktitle}{{\em Stoppable Paxos}}.
\newblock \bibinfo{type}{{T}echnical {R}eport}.
  \bibinfo{institution}{TechReport, Microsoft Research}.
\newblock
\showURL{%
\url{https://www.microsoft.com/en-us/research/publication/stoppable-paxos/}}


\bibitem[\protect\citeauthoryear{Lamport, Malkhi, and Zhou}{Lamport
  et~al\mbox{.}}{2009}]%
        {lamport_vertical_2009}
\bibfield{author}{\bibinfo{person}{Leslie Lamport}, \bibinfo{person}{Dahlia
  Malkhi}, {and} \bibinfo{person}{Lidong Zhou}.}
  \bibinfo{year}{2009}\natexlab{}.
\newblock \showarticletitle{Vertical Paxos and Primary-Backup Replication}. In
  \bibinfo{booktitle}{{\em Proceedings of the 28th {{ACM}} Symposium on
  {{Principles}} of Distributed Computing}}. \bibinfo{publisher}{{ACM}},
  \bibinfo{pages}{312--313}.
\newblock


\bibitem[\protect\citeauthoryear{Lamport, Malkhi, and Zhou}{Lamport
  et~al\mbox{.}}{2010}]%
        {lamport_reconfiguring_2010}
\bibfield{author}{\bibinfo{person}{Leslie Lamport}, \bibinfo{person}{Dahlia
  Malkhi}, {and} \bibinfo{person}{Lidong Zhou}.}
  \bibinfo{year}{2010}\natexlab{}.
\newblock \showarticletitle{Reconfiguring a {{State Machine}}}.
\newblock \bibinfo{journal}{{\em SIGACT News\/}} \bibinfo{volume}{41},
  \bibinfo{number}{1} (\bibinfo{date}{03} \bibinfo{year}{2010}),
  \bibinfo{pages}{63--73}.
\newblock


\bibitem[\protect\citeauthoryear{Leino}{Leino}{2010}]%
        {dafny}
\bibfield{author}{\bibinfo{person}{K~Rustan~M Leino}.}
  \bibinfo{year}{2010}\natexlab{}.
\newblock \showarticletitle{Dafny: An automatic program verifier for functional
  correctness}. In \bibinfo{booktitle}{{\em Logic for Programming, Artificial
  Intelligence, and Reasoning}}. Springer, \bibinfo{pages}{348--370}.
\newblock


\bibitem[\protect\citeauthoryear{Lewis}{Lewis}{1980}]%
        {LEWIS1980317}
\bibfield{author}{\bibinfo{person}{Harry~R. Lewis}.}
  \bibinfo{year}{1980}\natexlab{}.
\newblock \showarticletitle{Complexity results for classes of quantificational
  formulas}.
\newblock \bibinfo{journal}{{\it J. Comput. System Sci.}} \bibinfo{volume}{21},
  \bibinfo{number}{3} (\bibinfo{year}{1980}), \bibinfo{pages}{317 -- 353}.
\newblock


\bibitem[\protect\citeauthoryear{Lynch and Tuttle}{Lynch and Tuttle}{1987}]%
        {lynch_hierarchical_1987}
\bibfield{author}{\bibinfo{person}{Nancy~A. Lynch} {and}
  \bibinfo{person}{Mark~R. Tuttle}.} \bibinfo{year}{1987}\natexlab{}.
\newblock \showarticletitle{Hierarchical Correctness Proofs for Distributed
  Algorithms}. In \bibinfo{booktitle}{{\em Proceedings of the Sixth Annual
  {{ACM Symposium}} on {{Principles}} of Distributed Computing}}.
  \bibinfo{publisher}{{ACM}}, \bibinfo{pages}{137--151}.
\newblock


\bibitem[\protect\citeauthoryear{Maric, Sprenger, and Basin}{Maric
  et~al\mbox{.}}{2017}]%
        {MSB17}
\bibfield{author}{\bibinfo{person}{Ognjen Maric}, \bibinfo{person}{Christoph
  Sprenger}, {and} \bibinfo{person}{David~A. Basin}.}
  \bibinfo{year}{2017}\natexlab{}.
\newblock \showarticletitle{Cutoff Bounds for Consensus Algorithms}. In
  \bibinfo{booktitle}{{\em Computer Aided Verification - 29th International
  Conference, {CAV} 2017, Heidelberg, Germany, July 24-28, 2017, Proceedings,
  Part {II}}} {\em (\bibinfo{series}{Lecture Notes in Computer Science})},
  \bibfield{editor}{\bibinfo{person}{Rupak Majumdar} {and}
  \bibinfo{person}{Viktor Kuncak}} (Eds.), Vol.~\bibinfo{volume}{10427}.
  \bibinfo{publisher}{Springer}, \bibinfo{pages}{217--237}.
\newblock
\showISBNx{978-3-319-63389-3}
\showDOI{%
\url{https://doi.org/10.1007/978-3-319-63390-9_12}}


\bibitem[\protect\citeauthoryear{McMillan}{McMillan}{2016}]%
        {ken_fmcad16}
\bibfield{author}{\bibinfo{person}{Kenneth~L. McMillan}.}
  \bibinfo{year}{2016}\natexlab{}.
\newblock \showarticletitle{Modular specification and verification of a
  cache-coherent interface}. In \bibinfo{booktitle}{{\em 2016 Formal Methods in
  Computer-Aided Design, {FMCAD} 2016, Mountain View, CA, USA, October 3-6,
  2016}}, \bibfield{editor}{\bibinfo{person}{Ruzica Piskac} {and}
  \bibinfo{person}{Muralidhar Talupur}} (Eds.). \bibinfo{publisher}{{IEEE}},
  \bibinfo{pages}{109--116}.
\newblock
\showISBNx{978-0-9835678-6-8}
\showDOI{%
\url{https://doi.org/10.1109/FMCAD.2016.7886668}}


\bibitem[\protect\citeauthoryear{Newcombe, Rath, Zhang, Munteanu, Brooker, and
  Deardeuff}{Newcombe et~al\mbox{.}}{2015}]%
        {amazon}
\bibfield{author}{\bibinfo{person}{Chris Newcombe}, \bibinfo{person}{Tim Rath},
  \bibinfo{person}{Fan Zhang}, \bibinfo{person}{Bogdan Munteanu},
  \bibinfo{person}{Marc Brooker}, {and} \bibinfo{person}{Michael Deardeuff}.}
  \bibinfo{year}{2015}\natexlab{}.
\newblock \showarticletitle{How {Amazon} web services uses formal methods}.
\newblock \bibinfo{journal}{{\it Commun. ACM}} \bibinfo{volume}{58},
  \bibinfo{number}{4} (\bibinfo{year}{2015}), \bibinfo{pages}{66--73}.
\newblock


\bibitem[\protect\citeauthoryear{Nipkow, Paulson, and Wenzel}{Nipkow
  et~al\mbox{.}}{2002}]%
        {nipkow_isabelle/hol:_2002}
\bibfield{author}{\bibinfo{person}{Tobias Nipkow}, \bibinfo{person}{Lawrence~C.
  Paulson}, {and} \bibinfo{person}{Markus Wenzel}.}
  \bibinfo{year}{2002}\natexlab{}.
\newblock \bibinfo{booktitle}{{\em Isabelle/{{HOL}}: A Proof Assistant for
  Higher-Order Logic}}. Vol.~\bibinfo{volume}{2283}.
\newblock \bibinfo{publisher}{{Springer Science \& Business Media}}.
\newblock


\bibitem[\protect\citeauthoryear{Ongaro and Ousterhout}{Ongaro and
  Ousterhout}{2014}]%
        {raft}
\bibfield{author}{\bibinfo{person}{Diego Ongaro} {and} \bibinfo{person}{John~K.
  Ousterhout}.} \bibinfo{year}{2014}\natexlab{}.
\newblock \showarticletitle{In Search of an Understandable Consensus
  Algorithm}. In \bibinfo{booktitle}{{\em 2014 {USENIX} Annual Technical
  Conference, {USENIX} {ATC} '14, Philadelphia, PA, USA, June 19-20, 2014.}}
  \bibinfo{pages}{305--319}.
\newblock
\showURL{%
\url{https://www.usenix.org/conference/atc14/technical-sessions/presentation/ongaro}}


\bibitem[\protect\citeauthoryear{Padon, Losa, Sagiv, and Shoham}{Padon
  et~al\mbox{.}}{2017}]%
        {oopsla17-epr}
\bibfield{author}{\bibinfo{person}{Oded Padon}, \bibinfo{person}{Giuliano
  Losa}, \bibinfo{person}{Mooly Sagiv}, {and} \bibinfo{person}{Sharon Shoham}.}
  \bibinfo{year}{2017}\natexlab{}.
\newblock \showarticletitle{Paxos Made EPR: Decidable Reasoning About
  Distributed Protocols}.
\newblock \bibinfo{journal}{{\em Proc. ACM Program. Lang.\/}}
  \bibinfo{volume}{1}, \bibinfo{number}{OOPSLA}, Article
  \bibinfo{articleno}{108} (\bibinfo{date}{Oct.} \bibinfo{year}{2017}),
  \bibinfo{numpages}{31}~pages.
\newblock
\showISSN{2475-1421}
\showDOI{%
\url{https://doi.org/10.1145/3140568}}


\bibitem[\protect\citeauthoryear{Padon, McMillan, Panda, Sagiv, and
  Shoham}{Padon et~al\mbox{.}}{2016}]%
        {ivy}
\bibfield{author}{\bibinfo{person}{Oded Padon}, \bibinfo{person}{Kenneth~L.
  McMillan}, \bibinfo{person}{Aurojit Panda}, \bibinfo{person}{Mooly Sagiv},
  {and} \bibinfo{person}{Sharon Shoham}.} \bibinfo{year}{2016}\natexlab{}.
\newblock \showarticletitle{Ivy: safety verification by interactive
  generalization}. In \bibinfo{booktitle}{{\em Proceedings of the 37th {ACM}
  {SIGPLAN} Conference on Programming Language Design and Implementation,
  {PLDI} 2016, Santa Barbara, CA, USA, June 13-17, 2016}}.
  \bibinfo{pages}{614--630}.
\newblock


\bibitem[\protect\citeauthoryear{Paige and Koenig}{Paige and Koenig}{1982}]%
        {PaigeK82}
\bibfield{author}{\bibinfo{person}{Robert Paige} {and} \bibinfo{person}{Shaye
  Koenig}.} \bibinfo{year}{1982}\natexlab{}.
\newblock \showarticletitle{Finite Differencing of Computable Expressions}.
\newblock \bibinfo{journal}{{\em {ACM} Trans. Program. Lang. Syst.\/}}
  \bibinfo{volume}{4}, \bibinfo{number}{3} (\bibinfo{year}{1982}),
  \bibinfo{pages}{402--454}.
\newblock


\bibitem[\protect\citeauthoryear{Piskac, de~Moura, and Bj{\o}rner}{Piskac
  et~al\mbox{.}}{2010}]%
        {JAR:PiskacMB10}
\bibfield{author}{\bibinfo{person}{Ruzica Piskac},
  \bibinfo{person}{Leonardo~Mendon\c{c}a de Moura}, {and}
  \bibinfo{person}{Nikolaj Bj{\o}rner}.} \bibinfo{year}{2010}\natexlab{}.
\newblock \showarticletitle{Deciding Effectively Propositional Logic Using DPLL
  and Substitution Sets}.
\newblock \bibinfo{journal}{{\em J. Autom. Reasoning\/}} \bibinfo{volume}{44},
  \bibinfo{number}{4} (\bibinfo{year}{2010}), \bibinfo{pages}{401--424}.
\newblock


\bibitem[\protect\citeauthoryear{Rahli, Guaspari, Bickford, and
  Constable}{Rahli et~al\mbox{.}}{2015}]%
        {rahli_15th_2015}
\bibfield{author}{\bibinfo{person}{Vincent Rahli}, \bibinfo{person}{David
  Guaspari}, \bibinfo{person}{Mark Bickford}, {and} \bibinfo{person}{Robert~L.
  Constable}.} \bibinfo{year}{2015}\natexlab{}.
\newblock \showarticletitle{15th {{international workshop}} on {{automated
  verification}} of {{critical systems}} ({{avocs}} 2015)}.
\newblock \bibinfo{journal}{{\em electronic communications of the easst\/}}
  \bibinfo{volume}{72} (\bibinfo{year}{2015}).
\newblock


\bibitem[\protect\citeauthoryear{Reps, Sagiv, and Loginov}{Reps
  et~al\mbox{.}}{2010}]%
        {TOPLAS:RepsSL10}
\bibfield{author}{\bibinfo{person}{Thomas~W. Reps}, \bibinfo{person}{Mooly
  Sagiv}, {and} \bibinfo{person}{Alexey Loginov}.}
  \bibinfo{year}{2010}\natexlab{}.
\newblock \showarticletitle{Finite differencing of logical formulas for static
  analysis}.
\newblock \bibinfo{journal}{{\em ACM Trans. Program. Lang. Syst.\/}}
  \bibinfo{volume}{32}, \bibinfo{number}{6} (\bibinfo{year}{2010}).
\newblock


\bibitem[\protect\citeauthoryear{Riazanov and Voronkov}{Riazanov and
  Voronkov}{2002}]%
        {vampire}
\bibfield{author}{\bibinfo{person}{Alexandre Riazanov} {and}
  \bibinfo{person}{Andrei Voronkov}.} \bibinfo{year}{2002}\natexlab{}.
\newblock \showarticletitle{The Design and Implementation of VAMPIRE}.
\newblock \bibinfo{journal}{{\em AI Commun.\/}} \bibinfo{volume}{15},
  \bibinfo{number}{2,3} (\bibinfo{date}{Aug.} \bibinfo{year}{2002}),
  \bibinfo{pages}{91--110}.
\newblock
\showISSN{0921-7126}
\showURL{%
\url{http://dl.acm.org/citation.cfm?id=1218615.1218620}}


\bibitem[\protect\citeauthoryear{Schiper, Rahli, Renesse, Bickford, and
  Constable}{Schiper et~al\mbox{.}}{2014}]%
        {schiper_developing_2014}
\bibfield{author}{\bibinfo{person}{N. Schiper}, \bibinfo{person}{V. Rahli},
  \bibinfo{person}{R.~V. Renesse}, \bibinfo{person}{M. Bickford}, {and}
  \bibinfo{person}{R.~L. Constable}.} \bibinfo{year}{2014}\natexlab{}.
\newblock \showarticletitle{developing {{correctly replicated databases using
  formal tools}}}. In \bibinfo{booktitle}{{\em 2014 44th {{annual ieee}}/{{ifip
  international conference}} on {{dependable systems}} and {{networks}}}}.
  \bibinfo{pages}{395--406}.
\newblock


\bibitem[\protect\citeauthoryear{Schneider}{Schneider}{1990}]%
        {schneider_implementing_1990}
\bibfield{author}{\bibinfo{person}{Fred~B. Schneider}.}
  \bibinfo{year}{1990}\natexlab{}.
\newblock \showarticletitle{Implementing Fault-Tolerant Services Using the
  State Machine Approach: {{A}} Tutorial}.
\newblock \bibinfo{journal}{{\em ACM Computing Surveys (CSUR)\/}}
  \bibinfo{volume}{22}, \bibinfo{number}{4} (\bibinfo{year}{1990}),
  \bibinfo{pages}{299--319}.
\newblock


\bibitem[\protect\citeauthoryear{Sergey, Nanevski, and Banerjee}{Sergey
  et~al\mbox{.}}{2015}]%
        {DBLP:conf/pldi/SergeyNB15}
\bibfield{author}{\bibinfo{person}{Ilya Sergey}, \bibinfo{person}{Aleksandar
  Nanevski}, {and} \bibinfo{person}{Anindya Banerjee}.}
  \bibinfo{year}{2015}\natexlab{}.
\newblock \showarticletitle{Mechanized verification of fine-grained concurrent
  programs}. In \bibinfo{booktitle}{{\em Proceedings of the 36th {ACM}
  {SIGPLAN} Conference on Programming Language Design and Implementation,
  Portland, OR, USA, June 15-17, 2015}},
  \bibfield{editor}{\bibinfo{person}{David Grove} {and} \bibinfo{person}{Steve
  Blackburn}} (Eds.). \bibinfo{publisher}{{ACM}}, \bibinfo{pages}{77--87}.
\newblock
\showISBNx{978-1-4503-3468-6}
\showDOI{%
\url{https://doi.org/10.1145/2737924.2737964}}


\bibitem[\protect\citeauthoryear{v.~Gleissenthall, Bj{\o}rner, and
  Rybalchenko}{v.~Gleissenthall et~al\mbox{.}}{2016}]%
        {gleissenthall_cardinalities_2016}
\bibfield{author}{\bibinfo{person}{Klaus v. Gleissenthall},
  \bibinfo{person}{NikolajBj{\o}rner Bj{\o}rner}, {and} \bibinfo{person}{Andrey
  Rybalchenko}.} \bibinfo{year}{2016}\natexlab{}.
\newblock \showarticletitle{Cardinalities and {{Universal Quantifiers}} for
  {{Verifying Parameterized Systems}}}. In \bibinfo{booktitle}{{\em Proceedings
  of the 37th {{ACM SIGPLAN Conference}} on {{Programming Language Design}} and
  {{Implementation}}}} {\em (\bibinfo{series}{PLDI '16})}.
  \bibinfo{publisher}{{ACM}}, \bibinfo{pages}{599--613}.
\newblock
\showISBNx{978-1-4503-4261-2}


\bibitem[\protect\citeauthoryear{Weidenbach, Dimova, Fietzke, Kumar, Suda, and
  Wischnewski}{Weidenbach et~al\mbox{.}}{2009}]%
        {spass}
\bibfield{author}{\bibinfo{person}{Christoph Weidenbach},
  \bibinfo{person}{Dilyana Dimova}, \bibinfo{person}{Arnaud Fietzke},
  \bibinfo{person}{Rohit Kumar}, \bibinfo{person}{Martin Suda}, {and}
  \bibinfo{person}{Patrick Wischnewski}.} \bibinfo{year}{2009}\natexlab{}.
\newblock \showarticletitle{{SPASS} Version 3.5}. In \bibinfo{booktitle}{{\em
  Automated Deduction - CADE-22, 22nd International Conference on Automated
  Deduction, Montreal, Canada, August 2-7, 2009. Proceedings}}.
  \bibinfo{pages}{140--145}.
\newblock


\bibitem[\protect\citeauthoryear{Wilcox, Woos, Panchekha, Tatlock, Wang, Ernst,
  and Anderson}{Wilcox et~al\mbox{.}}{2015}]%
        {DBLP:conf/pldi/WilcoxWPTWEA15}
\bibfield{author}{\bibinfo{person}{James~R. Wilcox}, \bibinfo{person}{Doug
  Woos}, \bibinfo{person}{Pavel Panchekha}, \bibinfo{person}{Zachary Tatlock},
  \bibinfo{person}{Xi Wang}, \bibinfo{person}{Michael~D. Ernst}, {and}
  \bibinfo{person}{Thomas~E. Anderson}.} \bibinfo{year}{2015}\natexlab{}.
\newblock \showarticletitle{Verdi: a framework for implementing and formally
  verifying distributed systems}. In \bibinfo{booktitle}{{\em Proceedings of
  the 36th {ACM} {SIGPLAN} Conference on Programming Language Design and
  Implementation, Portland, OR, USA, June 15-17, 2015}}.
  \bibinfo{pages}{357--368}.
\newblock


\bibitem[\protect\citeauthoryear{Yu, Manolios, and Lamport}{Yu
  et~al\mbox{.}}{1999}]%
        {tlc}
\bibfield{author}{\bibinfo{person}{Yuan Yu}, \bibinfo{person}{Panagiotis
  Manolios}, {and} \bibinfo{person}{Leslie Lamport}.}
  \bibinfo{year}{1999}\natexlab{}.
\newblock \showarticletitle{Model Checking TLA\({}^{\mbox{+}}\)
  Specifications}. In \bibinfo{booktitle}{{\em Correct Hardware Design and
  Verification Methods, 10th {IFIP} {WG} 10.5 Advanced Research Working
  Conference, {CHARME} '99, Bad Herrenalb, Germany, September 27-29, 1999,
  Proceedings}} {\em (\bibinfo{series}{Lecture Notes in Computer Science})},
  \bibfield{editor}{\bibinfo{person}{Laurence Pierre} {and}
  \bibinfo{person}{Thomas Kropf}} (Eds.), Vol.~\bibinfo{volume}{1703}.
  \bibinfo{publisher}{Springer}, \bibinfo{pages}{54--66}.
\newblock
\showISBNx{3-540-66559-5}
\showDOI{%
\url{https://doi.org/10.1007/3-540-48153-2_6}}


\end{thebibliography}

\clearpage
\appendix
\section{Paxos Variants}
\label{sec:paxos-variants-long}

\subsection{Vertical Paxos}
\label{sec:paxos-vertical}

Vertical Paxos~\cite{lamport_vertical_2009} is a variant of Paxos whose set of
participating nodes and quorums (called a configuration) can be changed
dynamically by an external reconfiguration master (master for short) assumed
reliable.  Vertical Paxos is important in practice because reconfiguration
allows to replace failed nodes to achieve long-term reliability. Moreover, by
appropriately choosing the quorums, Vertical Paxos can survive the failure of
all but one node, compared to at most $\floor{n/2}$ nodes for Paxos, where $n$
is the total number of nodes.  Finally, the role of master in Vertical Paxos is
limited to managing configurations and requires few resources.  A reliable
master can therefore be implemented cheaply using an independent replicated
state machine.  Vertical Paxos has two variants, Vertical Paxos I and Vertical
Paxos II.  We consider only Vertical Paxos I.

In Vertical Paxos I, a node starts a round $r$ only when directed to by the
master through a \emph{configure-round} message that includes the configuration
to use for round $r$.  This allows the master to change the configuration by
directing a node to start a new round.  Since quorums are now different
depending on the round, and quorums of different rounds may not intersect at
all, we require the owner of $r$ to send start-round messages and to collect
join-acknowledgment messages from at least one quorum in each round $r'<r$ (we
say the the owner of round $r$ accesses all rounds $r'<r$).  However this may
be costly in practice or simply impossible as nodes from old configurations may
not stay reachable forever.

To limit the number of previous rounds that must be accessed when starting a
new round, the master tracks the rounds $r$ for which (a) no value is choosable
at any $r'<r$, or (b) a value $v$ has beed decided.  Such a round $r$ is said
complete.  Note that if $r$ is complete, then rounds lower than $r$ do not need
to be accessed anymore and can safely be retired, because a quorum of
join-acknowledgment messages from round $r$ suffices to compute a value that is
safe with respect to any decision that may have been made below round $r$.

To let the configuration master keep track of complete rounds, a node noticing
that a round has become complete when it receives a quorum of
join-acknowledgment messages notifies the master.  In turn, when directing a
node to start a new round, the master passes on the highest complete round it
knows of along with the new configuration.  A node starting a new round $r$
then only accesses the rounds above the highest complete round $r_c$ associated to the
configuration of $r$, starting from $r_c$.

\subsubsection{FOL model of Vertical Paxos}

We highlight the main changes compared to the FOL model of Paxos.

\paragraph{Axiomatization of configurations}
We model configurations in first-order logic by introducing a new sort
$\sconfig$, the relation $\rquorumin(c:\sconfig, q:\squorum)$, and the function
$\rcompleteof : \; \sconfig \to \sround$ where $\rquorumin(c,q)$ means that
quorum $q$ is a quorum of configuration $c$ and $\rcompleteof(c)=r$ means that
$r$ is the complete round that the master associated to configuration $c$.
Note that the function $\rcompleteof$ never needs to be updated: when the
master must associate a particular complete round $r$ to a configuration $c$,
we simply pick $c$ such that $\rcompleteof(c)=r$ holds.

We change the intersection property of quorums to only require that quorums of
the same configuration intersect, with the following axiom:
\begin{equation}
\forall c, q_1,q_2. \; \rquorumin(c,q_1) \land \rquorumin(c,q_2) \rightarrow \exists n. \; \rmember(n,q_1) \land \rmember(n, q_2)
\end{equation}

\paragraph{State of the master}
The state of the master consists of a single individual $\rmastercomp :
\sround$ that represents the highest complete round that the master knows of.
The individual $\rmastercomp$ is initializedd to the first round.

\paragraph{Messages}
In addition to the messages used in Paxos, there are two kinds of messages exchanged between the nodes and the master.
First, the master sends messages to instruct a round owner to start a round
with a prescribed configuration and complete round. The relation
$\rconfig(r:\sround, c:\sconfig)$ models configure-round messages from the
master, informing the owner of round $r$ that it must start round $r$ with
configuration $c$ and complete round $\rcompleteof(c)$.  Second, nodes send
messages to the master to notify it that a round has become complete.
The relation $\rcompletemsg(r : \sround)$ models a node notifying the master that round $r$ has become complete.

Finally, in contrast to Paxos, start-round messages are not sent to all nodes, but only to the nodes of the
configuration of particular rounds.  Therefore we make the $\ronea(r:\sround,
  \text{rp}:\sround)$ relation binary to model start-round messages from round
  $r$ to the nodes in the configuration of round $\text{rp}$.

\paragraph{Master actions}
The action $\aconfig(r:\sround, c:\sconfig)$ models the master sending a
message to the owner of round $r$ to inform it that it must start round $r$
with configuration $c$ and complete round $\rcompleteof(c)$.  When this action
happens, we say that the master configures round $r$.  The master can perform
this action when round $r$ has not been configured yet and is strictly greater
than the highest complete round known to the master.  Moreover, the master
picks a configuration $c$ whose complete round equals its highest known
complete round.

The action $\amark(r:\sround)$ models the master receiving a notification from
a node that round $r$ has become complete. The master then updates its estimate
of the highest complete round by updating the relation $\rmastercomp$.

\paragraph{Node actions}

As in Paxos, nodes perform five types of actions: $\asendonea$, $\ajoinround$,
$\apropose$, $\acastvote$, and $\adecide$.  The major changes compared to Paxos
is how the owner of a round starts a new round and determines what proposal to
make, i.e. the actions $\asendonea$, $\ajoinround$, and $\apropose$.

In the action $\asendonea(r, c, cr)$ the owner of round $r$ starts $r$ upon
receiving a configure-round message from the master instructing it to start $r$
with a configuration $c$ and a complete round $\rcompleteof(c)$.  The owner of
round $r$ broadcasts one join-round message to each configuration corresponding
to a round $r'$ such that $\rcompleteof(c)\leq r' < r$. 

The action $\ajoinround(n, r, \text{rp})$ differs from Paxos in that it models
node $n$ responding to a start-round message addressed specifically to round
$\text{rp}$.  Node $n$ responds with a join-acknowledgment message
$\ronebmaxvote(n,r,\text{rp},v)$ indicating that $n$ voted for $v$ in
$\text{rp}$.  If $n$ did not vote in $\text{rp}$, the value component of the
join-acknowledgment message is set to $\none$.  Note that the meaning of a
join-acknowledgment message is different from Paxos; in Paxos, a
$\ronebmaxvote(n,r,\text{rmax},v)$ message contains the highest round rmax lower than r
in which $n$ voted.

In the action $\apropose$, the owner of round $r$ makes a proposal to the
configuration $c$ associated to $r$ after receiving enough join-acknowledgment
messages.  The action requires that for each round $r'$ such that
$\rcompleteof(c)\leq r' < r$, the owner of $r$ received join-acknowledgment
messages from a quorum of the configuration of $r'$ (which also requires that
those rounds have a known configuration).  This is modeled in first-order logic
by fixing a function $\relation{\rquorumof} : \; \sround \to \squorum$
and assuming that for every round $r'$ such that $\rcompleteof(c)\leq r' < r$,
the quorum $\relation{\rquorumof}(r')$ is a quorum of the configuration
of $r'$.  
To determine its proposal, the owner of round $r$ checks whether a vote was reported
between $\rcompleteof(c)$ and $r$.  If no vote was reported (i.e. all
$\ronebmaxvote(n,r,r',v)$ have $v=\none$), then the owner of $r$ proposes an
arbitrary value and notifies the master that $r$ is complete, as no value can
be choosable below $r$.  If a vote was reported, then the owner of round $r$
computes the maximal round $maxr$, with $\rcompleteof(c)\leq maxr< r$, in which
a vote $v$ was reported and proposes $v$ (in this case $v$ is may still be
choosable at a lower round, so $r$ is not complete).

The action $\acastvote(n, v, r)$, modeling node $n$ casting a vote for $v$ in
round $r$  remains the same as in Paxos.  The action $\adecide(n, r, c, v, q)$ 
differs from Paxos in that it must take into account the configuration of a
round to determine the quorums that must vote for a value $v$ for $v$ to become
decided.  Moreover, a node making a decision notifies the master that round $r$
has become complete, modeled by updating the $\rcompletemsg$ relation.

The full FOL model of Vertical Paxos appears in~\cref{fig:vertical-paxos-fol-rml}.

\lstset{ %
  breakatwhitespace=false,         %
  keywordstyle=\bf,       %
  language=C,                 %
  otherkeywords={module,individual,init,action,returns,assert,assume,instantiate,isolate,mixin,before,relation,function,sort,variable,axiom,then,constant,let,*,local},           %
  numbers=left,                    %
  numbersep=5pt,                   %
  numberstyle=\tiny,               %
  rulecolor=\color{black},         %
  tabsize=8,	                   %
   columns=fullflexible,
}

\begin{figure}
  \begin{minipage}{\textwidth}
\begin{minipage}{.50\textwidth}
\begin{lstlisting}[
    %
    basicstyle=\scriptsize,%
    keepspaces=true,
    numbers=left,
    %
    xleftmargin=2em,
    numberstyle=\tiny,
    emph={
      %
      %
      %
      %
      %
      %
    },
    emphstyle={\bfseries},
    mathescape=true,
  ]
sort $\snode$, $\squorum$, $\sround$, $\svalue$, $\sconfig$

# in join_ack_msg, none indicates the absence of a vote.
individual $\none$: $\svalue$

relation $\leq$ : $\sround,\sround$
axiom total_order($\leq$)

relation $\rmember$ : $\snode,\squorum$
relation $\rquorumin(q:\squorum, c:\sconfig)$
axiom $\forall c:\sconfig, q_1, q_2. \rquorumin(q1,c) \land \rquorumin(q_2,c)$
  $\to \exists n:\snode. \; \rmember(n, q_1) \land \rmember(n, q_2)$

relation $\ronea$ : $\sround,\sround$
relation $\ronebmaxvote$ : $\snode,\sround,\sround,\svalue$
relation $\rproposal$ : $\sround,\svalue$
relation $\rvote$ : $\snode,\sround,\svalue$
relation $\rdecision$ : $\snode,\sround,\svalue$
relation $\rconfig$ : $\sround,\sconfig$
individual $\rcompletemsg$ : $\sround$
individual $\rcompleteof$ : $\sconfig \to \sround$
# highest $\text{complete}$ round known to the master (master state)
individual $\rmastercomp$ : $\sround$

init $\forall r_1,r_2. \; \neg\ronea(r_1,r_2)$
init $\forall n,r_1,r_2,v. \; \neg\ronebmaxvote(n,r_1,r_2,v)$
init $\forall r,v. \; \neg\rproposal(r,v)$
init $\forall n,r,v. \; \neg\rvote(n,r,v)$
init $\forall n,r,v. \; \neg\rdecision(n,r,v)$
init $\forall r,c. \; \neg\rconfig(r,c)$
init $\forall r . \; \neg\rcompletemsg(r)$
# $\rmastercomp$ is $\text{initially}$ set to the first round:
init $\forall r' . \; \rmastercomp(r)\leq r'$ 

# master $\text{actions}$:
action $\aconfig(r:\sround, c:\sconfig)$ {
  assume $\forall c. \; \neg\rconfig(r,c)$ # r is not configured
  assume $\rmastercomp \le r$
  assume $\rcompleteof(c) = \rmastercomp$
  $\rconfig(r,c)$ := true 
}
action $\amark(r:\sround)$ {
  assume $\rcompletemsg(r)$ # a node sent a "$\text{complete}$" message
  if ($\rmastercomp < r$) {
    $\rmastercomp$ := r
  }
}

# node $\text{actions}$:
action $\asendonea(r : \sround, c:\sconfig, \text{cr}:\sround)$ {
  # receive a configure-round message:
  assume $\rconfig(r,c)$ 
  # get the $\text{complete}$ round sent with the configuration:
  assume cr = $\rcompleteof(c)$ 
  $\ronea(r,R) :=$ 
    $\ronea(r,R) \lor (\text{cr} \leq R \land R<r)$
}

\end{lstlisting}
\end{minipage}%
\begin{minipage}{.50\textwidth}
\begin{lstlisting}[
    %
    basicstyle=\scriptsize,%
    keepspaces=true,
    numbers=left,
    %
    xleftmargin=2em,
    numberstyle=\tiny,
    firstnumber=58,
    emph={
      %
      %
      %
      %
      %
      %
    },
    emphstyle={\bfseries},
    mathescape=true,
  ]
action $\ajoinround(\text{n} : \snode ,\, \text{r} : \sround, \text{rp} : \sround)$ {
  assume $\ronea(r,\text{rp})$
  assume $\not\exists r',\text{rp},v . \; \ronebmaxvote(n,r',\text{rp},v) \land r<r'$
  local v:$\svalue$ {
    if ($\forall v'. \; \neg\rvote(n,\text{rp},v')$) {
      v := $\none$ } 
    else {
      assume $\rvote(r,\text{rp},v)$
    }
    $\ronebmaxvote(n,r,\text{rp},v)$ := true
  }
}
# $\text{local}$ to propose:
individual $\rquorumof$ : $\sround\to\squorum$ 
action $\apropose(\text{r} : \sround ,\, c:\sconfig, \, \text{cr}:\sround)$ {
  $\rquorumof$ := *
  assume $\rconfig(r,c)$
  assume $\rcompleteof(c) = \text{cr}$
  assume $\forall v . \; \neg\rproposal(r,v)$
  # rounds between the $\text{complete}$ round and r must be configured:
  assume $\forall r' . \;  \text{cr}\le r' < r \to \exists c . \rconfig(r',c)$
  # $\rquorumof(r')$ is a quorum of the configuration of $r'$:
  assume $\forall r' .  \; \text{cr}\le r' < r \land \rconfig(r',c) \to$
        $\rquorumin(\rquorumof(r'),c)$
  # got messages from all quorums between cr and r
  assume $\forall r',n . \; \text{cr}\le r' < r \land \rmember(n, \rquorumof(r')) \to$
        $\exists v . \; \ronebmaxvote(n,r,r',v)$
  local maxr:$\sround$, v:$\svalue$ {
    # find the maximal maximal vote in the quorums:
    maxr, v = max $\{ (r',v') \mid$ 
        $\exists n. cr\leq r' \land r'<r \land \rmember(n, \rquorumof(r'))$
        $\land \ronebmaxvote(n,r,r',v') \land v' \neq \none\}$
    if (v = $\none$) {
      v := * # set v to an arbitrary value different from none.
      assume v $\neq \none$
      $\rcompletemsg(r)$ := true # notify master that r is $\text{complete}$
    }
    $\rproposal(r, v)$ := true # propose value v
  }
action $\acastvote(\text{n} : \snode ,\, \text{r} : \sround ,\, \text{v} : \svalue)$ {
    assume $\text{v} \neq \none$
    assume $\rproposal(\text{r}, \text{v})$
    # never joined a higher round:
    assume $\neg \exists r',r'',v. \;  r' > \text{r} \land \ronebmaxvote(\text{n},r',r'',v)$
    $\rvote(\text{n}, \text{r}, \text{v})$ := true
}
action $\adecide(\text{n} : \snode, \text{r} : \sround ,\, c:\sconfig, \, \text{v} : \svalue ,\, \text{q} : \squorum)$ {
    assume $\text{v} \neq \none$
    assume $\rconfig(r,c)$
    assume $\rquorumin(q,c)$
    assume $\forall n. \; \rmember(n, \text{q}) \to \rvote(n, \text{r}, \text{v})$  # 2b from quorum q
    $\rdecision(\text{n}, \text{r}, \text{v})$ := true
    $\rcompletemsg(r)$ := true
}
\end{lstlisting}
\end{minipage}%
\end{minipage}%
\captionof{figure}{
Model of Vertical Paxos in many-sorted first-order logic.
}
\label{fig:vertical-paxos-fol-rml}
\end{figure}

\subsubsection{Inductive Invariant}

As for Flexible Paxos and Fast Paxos, the safety property that we prove is the
same as in Paxos.

The core property of the inductive invariant, on top of the relation between
proposals and choosable values established for Paxos, is that a round $r$
declared complete by a node is such that either (a) no value is choosable in
any round $r'<r$, or (b) there is a value $v$ decided in $r$.

As for Paxos, we start with rather mundane properties that are required for the inductiveness.
The first such property is that there is at most one proposal per round:
\begin{equation}
  \forall r:\sround, v_1,v_2:\svalue. \; \rproposal(r,v_1) \land \rproposal(r,v_2) \to v_1 = v_2.
\end{equation}
There is a most one configuration assigned to a round:
\begin{equation}
  \forall r:\sround, c_1,c_2:\sconfig. \; \rconfig(r,c_1) \land \rconfig(r, c_2) \to c_1 = c_2.
\end{equation}
A start-round message is sent only upon receiving a configure-round message from the master:
\begin{equation}
  \forall r_1,r_2:\sround . \; \ronea(r_1,r_2) \to \exists c:\sconfig . \; \rconfig(r_1, c).\label{eq:start-imp-exists-config}
\end{equation}
A start-round message starting round $r_1$ is sent only to round strictly lower than $r_1$:
\begin{equation}
  \forall r_1,r_2:\sround . \; \ronea(r_1,r_2) \to r_2<r_1.
\end{equation}
A node votes only for a proposal:
\begin{equation}
  \forall r:\sround, n, v:\svalue. \; \rvote(n,r,v) \to \rproposal(r,v)\\
\end{equation}
A proposal is made only when the configuration of lower rounds is known:
\begin{equation}
  \forall r_1,r_2 : \sround, v:\svalue . \; \rproposal(r_2,v) \land r_1<r_2 \to \exists c:\sconfig . \; \rconfig(r_1,c).\label{eq:vp-lower-than-proposal-config-exists}
\end{equation}
The special round $\none$ is never used for deciding a value:
\begin{equation}
  \forall r:\sround,n:\snode . \; \neg\rproposal(r,\none) \land \neg \rvote(n,r,\none) \land \neg\rdecision(n,r,\none).
\end{equation}
A join-acknowledgment message is sent only in response to a start-round message:
\begin{equation}
\forall n,r_1,r_2,v . \; \ronebmaxvote(n,r_1,r_2,v) \to \ronea(r_1,r_2).
\end{equation}
The join-acknowledgment messages faithfully represent votes:
\begin{align}
  & \forall n,r_1,r_2,v . \; \ronebmaxvote(n,r_1,r_2,\none) \to \neg\rvote(n,r_2,v)\\
  & \forall n,r_1,r_2,v . \; \ronebmaxvote(n,r_1,r_2,v) \land v \neq \none \to \rvote(n,r_2,v).
\end{align}
A configure-round message configuring round $r$ contains a complete round lower than $r$ and lower than the highest complete round known to the master:
\begin{equation}
  \forall r,c,\text{cr} . \; \rconfig(r,c) \land \rcompleteof(c,\text{cr}) \to \text{cr} \leq r \land \text{cr} \leq \rmastercomp.
\end{equation}
Any round lower than a round appearing in a complete messages has been configured:
\begin{equation}
  \forall r_1,r_2 . \; \rcompletemsg(r_2) \land r_1\leq r_2 \to \exists c . \; \rconfig(r_1,c).\label{eq:vp-lower-complte-exists-config}
\end{equation}
The highest complete round known to the master or the complete round that the master assigns to a configuration has been marked complete by a node or is the first round:
\begin{equation}
  \begin{split}
    \forall r_1,r_2,r_3,c . & \; (r_2 = \rmastercomp \lor (\rconfig(r_3,c) \land \rcompleteof(c) = r_2) \land r_1<r_2 \to \\
                            &\rcompletemsg(r_2).
  \end{split}
\end{equation}
A decisions comes from a quorum of votes from a configured round:
\begin{equation}
  \begin{split}
    \forall r,v .\; &(\exists n .\; \rdecision(n,r,v)) \to \\
                    &\exists c:\sconfig, q . \; \rconfig(r,c) \land \rquorumin(q,c) \land \\
                    &\qquad (\forall n . \rmember(n,q) \to \rvote(n,rv)).\label{eq:decision-from-quorum-config}
  \end{split}
\end{equation}

Note that the number of the invariants above may seem overwhelming, but those invariants are easy to infer from the counterexamples to induction displayed graphically by IVy when they are missing. 

Finally, the two crucial invariants of Vertical Paxos I relate, first, proposals to choosable values (similarly to Paxos), and, second, relate rounds declared complete to choosable values:
\begin{align}
  \begin{split}
    \forall r_1,r_2,& v_1,v_2,q,c . \\
                    &r_1<r_2 \land \rproposal(r_2,v_2) \land v_1\neq v_2 \land \rconfig(r_1,c) \land \rquorumin(q,c) \to\\
                    & \exists n, r_3,r_4,v .\; \rmember(n,q) \land r_1<r_3 \land \ronebmaxvote(n,r_3,r_4,v) \land \neg\rvote(n,r_1,v_1)
  \end{split}\\
  \begin{split}
    \forall r_1, r_2, &c, q . \\
                      &\rcompletemsg(r_2) \land r_1<r_2 \land \rconfig(r_1,c) \land \rquorumin(q,c))\land\\
                      &\neg(\exists n,r_3,r_4,v . \rmember(n,q) \land r_1<r_3 \land \ronebmaxvote(n,r_3,r_4,v) \land \neg\rvote(n,r,v))\to\\
                      &\quad\qquad \exists n. \rdecision(n,r,v)
  \end{split}
\end{align}

\subsubsection{Transformation to EPR}

As in Paxos, we start by introducing the derived relation $\rleftround(n,r)$ with representation invariant 
\begin{equation}
  \forall n, r . \; \rleftround(n, r) \leftrightarrow \exists r'>r,\text{rp},v . \; \ronebmaxvote(n, r', \text{rp}, v)
\end{equation}
and we rewrite the $\ajoinround$ and $\apropose$ accordingly.

Then, we notice that the function $\rcompleteof : \; \sconfig \to \sround$, together with~\cref{eq:vp-lower-than-proposal-config-exists,eq:vp-lower-complte-exists-config,eq:decision-from-quorum-config,eq:start-imp-exists-config}, creates a cycle in the quantifier alternation graph involving the sort $\sconfig$ and $\sround$.
We eliminate this cycle by introducing the derived relation $\rcompleteof(c:\sconfig, r:\sround)$ (here we overload $\rcompleteof$ to represent both the function and the relation) with representation invariant 
\begin{equation}
  \forall c, r . \; \rcompleteof(c) = r \leftrightarrow \rcompleteof(c,r)
\end{equation}
We rewrite  the $\aconfig$, $\asendonea$, and $\apropose$ actions to use the $\rcompleteof$ in its relation form instead of the function.
The function $\rcompleteof$ is now unused in the model and we remove it, thereby eliminating the cycle in the quantifier alternation graph involving the sorts $\sconfig$ and $\sround$.

After this step we observe that the quantifier alternation graph is acyclic, and IVy successfully verifies that the invariant is inductive. 

Surprisingly, the transformation of the Vertical Paxos I model to EPR is simpler than the transformation of the Paxos model to EPR because we do not need to introduce the derived relation $\ronebproj$ and, in consequence, to rewrite the propose action to use $\rvote$ instead of $\ronebmaxvote$.
The derived relation $\ronebproj$ is not needed because, compared to Paxos, the assume condition of the $\apropose$ action contains the formula 
\begin{equation}
  \forall r' n . \; \text{cr}\leq r' \land r'<r \land \rmember(n, \relation{\rquorumof(c))} \to \exists v:\svalue . \ronebmaxvote(n,r,r',v)\label{eq:vp-nice-ae}
\end{equation}
instead of the formula 
\begin{equation}
\forall n : \snode. \; \rmember(n, q) \to \exists r' : \sround, \, v:\svalue. \;  \ronebmaxvote(n,r,r',v).\label{eq:vp-not-nice-ae}
\end{equation}
Notice how~\cref{eq:vp-nice-ae} introduces an edge in the quantifier alternation graph from sort $\snode$ to $\svalue$ only, whereas~\cref{eq:vp-not-nice-ae} introduces an edge from sort $\sround$ to sorts $\svalue$ and $\sround$.
In combination with the conjunct of the inductive invariant that expresses the relation between choosable values and proposals, \cref{eq:vp-not-nice-ae} introduces a cycle in the quantifier alternation graph, whereas~\cref{eq:vp-nice-ae} does not. 

The full model of Vertical Paxos in EPR appears in~\cref{fig:vertical-paxos-epr-rml}

\lstset{ %
  breakatwhitespace=false,         %
  keywordstyle=\bf,       %
  language=C,                 %
  otherkeywords={module,individual,init,action,returns,assert,assume,instantiate,isolate,mixin,before,relation,function,sort,variable,axiom,then,constant,let,*,local},           %
  numbers=left,                    %
  numbersep=5pt,                   %
  numberstyle=\tiny,               %
  rulecolor=\color{black},         %
  tabsize=8,	                   %
   columns=fullflexible,
}

\begin{figure}
  \begin{minipage}{\textwidth}
\begin{minipage}{.50\textwidth}
\begin{lstlisting}[
    %
    basicstyle=\scriptsize,%
    keepspaces=true,
    numbers=left,
    %
    xleftmargin=2em,
    numberstyle=\tiny,
    emph={
      %
      %
      %
      %
      %
      %
    },
    emphstyle={\bfseries},
    mathescape=true,
  ]
sort $\snode$, $\squorum$, $\sround$, $\svalue$, $\sconfig$

# in join_ack_msg, none indicates the absence of a vote.
individual $\none$: $\svalue$

relation $\leq$ : $\sround,\sround$
axiom total_order($\leq$)

relation $\rmember$ : $\snode,\squorum$
relation $\rquorumin(q:\squorum, c:\sconfig)$
axiom $\forall c:\sconfig, q_1, q_2. \rquorumin(q1,c) \land \rquorumin(q_2,c)$
  $\to \exists n:\snode. \rmember(n, q_1) \land \rmember(n, q_2)$

relation $\ronea$ : $\sround,\sround$
relation $\ronebmaxvote$ : $\snode,\sround,\sround,\svalue$
relation $\rproposal$ : $\sround,\svalue$
relation $\rvote$ : $\snode,\sround,\svalue$
relation $\rdecision$ : $\snode,\sround,\svalue$
relation $\rconfig$ : $\sround,\sconfig$
individual $\rcompletemsg$ : $\sround$
relation $\rcompleteof$ : $\sconfig, \sround$ # replaced the $\text{function}$
# immutable, so we can use $\text{axiom}$:
axiom $\forall c,r_1,r_2.\; \rcompleteof(c,r_1) \land \rcompleteof(c,r_2) \to $
  $r_1 = r_2$ 
# highest $\text{complete}$ round known to the master (master state)
individual $\rmastercomp$ : $\sround$
relation $\rleftround$ : $\snode,\sround$

init $\forall r_1,r_2. \; \neg\ronea(r_1,r_2)$
init $\forall n,r_1,r_2,v. \; \neg\ronebmaxvote(n,r_1,r_2,v)$
init $\forall r,v. \; \neg\rproposal(r,v)$
init $\forall n,r,v. \; \neg\rvote(n,r,v)$
init $\forall n,r,v. \; \neg\rdecision(n,r,v)$
init $\forall r,c. \; \neg\rconfig(r,c)$
init $\forall r . \; \neg\rcompletemsg(r)$
# $\rmastercomp$ is $\text{initially}$ set to the first round:
init $\forall r' . \; \rmastercomp(r)\leq r'$ 

# master $\text{actions}$:
action $\aconfig(r:\sround, c:\sconfig)$ {
  assume $\forall c. \; \neg\rconfig(r,c)$ # r is not configured
  assume $\rmastercomp \le r$
  assume $\rcompleteof(c) = \rmastercomp$
  $\rconfig(r,c)$ := true 
}
action $\amark(r:\sround)$ {
  assume $\rcompletemsg(r)$ # a node sent a "$\text{complete}$" message
  if ($\rmastercomp < r$) {
    $\rmastercomp$ := r
  }
}

# node $\text{actions}$:
action $\asendonea(r : \sround, c:\sconfig, \text{cr}:\sround)$ {
  # receive a configure-round message:
  assume $\rconfig(r,c)$ 
  # get the $\text{complete}$ round sent with the configuration:
  assume $\rcompleteof(c,\text{cr})$ 
  $\ronea(r,R) :=$
    $\ronea(r,R) \lor (\text{cr} \leq R \land R<r)$
}

\end{lstlisting}
\end{minipage}%
\begin{minipage}{.50\textwidth}
\begin{lstlisting}[
    %
    basicstyle=\scriptsize,%
    keepspaces=true,
    numbers=left,
    %
    xleftmargin=2em,
    numberstyle=\tiny,
    firstnumber=62,
    emph={
      %
      %
      %
      %
      %
      %
    },
    emphstyle={\bfseries},
    mathescape=true,
  ]
action $\ajoinround(\text{n} : \snode ,\, \text{r} : \sround, \text{rp} : \sround)$ {
  assume $\ronea(r,\text{rp})$
  assume $\neg \rleftround(n,r)$ # rewritten
  local v:$\svalue$ {
    if ($\forall v'. \; \neg\rvote(n,\text{rp},v')$) {
      v := $\none$ } 
    else {
      assume $\rvote(r,\text{rp},v)$
    }
    $\ronebmaxvote(n,r,\text{rp},v)$ := true
    $\rleftround(\text{n},R)$ := $\rleftround(\text{n},R) \lor R < \text{r}$
  }
}
# $\text{local}$ to propose:
individual $\rquorumof$ : $\sround\to\squorum$ 
action $\apropose(\text{r} : \sround ,\, c:\sconfig, \, \text{cr}:\sround)$ {
  $\rquorumof$ := *
  assume $\rconfig(r,c)$
  assume $\rcompleteof(c,\text{cr})$
  assume $\forall v . \; \neg\rproposal(r,v)$
  # rounds between the $\text{complete}$ round and r must be configured:
  assume $\forall r' . \;  \text{cr}\le r' < r \to \exists c . \rconfig(r',c)$
  # $\rquorumof(r')$ is a quorum of the configuration of $r'$:
  assume $\forall r' .  \; \text{cr}\le r' < r \land \rconfig(r',c) \to$
        $\rquorumin(\rquorumof(r'),c)$
  # got messages from all quorums between cr and r
  assume $\forall r',n . \; \text{cr}\le r' < r \land \rmember(n, \rquorumof(r')) \to$
        $\exists v . \; \ronebmaxvote(n,r,r',v)$
  local maxr:$\sround$, v:$\svalue$ {
    # find the maximal maximal vote in the quorums:
    assume maxr, v = 
      max $\{ (r',v') \mid \exists n. cr\leq r' \land r'<r \land \rmember(n, \rquorumof(r'))$
        $\land \ronebmaxvote(n,r,r',v') \land v' \neq \none\}$
    if (v = $\none$) {
      v := * # set v to an arbitrary value different from none.
      assume v $\neq \none$
      $\rcompletemsg(r)$ := true # notify master that r is $\text{complete}$
    }
    $\rproposal(r, v)$ := true # propose value v
  }
action $\acastvote(\text{n} : \snode ,\, \text{r} : \sround ,\, \text{v} : \svalue)$ {
    assume $\text{v} \neq \none$
    assume $\rproposal(\text{r}, \text{v})$
    # never joined a higher round:
    assume $\neg \rleftround(\text{n},\text{r})$ # rewritten
    $\rvote(\text{n}, \text{r}, \text{v})$ := true
}
action $\adecide(\text{n} : \snode, \text{r} : \sround ,\, c:\sconfig, \, \text{v} : \svalue ,\, \text{q} : \squorum)$ {
    assume $\text{v} \neq \none$
    assume $\rconfig(r,c)$
    assume $\rquorumin(q,c)$
    assume $\forall n. \; \rmember(n, \text{q}) \to \rvote(n, \text{r}, \text{v})$  # 2b from quorum q
    $\rdecision(\text{n}, \text{r}, \text{v})$ := true
    $\rcompletemsg(r)$ := true
}
\end{lstlisting}
\end{minipage}%
\captionof{figure}{
Model of Vertical Paxos in EPR.
}
\label{fig:vertical-paxos-epr-rml}
\end{minipage}%
\end{figure}

\subsection{Fast Paxos}
\label{sec:fast-paxos}

When a unique node starts a round, a value sent by a client to that node is decided by Paxos in at most 3 times the worst-case message latency (one message to deliver the proposal to the owner of the round, and one round trip for the owner of the round to complete phase 2).
Fast Paxos~\cite{lamport_fast_2006} reduces this latency to twice the worst-case message latency under the assumption that messages broadcast by the nodes are received in the same order by all nodes (a realistic assumption in some settings, e.g.,  in an Ethernet network).
When this assumption is violated, the cost of Fast Paxos increases depending on the conflict-recovery mechanism employed.

In Fast Paxos, rounds are split into a set of fast rounds and a set of classic rounds (e.g. even rounds are fast and odd ones are classic).
In a classic round, nodes vote for the proposal that the owner of the round sends and a value is decided when a quorum of nodes vote for it, as in Paxos.
However, in a fast round $r$, the owner can send an ``any'' message instead of a proposal, and a node receiving it is allowed to vote for any value of its choice in $r$ (but it cannot change its mind after having voted).
This allows clients to broadcast their proposals without first sending it to the leader, saving on message delay.
But, as a result, different values may be voted for in the same (fast) round, a situation that cannot arise in Paxos.

Now that different values may be voted for in the same (fast) round, the rule used in the propose action of Paxos to determine a value to propose does not work anymore, as there may be different votes in the highest reported round.
Remember that this rule ensured that the owner $p$ of round $r$ would not miss any value choosable in a lower round $r'<r$.
With different values being voted for in fast round $r'$, how is node $p$ to determine which value is choosable in $r'$?
This problem is solved by requiring that any 2 fast quorums and one classic quorum have a common node, and by modifying the way $p$ determines which value is choosable at $r'$, as follows.

As in Paxos, the owner $p$ waits for a classic quorum $q$ of nodes to have joined its round $r$.
If the highest reported round $maxr$ is a fast round and there are different votes reported in $maxr$, then $p$ checks whether there exists a fast quorum $f$ such that all nodes in $q\cap f$ voted for the same value $v$ in $maxr$; if there is such a fast quorum $f$ and value $v$, then $p$ proposes $v$, and otherwise it can propose any value.
By the intersection property of quorums there can be at most one value $v$ satisfying those conditions.
Moreover, if $v$ is choosable at $r'$, then there will be a fast quorum $f$ such that all nodes in $q\cap f$ voted for $v$ at $maxr$.
If there is only a single value $v$ reported voted for in $maxr$, then $p$ proposes $v$, as in Paxos.

In Paxos, a round may not decide a value if its owner is not able to contact a quorum of nodes before they leave the round (e.g. because the owner crashed, is slow, or because of message losses in the network).
In Fast Paxos there is one more cause for a round not deciding a value: a fast round $r$ may not decide a value if the votes cast in $r$ are such that, no matter what new votes are cast from this point on, no value can be voted for by a fast quorum (e.g. when every nodes voted for a different value).
In this case we say that a conflict has occured.
Fast Paxos has three different conflict-recovery mechanisms: starting a new round, coordinated recovery, and uncoordinated recovery.
Our model does not include coordinated or uncoordinated recovery, but starting a new round is of course part of the model.

\subsubsection{FOL model of Fast Paxos}

The FOL model of Fast Paxos appears in~\cref{fig:fast-paxos-fol}.
We now highlight the main changes compared to the FOL model of Paxos.

\paragraph{Quorums}
We axiomatize the properties of fast quorums and classic quorums in first-order
logic by defining two different sorts $\scquorum$ and $\sfquorum$ for classic
and fast quorums, and a separate membership relation for each.  The
intersection properties of quorums are expressed as follows:
\begin{align}
& \forall q_1,q_2:\scquorum. \exists n:\snode. \rcmember(n,q_1) \land \rcmember(n, q_2) \\
& \forall q:\scquorum, f_1,f_2:\sfquorum.
  \exists n:\snode. \rcmember(n,q) \land \rfmember(n, f_1) \land \rfmember(n, f_2)
\end{align}

\paragraph{New relations}
To identify fast rounds we add a unary relation $\rfast : \sround$ which contains all fast rounds.
We add a relation $\rany$ to model the ``any'' messages sent by the owners of fast rounds.

\paragraph{Actions}

The $\asendonea$, $\acastvote$, and $\ajoinround$ actions remain the same as in
Paxos.  The $\apropose$ action is modified to reflect the new rule that the
owner of a round $r$ uses to determine what command to propose.
We start by assuming the owner has received join-acknowledgment
messages from a quorum $q$ (\cref{line:fp-propose-assume-ae}), and we
compute the maximal reported round $maxr$ and pick a vote $v$ reported at $maxr$ (\cref{line:fp-propose-max}).  Note that there may be different
votes reported in $maxr$ if $maxr$ is a fast round.

If at least one vote was reported ($maxr\neq \bot$), then the owner propose the value $v'$ chosen as follows (\cref{line:fp-propose-choosable}):
if there exists a fast quorum $f$ such that all nodes in $f\cap q$ reported voting for $v'$ at $maxr$, then the owner proposes $v'$.
Otherwise, if no such fast quorum exists, it proposes $v$ (as defined above).
Formally, we assume
\begin{equation}
  \begin{split}
    &(\exists f.\;  \forall n.\;  \rfmember(n,f) \land \rcmember(n,q) \rightarrow \ronebmaxvote(n,r,maxr,v')) \\
    &\lor \; (v' = v \land \forall v'', f' .\;  \exists n.\;  \rfmember(n,f') \land \rcmember(n,q) \land \neg\ronebmaxvote(n,r,maxr,v''))
\end{split}
\end{equation}
If no vote was reported ($maxr=\bot$) then, if the round being started is a fast round then the owner sends an ``any'' message, and otherwise the owner proposes an arbitrary value.

Finally, we replace the $\adecide$ action by two actions $\acdecide$ and
$\afdecide$ that update the $\rdecision$ relation when a classic (for
$\acdecide$) or fast (for $\afdecide$) quorum has voted for the same value in
the same round.

\subsubsection{Inductive Invariant}

The safety property we prove is the same as in Paxos, namely that all decisions are for the same value regardless of the node making them and of the round in which they are made:
\begin{equation}
\forall n_1,n_2,r_1,r_2,v . \; \rdecision(n_1,r_1,v_1) \land \rdecision(n_2, r_2,v_2) \to v_1 = v_2
\end{equation}

The inductive invariant is similar to the one of Paxos with cases to distinguish classic and fast rounds.
The main addition is that no value can be choosable at a round $r'<r$ when there is an ``any'' message in $r$.

We start by expressing rather mundane facts about the algorithm, but which are necessary for inductiveness. 
There is at most one proposal per round:
\begin{equation}
  \forall r:\sround, v_1,v_2:\svalue. \; \neg\rfast(r) \land \rproposal(r,v_1) \land \rproposal(r,v_2) \to v_1 = v_2.\label{eq:fp1}
\end{equation}
In a classic round, nodes voted for a value $v$ only if $v$ was proposed: 
\begin{equation}
  \forall r:\sround, n, v:\svalue. \; \neg\rfast(r) \land \rvote(n,r,v) \to \rproposal(r,v).\label{eq:fp2}
\end{equation}
An ``any'' message can occur only in a fast round: 
\begin{equation}
  \forall r:\sround . \; \rany(r) \to \rfast(r).\label{eq:fp3}
\end{equation}
In a fast round, a node votes for a value $v$ only if $v$ was proposed or if an ``any'' message was sent in the round: 
\begin{equation}
  \forall r:\sround, n, v . \; \rfast(r) \land \rvote(n,r,v) \to (\rproposal(r,v) \lor \rany(r)).\label{eq:fp4}
\end{equation}
There cannot be both a proposal and an any message in the same round: 
\begin{equation}
  \forall r:\sround, v . \; \neg(\rproposal(r,v) \land \rany(r)).\label{eq:fp5}
\end{equation}
A node votes only once per round: 
\begin{equation}
\forall n:\snode, r:\sround, v_1,v_2:\svalue. \; \rvote(n,r,v_1)\land \rvote(n,r,v_2) \to v_1=v_2.\label{eq:fp6}
\end{equation}
There is no vote in the round $\bot$: 
\begin{equation}
\forall n:\snode, v:\svalue. \; \neg \rvote(n,\none,v).\label{eq:fp7}
\end{equation}
Decisions come from fast quorums in fast rounds and classic quorums in classic rounds:
\begin{align}
  \begin{split} 
    \forall r:\sround, v:\svalue. \; &\neg \rfast(r) \land (\exists n:\snode. \rdecision(n,r,v)) \to \\
                                     &\exists q:\scquorum. \forall n:\snode. \rcmember(n, q) \to \rvote(n,r,v)\label{eq:fp8}
  \end{split}\\
  \begin{split}
    \forall r:\sround, v:\svalue. \; &\rfast(r) \land (\exists n:\snode. \rdecision(n,r,v)) \to \\
                                     &\exists f:\sfquorum. \forall n:\snode. \rfmember(n, f) \to \rvote(n,r,v).\label{eq:fp9}
  \end{split} 
\end{align}

Then we express the fact that join-acknowledgment messages faithfully represent the node votes (exactly as in Paxos).
\begin{align}
  \forall n:\snode,\, &r,r':\sround,\, v,v':\svalue. \; \ronebmaxvote(n,r,\none,v) \land r' < r \to \neg \rvote(n,r',v')\label{eq:fp10}  \\
  \forall n:\snode,\, &r,r':\sround,\, v:\svalue. \; \ronebmaxvote(n,r,r',v) \land r' \neq \none \to  \\
                                                 &r' < r \land \rvote(n,r',v)\label{eq:fp11} \\
  \begin{split}
    \forall n:\snode,\, &r,r',r'':\sround,\, v,v':\svalue. \; \ronebmaxvote(n,r,r',v) \land r' \neq \none \land r' < r'' < r \to \\
                                                          &\neg \rvote(n,r'',v').\label{eq:fp12}
  \end{split}
\end{align}

Finally, we express that if $v$ is choosable at round $r$, then only $v$ can be proposed at a round $r'>r$, and that there cannot be an ``any'' message at a round $r'>r$.
We differentiate the case of a value choosable in a classic round and in a fast round.
\begin{equation}
\begin{split}
\forall & r_1,r_2,v_1,v_2, q:\scquorum. \;
\neg\rfast(r_1) \land ((\rproposal(r_2,v_2)\land v_1 \neq v_2) \lor \rany(r_2)) \land r_1 < r_2 \to \\
&\exists n:\snode, r', r'':\sround, v:\svalue. \; \rcmember(n,q) \\
&\qquad \land \neg \rvote(n,r_1,v_1) \land r' > r_1 \land \ronebmaxvote(n,r',r'',v)
\end{split}
\end{equation}
\begin{equation}
\begin{split}
\forall & r_1,r_2,v_1,v_2, f:\sfquorum. \;
\rfast(r_1) \land ((\rproposal(r_2,v_2)\land v_1 \neq v_2) \lor \rany(r_2)) \land r_1 < r_2 \to \\
&\exists n:\snode, r', r'':\sround, v:\svalue. \; \rfmember(n,f) \\
&\qquad \land \neg \rvote(n,r_1,v_1) \land
r' > r_1 \land \ronebmaxvote(n,r',r'',v)
\end{split}
\end{equation}

\subsubsection{Transformation to EPR}

To transform the Fast Paxos model to EPR we introduce the same derived relations as in Paxos (i.e. $\rleftround$ and $\ronebproj$) and we rewrite as explained in~\cref{sec:paxos-epr}.
However the verification of the second rewrite step is not as simple as in Paxos.
This step consists in rewriting the $\apropose$ action by considering directly votes instead of join-acknowledgment messages. In Paxos, we verified the rewrite using the auxiliary invariant $INV_{aux}$, and the verification condition is in EPR.
We employ the same method for Fast Paxos, using~\crefrange{eq:fp1}{eq:fp12} as auxiliary invariant.
However, for Fast Paxos, the verification condition of the rewrite does not fall in EPR when done naively.

Verifying the rewrite of the statement at line \cref{line:fp-propose-max}, from
\begin{equation}
    \max \{ (r',v') \mid \exists n. \; \rcmember(n, q)
    \land \ronebmaxvote(n,r,r',v') \land r' \neq \none \}
\end{equation}
to
\begin{equation}
    \max \{ (r',v') \mid \exists n. \; \rcmember(n, q)
    \land \rvote(n,r,v') \land r' \neq \none \land r' < r \}
\end{equation}
is exactly as in Paxos and poses no problem.

However, we also need to verify the rewrite of the assume statement at~\cref{line:fp-propose-choosable} from assuming
\begin{equation}
\begin{split}
  (\exists & f.\;  \forall n.\;  \rfmember(n,f) \land \rcmember(n,q) \rightarrow \ronebmaxvote(n,r,maxr,v'))\\
        &\lor \; (v' = v \land \forall v'', f' .\;  \exists n.\;  \rfmember(n,f') \land \rcmember(n,q)\land \neg\ronebmaxvote(n,r,maxr,v''))
\end{split}
\end{equation}
to assuming
\begin{equation}
\begin{split}
  (\exists & f.\;  \forall n.\;  \rfmember(n,f) \land \rcmember(n,q) \rightarrow \rvote(n,maxr,v'))\\
        &\lor \; (v' = v \land \forall v'', f' .\;  \exists n.\;  \rfmember(n,f') \land \rcmember(n,q)\land \neg\rvote(n,maxr,v''))
\end{split}
\end{equation}
With the assume statement at the beginning of the $\apropose$ action (without which the equivalence does not hold)
\begin{equation}
  \forall n. \; \rcmember(n, q) \to \exists r',v. \;  \ronebmaxvote(n,r,r',v)
\end{equation}
we get a verification condition that is not stratified: there is a cycle involving the sorts node and value in the quantifier alternation graph.
However, as explained in~\cref{sec:transformation-epr}, we observe that we only have to verify the rewrite of the sub-formula $\rcmember(n,q)\land \neg\ronebmaxvote(n,r,maxr,v'')$ to $\rcmember(n,q) \land \neg\rvote(n,maxr,v'')$.
Assuming the auxiliary invariant and the conditions of the assume statements of the $\apropose$ action before \cref{line:fp-propose-choosable}, this equivalence check falls in EPR and successfully verifies.

The EPR model of Fast Paxos appears in~\cref{fig:fast-paxos-epr}.

\lstset{ %
  breakatwhitespace=false,         %
  deletekeywords={for},            %
  keywordstyle=\bf,       %
  language=C,                 %
  otherkeywords={if,else,module,individual,init,action,returns,assert,assume,instantiate,isolate,mixin,before,relation,function,sort,variable,axiom,then,constant,let,*,local},           %
  numbers=left,                    %
  numbersep=5pt,                   %
  numberstyle=\tiny,               %
  rulecolor=\color{black},         %
  tabsize=8,	                   %
   columns=fullflexible,
}

\begin{figure}
  \begin{minipage}{\textwidth}
    \begin{minipage}[t]{.50\textwidth}
\begin{lstlisting}[
    %
    basicstyle=\scriptsize,%
    keepspaces=true,
    numbers=left,
    %
    xleftmargin=2em,
    numberstyle=\tiny,
    emph={
      %
      %
      %
      %
      %
      %
    },
    emphstyle={\bfseries},
    mathescape=true,
  ]
sort $\snode$
sort $\sround$
sort $\svalue$
sort $\scquorum$ # classic quorums
sort $\sfquorum$ # fast quorums
relation $\rfast$ : $\sround$ # identifies fast rounds

relation $\leq$ : $\sround,\sround$
axiom total_order($\leq$)
constant $\none$ : $\sround$

relation $\rcmember$ : $\snode,\scquorum$
relation $\rfmember$ : $\snode,\sfquorum$
# classic quorums intersect
axiom $\forall q_1,q_2 : \scquorum .$
  $\exists n. \; \rcmember(n,q_1) \land \rcmember(n,q_2)$
# a classic quorum and a two fast quorums intersect
axiom $\forall q_1 : \scquorum . \forall q_2,q_3 : \sfquorum .$ 
  $\exists n. \; \rcmember(n,q_1) \land \rfmember(n,q_2) \land \rfmember(n,q_3)$

relation $\ronea$ : $\sround$
relation $\ronebmaxvote$ : $\snode,\sround,\sround,\svalue$
relation $\rproposal$ : $\sround,\svalue$
relation $\rvote$ : $\snode,\sround,\svalue$
relation $\rdecision$ : $\snode,\sround,\svalue$
relation $\rany$ : $\sround$ # the any messages

init $\forall r. \; \neg\ronea(r)$
init $\forall n,r_1,r_2,v. \; \neg\ronebmaxvote(n,r_1,r_2,v)$
init $\forall r,v. \; \neg\rproposal(r,v)$
init $\forall n,r,v. \; \neg\rvote(n,r,v)$
init $\forall n,r,v. \; \neg\rdecision(n,r,v)$
init $\forall r. \; \neg\rany(r)$

action $\asendonea(\text{r} : \sround)$ {
    assume $\text{r} \neq \none$
    $\ronea(\text{r})$ := true
}
action $\ajoinround(\text{n} : \snode ,\, \text{r} : \sround)$ {
    assume $\text{r} \neq \none$
    assume $\ronea(\text{r})$
    assume $\neg \exists r',r'',v. \; r' > \text{r} \land \ronebmaxvote(\text{n},r',r'',v) \label{line:fp-join-round-if}$
    # find the maximal round in which n voted, 
    # and the corresponding vote.
    # maxr = $\bot$ and v is arbitrary when n never voted.
    local maxr, v := max $\{ (r',v') \mid \rvote(\text{n},r',v') \land r' < \text{r}\} \label{line:fp-join-round-max}$
    $\ronebmaxvote(\text{n},\text{r},\text{maxr},\text{v})$ := true $\label{line:fp-join-round-send}$
}
\end{lstlisting}
\end{minipage}%
\begin{minipage}[t]{.50\textwidth}
\begin{lstlisting}[
    %
    basicstyle=\scriptsize,%
    keepspaces=true,
    numbers=left,
    firstnumber=49,
    %
    xleftmargin=2em,
    numberstyle=\tiny,
    emph={
      %
      %
      %
      %
      %
      %
    },
    emphstyle={\bfseries},
    mathescape=true,
  ]
action $\acastvote$ = {
    # receive a 2a or "any" message and send 2b
    local n:$\snode$, v:$\svalue$, r:$\sround$ {
        assume r $\neq \none$
        assume $\not\exists r',r'',v .\; \ronebmaxvote(n,r',r'',v) \land r < r'$
        assume $\forall v . \neg\rvote(n, r, v)$
        # vote $\text{for}$ a proposal, or vote arbitrarily $\text{if}$ there is an "any" message.
        assume $\rproposal$(r, v) $\lor \rany$(r)
        $\rvote$(n, r, v) := true
    }
}
action $\apropose(r : \sround ,\, q : \scquorum)$ {
    assume $r \neq \none$
    assume $\forall v. \; \neg\rproposal(r,v) \label{line:fp-propose-assume-unique}$
    # 1b from quorum q
    assume $\forall n. \; \rcmember(n, \text{q}) \to \exists r',v. \;  \ronebmaxvote(n,\text{r},r',v) \label{line:fp-propose-assume-ae}$
    # find the maximal round in which a node in the quorum reported 
    # voting, and pick corresponding vote (there may be several).
    # v is arbitrary $\text{if}$ the nodes reported not voting.
    local maxr, v := max $\{ (r',v') \mid \exists n. \; \rcmember(n, q)$ 
                                               $\land \ronebmaxvote(n,r,r',v') \land r' \neq \none \} \label{line:fp-propose-max}$
    if (maxr $\neq \none$) {
      # a vote was reported in round maxr, 
      # and there are no votes in higher rounds.
      local $v'$ # the proposal the node will make
      assume 
          $(\exists f.\;  \forall n.\;  \rfmember(n,f) \land \rcmember(n,q)\rightarrow$ 
            $\ronebmaxvote(n,r,\text{maxr},v'))\label{line:fp-propose-choosable}$
          $\lor \; (v' = v \land \forall v'', f' .\;  \exists n.\;  \rfmember(n,f') \land \rcmember(n,q)$ 
                    $\land \neg\ronebmaxvote(n,r,\text{maxr},v''))$
      $\rproposal(r, v')$ := true
    } else { # no vote was reported at all.
      if $\rfast(r)$ {
        $\rany(r)$ := true # fast round, send any_msg
        } else {
        $\rproposal(r, v)$ := true # classic round, propose an arbitrary value
      }
    }
}
action $\acdecide(\text{n} : \snode, \text{r} : \sround ,\, \text{v} : \svalue ,\, \text{q} : \scquorum)$ {
    assume $\text{r} \neq \none$
    # 2b from classic quorum q:
    assume $\forall n. \; \rcmember(n, \text{q}) \to \rvote(n, \text{r}, \text{v})$  
    $\rdecision(\text{n}, \text{r}, \text{v})$ := true
}
action $\afdecide(\text{n} : \snode, \text{r} : \sround ,\, \text{v} : \svalue ,\, \text{q} : \sfquorum)$ {
    assume $\text{r} \neq \none$
    # 2b from fast quorum q:
    assume $\forall n. \; \rfmember(n, \text{q}) \to \rvote(n, \text{r}, \text{v})$  
    $\rdecision(\text{n}, \text{r}, \text{v})$ := true
}
\end{lstlisting}
\end{minipage}%
\end{minipage}%
\captionof{figure}{
Model of Fast Paxos in in many-sorted first-order logic.
}
\label{fig:fast-paxos-fol}
\end{figure}

\lstset{ %
  breakatwhitespace=false,         %
  deletekeywords={for},            %
  keywordstyle=\bf,       %
  language=C,                 %
  otherkeywords={if,else,module,individual,init,action,returns,assert,assume,instantiate,isolate,mixin,before,relation,function,sort,variable,axiom,then,constant,let,*,local},           %
  numbers=left,                    %
  numbersep=5pt,                   %
  numberstyle=\tiny,               %
  rulecolor=\color{black},         %
  tabsize=8,	                   %
   columns=fullflexible,
}

\begin{figure}
  \begin{minipage}{\textwidth}
    \begin{minipage}[t]{.50\textwidth}
\begin{lstlisting}[
    %
    basicstyle=\scriptsize,%
    keepspaces=true,
    numbers=left,
    %
    xleftmargin=2em,
    numberstyle=\tiny,
    emph={
      %
      %
      %
      %
      %
      %
    },
    emphstyle={\bfseries},
    mathescape=true,
  ]
sort $\snode$
sort $\sround$
sort $\svalue$
sort $\scquorum$ # classic quorums
sort $\sfquorum$ # fast quorums
relation $\rfast$ : $\sround$ # identifies fast rounds

relation $\leq$ : $\sround,\sround$
axiom total_order($\leq$)
constant $\none$ : $\sround$

relation $\rcmember$ : $\snode,\scquorum$
relation $\rfmember$ : $\snode,\sfquorum$
# classic quorums intersect
axiom $\forall q_1,q_2 : \scquorum .$ 
  $\exists n. \; \rcmember(n,q_1) \land \rcmember(n,q_2)$
# a classic quorum and a two fast quorums intersect
axiom $\forall q_1 : \scquorum . \forall q_2,q_3 : \sfquorum .$ 
  $\exists n. \; \rcmember(n,q_1) \land \rfmember(n,q_2) \land \rfmember(n,q_3)$

relation $\ronea$ : $\sround$
relation $\ronebmaxvote$ : $\snode,\sround,\sround,\svalue$
relation $\rproposal$ : $\sround,\svalue$
relation $\rvote$ : $\snode,\sround,\svalue$
relation $\rdecision$ : $\snode,\sround,\svalue$
relation $\rany$ : $\sround$ # the any messages
relation $\rleftround$ : $\snode,\sround$
relation $\ronebproj$ : $\snode,\sround$

init $\forall r. \; \neg\ronea(r)$
init $\forall n,r_1,r_2,v. \; \neg\ronebmaxvote(n,r_1,r_2,v)$
init $\forall r,v. \; \neg\rproposal(r,v)$
init $\forall n,r,v. \; \neg\rvote(n,r,v)$
init $\forall n,r,v. \; \neg\rdecision(n,r,v)$
init $\forall r. \; \neg\rany(r)$
init $\forall n,r. \; \neg \rleftround(n,r)$
init $\forall n,r. \; \neg \ronebproj(n,r)$

action $\asendonea(\text{r} : \sround)$ {
    assume $\text{r} \neq \none$
    $\ronea(\text{r})$ := true
}
action $\ajoinround(\text{n} : \snode ,\, \text{r} : \sround)$ {
    assume $\text{r} \neq \none$
    assume $\ronea(\text{r})$
    assume $\neg \rleftround(n,r)$ # rewritten
    local maxr, v := max $\{ (r',v') \mid \rvote(\text{n},r',v') \land r' < \text{r}\}$
    $\ronebmaxvote(\text{n},\text{r},\text{maxr},\text{v})$ := true 
    $\text{\# generated update code for derived relations:}$
    $\rleftround(\text{n},R)$ := $\rleftround(\text{n},R) \lor R < \text{r}$
    $\ronebproj(\text{n},\text{r})$ := true
}
\end{lstlisting}
\end{minipage}%
\begin{minipage}[t]{.50\textwidth}
\begin{lstlisting}[
    %
    basicstyle=\scriptsize,%
    keepspaces=true,
    numbers=left,
    firstnumber=53,
    %
    xleftmargin=2em,
    numberstyle=\tiny,
    emph={
      %
      %
      %
      %
      %
      %
    },
    emphstyle={\bfseries},
    mathescape=true,
  ]
action $\acastvote$ = {
    # receive a 2a or "any" message and send 2b
    local n:$\snode$, v:$\svalue$, r:$\sround$ {
        assume r $\neq \none$
        assume $\not\exists r',r'',v .\; \ronebmaxvote(n,r',r'',v) \land r < r'$
        assume $\forall v . \neg\rvote(n, r, v)$
        # vote $\text{for}$ a proposal, or vote arbitrarily $\text{if}$ there is an "any" message.
        assume $\rproposal$(r, v) $\lor \rany$(r)
        $\rvote$(n, r, v) := true
    }
}
action $\apropose(r : \sround ,\, q : \scquorum)$ {
    assume $r \neq \none$
    assume $\forall v. \; \neg\rproposal(r,v)$
    # 1b from quorum q
    assume $\forall n. \; \rcmember(n, q) \to \ronebproj(n,r)$ # rewriten
    # find the maximal round in which a node in the quorum reported 
    # voting, and pick corresponding vote (there may be several).
    # v is arbitrary $\text{if}$ the nodes reported not voting.
    local maxr, v := max $\{ (r',v') \mid \exists n. \; \rcmember(n, q)$ 
      $\land \rvote(n,r',v') \land r' \neq \none \land r' < r\}$ # rewritten
    if (maxr $\neq \none$) {
      # a vote was reported in maxr, and there are no votes in higher rounds.
      local $v'$ # the proposal the node will make
      assume 
          $(\exists f.\;  \forall n.\;  \rfmember(n,f) \land \rcmember(n,q) \rightarrow \rvote(n,\text{maxr},v'))$
          $\lor \; (v' = v \land \forall v'', f' .\;  \exists n.\;  \rfmember(n,f') \land \rcmember(n,q)$ 
                    $\land \neg\rvote(n,\text{maxr},v''))$ #rewritten
      $\rproposal(r, v')$ := true
    } else { # no vote was reported at all.
      if $\rfast(r)$ {
        $\rany(r)$ := true # fast round, send any_msg
        } else {
        $\rproposal(r, v)$ := true # classic round, propose an arbitrary value
      }
    }
}
action $\acdecide(\text{n} : \snode, \text{r} : \sround ,\, \text{v} : \svalue ,\, \text{q} : \scquorum)$ {
    assume $\text{r} \neq \none$
    # a classic quorum q sent 2b messages:
    assume $\forall n. \; \rcmember(n, \text{q}) \to \rvote(n, \text{r}, \text{v})$ 
    $\rdecision(\text{n}, \text{r}, \text{v})$ := true
}
action $\afdecide(\text{n} : \snode, \text{r} : \sround ,\, \text{v} : \svalue ,\, \text{q} : \sfquorum)$ {
    assume $\text{r} \neq \none$
    # a fast quorum q sent 2b messages:
    assume $\forall n. \; \rfmember(n, \text{q}) \to \rvote(n, \text{r}, \text{v})$
    $\rdecision(\text{n}, \text{r}, \text{v})$ := true
}
\end{lstlisting}
\end{minipage}%
\end{minipage}%
\captionof{figure}{
Model of Fast Paxos in EPR.
}
\label{fig:fast-paxos-epr}
\end{figure}

\subsection{Flexible Paxos}
\label{sec:flexible-paxos}

Flexible Paxos~\cite{howard_flexible_2016} extends Paxos based on the observation that in Paxos it is only necessary that every phase-1 quorum intersects with every phase-2 quorum (quorums of the same phase do not have to intersect).
Thus nodes may use different sets of quorums for phase 1 (to compute which value may be choosable in lower rounds) and for phase 2 (to get a value decided in the current round) as long as every phase-1 quorum  intersects every phase-2 quorum.
For example, in a system with a large number of nodes, one may choose that a value is decided if any set of 3 nodes (a phase-2 quorum) votes for it, and at the same time require that a node wait for $n-2$ nodes (a phase-1 quorum) to join its round when starting a new round, where $n$ is the total number of nodes. The opposite configuration is also possible (i.e. phase-1 quorums of size 3 and phase-2 quorums of size $n-2$). This approach allows a trade-off between the cost of starting a new round and the cost of deciding on a value.
Note that no inconsistency may arise due to the fact that same-phase quorums do not intersect because the leader of a round proposes a unique value.

To model Flexible Paxos in IVY we introduce two sorts for the two different types of quorums, $quorum\_1$ and $quorum\_2$, and we modify the actions to use quorums of sort $quorum\_1$ in phase 1 and of sort $quorum\_2$ in phase 2.
We also adapt the quorum intersection axiom:
\[
\forall q_1:\squorum_1, q_2:\squorum_2. \; \exists n:\snode. \;
\rmember_1(n,q_1) \land \rmember_2(n, q_2).
\]

The derived relations and rewriting steps are the same as in Paxos, and the safety property and inductive invariant is also the same as in Paxos except that the quorums are taken from the sort $quorum\_2$ in~\cref{eq:quorum-of-decision} and~\cref{eq:choosable}.

The FOL model of Flexible Paxos appears in~\cref{fig:flexible-paxos-fol}.

\lstset{ %
  breakatwhitespace=false,         %
  keywordstyle=\bf,       %
  language=C,                 %
  otherkeywords={module,individual,init,action,returns,assert,assume,instantiate,isolate,mixin,before,relation,function,sort,variable,axiom,then,constant,let,*,local},           %
  numbers=left,                    %
  numbersep=5pt,                   %
  numberstyle=\tiny,               %
  rulecolor=\color{black},         %
  tabsize=8,	                   %
   columns=fullflexible,
}

\begin{figure}
\begin{minipage}{\textwidth}
\begin{lstlisting}[
    %
    basicstyle=\scriptsize,%
    keepspaces=true,
    numbers=left,
    %
    xleftmargin=2em,
    numberstyle=\tiny,
    emph={
      %
      %
      %
      %
      %
      %
    },
    emphstyle={\bfseries},
    mathescape=true,
  ]
sort $\snode$, $\sround$, $\svalue$
# two types of quorums
sort $\sort{quorum\_1}$
sort $\sort{quorum\_2}$

relation $\leq$ : $\sround,\sround$
axiom total_order($\leq$)
constant $\none$ : $\sround$

relation $\relation{member\_1}(n:\snode, q:\sort{quorum\_1})$
relation $\relation{member\_2}(n:\snode, q:\sort{quorum\_2})$
axiom $\forall q_1 : \sort{quorum\_1},q_2 : \sort{quorum\_2}. \;  \exists n:\snode. \;  \relation{member\_1}(n,q_1) \land \relation{member\_2}(n, q_2)$

relation $\ronea$ : $\sround$
relation $\ronebmaxvote$ : $\snode,\sround,\sround,\svalue$
relation $\rproposal$ : $\sround,\svalue$
relation $\rvote$ : $\snode,\sround,\svalue$
relation $\rdecision$ : $\snode,\sround,\svalue$

init $\forall r. \; \neg\ronea(r)$
init $\forall n,r_1,r_2,v. \; \neg\ronebmaxvote(n,r_1,r_2,v)$
init $\forall r,v. \; \neg\rproposal(r,v)$
init $\forall n,r,v. \; \neg\rvote(n,r,v)$
init $\forall n,r,v. \; \neg\rdecision(n,r,v)$

action $\asendonea(\text{r} : \sround)$ {
    assume $\text{r} \neq \none$
    $\ronea(\text{r})$ := true
}
action $\ajoinround(\text{n} : \snode ,\, \text{r} : \sround)$ {
    assume $\text{r} \neq \none$
    assume $\ronea(\text{r})$
    assume $\neg \exists r',r'',v. \; r' > \text{r} \land \ronebmaxvote(\text{n},r',r'',v)$
    # find the maximal round in which n voted, and the corresponding vote.
    # maxr = $\bot$ and v is arbitrary when n never voted.
    local maxr, v := max $\{ (r',v') \mid \rvote(\text{n},r',v') \land r' < \text{r}\}$
    $\ronebmaxvote(\text{n},\text{r},\text{maxr},\text{v})$ := true
}
action $\apropose(\text{r} : \sround ,\, \text{q} : \sort{quorum\_1})$ {
    assume $\text{r} \neq \none$
    assume $\forall v. \; \neg\rproposal(\text{r},v)$
    # 1b from a phase-1 quorum q:
    assume $\forall n. \; \relation{member\_1}(n, \text{q}) \to \exists r',v. \;  \ronebmaxvote(n,\text{r},r',v)$
    # find the maximal round in which a node in the quorum reported
    # voting, and the corresponding vote.
    # v is arbitrary $\text{if}$ the nodes reported not voting.
    local maxr, v := max $\{ (r',v') \mid \exists n. \; \relation{member\_1}(n, \text{q})$
                                               $\land \ronebmaxvote(n,\text{r},r',v') \land r' \neq \none \}$
    $\rproposal(\text{r}, \text{v})$ := true # propose value v
}
action $\acastvote(\text{n} : \snode ,\, \text{r} : \sround ,\, \text{v} : \svalue)$ {
    assume $\text{r} \neq \none$
    assume $\rproposal(\text{r}, \text{v})$
    assume $\neg \exists r',r'',v. \;  r' > \text{r} \land \ronebmaxvote(\text{n},r',r'',v)$
    $\rvote(\text{n}, \text{r}, \text{v})$ := true
}
action $\adecide(\text{n} : \snode, \text{r} : \sround ,\, \text{v} : \svalue ,\, \text{q} : \sort{quorum\_2})$ {
    assume $\text{r} \neq \none$
    assume $\forall n. \; \relation{member\_2}(n, \text{q}) \to \rvote(n, \text{r}, \text{v})$  # 2b from a phase-2 quorum q
    $\rdecision(\text{n}, \text{r}, \text{v})$ := true
}
\end{lstlisting}
\captionof{figure}{Model of the Flexible Paxos consensus algorithm as a transition system in many-sorted first-order logic}
  \label{fig:flexible-paxos-fol}
\end{minipage}%
\end{figure}

\subsection{Stoppable Paxos}
\label{sec:stoppable-paxos}

Stoppable Paxos~\cite{lamport_stoppable_2008} extends Multi-Paxos with
the ability for a node to propose a special stop command in order to
stop the algorithm, with the guarantee that if the stop command is
decided in instance $i$, then no command is ever decided at an
instance $j>i$.  Stoppable Paxos therefore enables Virtually
Synchronous system
reconfiguration~\cite{birman_virtually_2010,chockler_group_2001}:
Stoppable Paxos stops in a state known to all participants, which can
then start a new instance of Stoppable Paxos in a new configuration
(e.g., in which participants have been added or removed); moreover, no
pending commands can leak from a configuration to the next, as only
the final state of the command sequence is transfered from one
configuration to the next.

Stoppable Paxos may be the most intricate algorithm in the Paxos
family: as acknowledged by Lamport et
al.~\cite{lamport_stoppable_2008}, ``getting the details right was not
easy''.  The main algorithmic difficulty in Stoppable Paxos is to
ensure that no command may be decided after a stop command while at
the same time allowing a node to propose new commands without waiting,
when earlier commands are still in flight (which is important for
performance).  In contrast, in the traditional approach to
reconfigurable SMR~\cite{lamport_reconfiguring_2010}, a node that has
$c$ outstanding command proposals may cause up to $c$ commands to be
decided after a stop command is decided; Those commands needs to be
passed-on to the next configuration and may contain other stop
commands, adding to the complexity of the reconfiguration system.

Before proposing a command in an instance in Stoppable Paxos, a node
must check if other instances have seen stop commands proposed and in
which round.  This creates a non-trivial dependency between rounds and
instances, which are mostly orthogonal concepts in other variants of
Paxos. This manifest as the instance sort having no incoming edge in
the quantifier alternation graph in other variants, while such edges
appear in Stoppable Paxos. Interestingly, the rule given by Lamport
et al. to propose commands results in verification conditions that are
not in EPR, and rewriting seems difficult.  However, we found an
alternative rule which results in EPR verification conditions.  This
alternative rule soundly overapproximates the original rule (and
introduces new behaviors), and, as we have verified (in EPR), it also
maintains safety. The details of the modified rule and its
verification appear in

Stoppable Paxos uses the same messages and actions as Multi-Paxos, and
a special command value $\rstop$. In addition to the usual consensus
property of Paxos, Stoppable Paxos ensures the following safety
property, ensuring that nothing is ever decided after a $\rstop$ value is decided:
\begin{equation} \label{eq:stoppable-safety}
\begin{split}
& \forall i_1,i_2:\sinst, n_1,n_2:\snode, r_1,r_2:\sround, v:\svalue. \; \\
& \qquad \qquad \qquad \qquad \rdecision(n_1,i_1,r_1,\rstop) \land i_2 > i_1 \to \neg \rdecision(n_2,i_2,r_2,v)
\end{split}
\end{equation}
In order to obtain this property, Stoppable Paxos uses an intricate condition
in the $\ainstate$ action, to ensure that if a $\rstop$ value is choosable at
instance $i_1$, then no value is be proposed for any instance $i_2 > i_1$.
Recall that in Multi-Paxos, the owner of round $r$ takes the $\ainstate$ action
once it has received join-acknowledgment messages from a quorum of nodes.
These messages allow to compute the maximal vote (by round number) in each
instance, by any node in the quorum. Let $m$ denote the $\svotemap$
representing these voted (as computed in \Cref{fig:multi-paxos-fol}
\cref{line:multi-fol-instate-max}).  In Multi-Paxos, these votes are simply
re-proposed for round $r$.  However, in Stoppable Paxos we must take special
care for $\rstop$ commands.  Suppose that for some instance $i_1$, we have
$\rroundof(m,i_1) \neq \none \land \rvalueof(m,i_1) = \rstop$. Naively, this
suggests we should not propose any value for instances larger than $i_1$;
otherwise, if the stop value is eventually decided at $i_i$, we will violate
the safety property that requires that no value be decided after a stop command.
However, it could be that for some $i_2 > i_1$, we also find $\rroundof(m,i_2)
\neq \none$. Here we face a dilemma: if the stop value at $i_1$ is eventually
decided, proposing at $i_2$ may lead to a safety violation; but if the stop
value at $i_1$ is in fact not choosable and the value at $i_2$ is eventually
decided, then re-proposing the stop value at $i_1$ may lead to a safety
violation too.  As explained in \cite{lamport_stoppable_2008}, there is a way
to ensure that either the stop value at $i_1$ is choosable or the other value
at $i_2$ is choosable, but not both, and to know which one is choosable. The solution
depends on whether $\rroundof(m,i_2) > \rroundof(m,i_1)$, in which case the
$\rstop$ command for $i_1$ cannot be choosable and is \emph{voided} by treating
it as if $\rroundof(m,i_1) = \none$.  Otherwise, the value at $i_2$ cannot be
choosable and the owner should propose the $\rstop$ command for $i_1$, and not
propose any other values for instances larger than $i_1$.

The rule described in \cite{lamport_stoppable_2008} is to first
compute which $\rstop$ commands are voided, and then to propose all
commands except those made impossible by a non-voided $\rstop$ command
(i.e., a non-voided stop command at a lower instance). Formalizing
this introduces cyclic quantifier alternation over instances, since
the condition for voiding a $\rstop$ command involves an existential
quantifier over instances. Formally, let $m_L$ denote the $\svotemap$
obtained from $m$ by voiding according to the rule of
\cite{lamport_stoppable_2008} (where it is called
$\textit{sval2a}$). $m_L$ is given by:

\begin{equation*}\label{eq:stoppable-lamport-sval2a}
\begin{split}
& \forall i:\sinst. \; (\rvalueof(m_L,i) = \rvalueof(m,i)) \land (\varphi_\textit{void} \to \rroundof(m_L,i) = \none) \land \\
& \qquad \qquad \qquad \qquad \qquad \qquad \qquad \qquad \qquad  (\neg \varphi_\textit{void} \to \rroundof(m_L,i) = \rroundof(m,i)) \\
& \text{where:} \\
& \varphi_\textit{void} = \exists i':\sinst . \; i' > i \land \rroundof(m,i') \neq \none \land \rroundof(m,i') > \rroundof(m,i)
\end{split}
\end{equation*}
Then, the rule in \cite{lamport_stoppable_2008} makes proposals for
instances $i$ that satisfy the condition condition:
\begin{equation*}\label{eq:stoppable-lamport-propose}
\rroundof(m_L,i) \neq \none \land \forall i':\sinst . \; i' < i \land \rvalueof(m_L,i') = \rstop \to \rroundof(m_L,i') = \none
\end{equation*}
This introduces yet another quantifier alternation cycle, since it
must be applied with universal quantification to all instances.

To avoid this cyclic quantifier alternation, we observe that a relaxed
rule can be used, and verify a realization of Stoppable Paxos based on
our relaxed rule. The relaxed rule avoids the quantifier alternation,
and it also provides an overapproximation of the rule of
\cite{lamport_stoppable_2008}. The relaxed rule is to first find the
$\rstop$ command with the highest round in $m$, and then check if it
is voided (by a value at a higher instance and higher round). If it is
not voided, then we void all other stop commands, and all proposals at
higher instances. If the maximal stop is voided, we void all stop
commands, as well as all proposals with higher instances and lower
rounds (compared to the $\rstop$ command with the highest round). This
rule does not lead to any quantifie alternation cycles.

\subsubsection{Model of the Protocol}

\lstset{ %
  breakatwhitespace=false,         %
  keywordstyle=\bf,       %
  language=C,                 %
  otherkeywords={module,individual,init,action,returns,assert,assume,instantiate,isolate,mixin,before,relation,function,sort,variable,axiom,then,constant,let,*,local},           %
  numbers=left,                    %
  numbersep=5pt,                   %
  numberstyle=\tiny,               %
  rulecolor=\color{black},         %
  tabsize=8,	                   %
   columns=fullflexible,
}

\begin{figure}
\begin{lstlisting}[
    %
    basicstyle=\scriptsize,%
    keepspaces=true,
    numbers=left,
    %
    xleftmargin=2em,
    numberstyle=\tiny,
    emph={
      %
      %
      %
      %
      %
      %
    },
    emphstyle={\bfseries},
    mathescape=true,
  ]
sort $\snode$
sort $\squorum$
sort $\sround$
sort $\svalue$
sort $\sinst$
sort $\svotemap$

relation $\leq_r$ : $\sround,\sround$
relation $\leq_i$ : $\sinst,\sinst$
axiom total_order($\leq_r$)
axiom total_order($\leq_i$)
constant $\none$ : $\sround$
constant $\rstop$ : $\sinst$
relation $\rmember$ : $\snode,\squorum$
axiom $\forall q_1,q_2 : \squorum. \;  \exists n:\snode. \;  \rmember(n,q_1) \land \rmember(n, q_2) \label{line:stoppable-fol-axiom-quorum}$

relation $\ronea$ : $\sround$
relation $\ronebmaxvote$ : $\snode,\sround,\svotemap$
relation $\rproposal$ : $\sinst,\sround,\svalue$
relation $\ractive$ : $\sround$
relation $\rvote$ : $\snode,\sinst,\sround,\svalue$
relation $\rdecision$ : $\snode,\sinst,\sround,\svalue$
function $\rroundof$ : $\svotemap,\sinst \to \sround$
function $\rvalueof$ : $\svotemap,\sinst \to \svalue$

init $\forall r. \; \neg\ronea(r)$
init $\forall n,r,m. \; \neg\ronebmaxvote(n,r,m)$
init $\forall i,r,v. \; \neg\rproposal(i,r,v)$
init $\forall r. \; \neg\ractive(r)$
init $\forall n,i,r,v. \; \neg\rvote(n,i,r,v)$
init $\forall r,v. \; \neg\rdecision(r,v)$

action $\asendonea$  # same as Multi-Paxos ($\text{\Cref{fig:multi-paxos-fol} \cref{line:multi-fol-sendonea}}$)
action $\ajoinround$  # same as Multi-Paxos ($\text{\Cref{fig:multi-paxos-fol} \cref{line:multi-fol-joinround}}$)
action $\acastvote$  # same as Multi-Paxos ($\text{\Cref{fig:multi-paxos-fol} \cref{line:multi-fol-castvote}}$)
action $\adecide$  # same as Multi-Paxos ($\text{\Cref{fig:multi-paxos-fol} \cref{line:multi-fol-decide}}$)
action $\ainstate(\text{r} : \sround ,\, \text{q} : \squorum)$ {
    assume $\text{r} \neq \none \land \neg\ractive(\text{r}) \label{line:stoppable-fol-instate-assume-once}$
    assume $\forall n. \; \rmember(n, \text{q}) \to \exists m. \;  \ronebmaxvote(n,\text{r},m) \label{line:stoppable-fol-instate-assume-ae}$
    local m : $\svotemap$ := *
    assume $\forall i. \; (\rroundof(\text{m},i),\rvalueof(\text{m},i)) = \text{max }\left\{ (r',v') \mid \exists n,m'. \; \rmember(n, \text{q}) \land \ronebmaxvote(n,\text{r},m') \, \land \right. \label{line:stoppable-fol-instate-max}$
                                                                                                                $\left.  r' = \rroundof(m',\text{i}) \land v' = \rvalueof(m',\text{i}) \land r' \neq \none \, \right\}$
    $\ractive(\text{r})$ := $\true$ $\label{line:stoppable-fol-instate-active}$
    if $(\forall i. \rroundof(\text{m},i) \neq \none \to \rvalueof(\text{m},i) \neq \rstop)$ {  # no stops in m $\label{line:stoppable-fol-instate-if-1}$
            $\rproposal(I,\text{r},V)$ := $\rproposal(I,\text{r},V) \lor (\rroundof(\text{m},I) \neq \none \land V = \rvalueof(\text{m},I))$ $\label{line:stoppable-fol-instate-propose-1}$
    } else {  # find maximal stop in m and propose accordingly
        local $\text{i}_\text{s}$ : $\sinst$ := *
        assume $\rroundof(\text{m},\text{i}_\text{s}) \neq \none \land \rvalueof(\text{m},\text{i}_\text{s}) = \rstop \, \land \label{line:stoppable-fol-instate-max-stop}$
                   $\forall i . \; (\rroundof(\text{m},i) \neq \none \land \rvalueof(\text{m},i) = \rstop) \to \rroundof(\text{m},i) \leq_r \rroundof(\text{m},\text{i}_\text{s})$ 
        if $(\exists i. \; i >_i \text{i}_\text{s} \land \rroundof(\text{m},i) \neq \none \land \rroundof(\text{m},i) >_r \rroundof(\text{m},\text{i}_\text{s}) )$ { # maximal stop is voided $\label{line:stoppable-fol-instate-if-2}$
            $\rproposal(I,\text{r},V)$ := $\rproposal(I,\text{r},V) \lor (\rroundof(\text{m},I) \neq \none \land V = \rvalueof(\text{m},I) \land V \neq \rstop \, \land $ $\label{line:stoppable-fol-instate-propose-2}$
                                                                                           $(I >_i \text{i}_\text{s} \to \rroundof(\text{m},I) >_r \rroundof(\text{m},\text{i}_\text{s})))$
        } else { # maximal stop not voided
            $\rproposal(I,\text{r},V)$ := $\rproposal(I,\text{r},V) \lor (\rroundof(\text{m},I) \neq \none \land V = \rvalueof(\text{m},I) \land I \leq_i \text{i}_\text{s} \land (V = \rstop \to I = \text{i}_\text{s}))$ $\label{line:stoppable-fol-instate-propose-3}$
        }
    }
}
action $\aproposenew(\text{r} : \sround ,\, \text{i} : \sinst ,\, \text{v} : \svalue)$ {
    assume $\text{r} \neq \none \land \ractive(\text{r}) \land \forall v. \; \neg\rproposal(\text{i},\text{r},v) \label{line:stoppable-fol-propose-available-1}$
    assume $\forall i. \; i \leq_i \text{i} \to \neg \rproposal(i, \text{r}, \rstop)) \label{line:stoppable-fol-propose-available-2}$
    assume $\text{v} = \rstop \to \forall i,v. \; \text{i} \leq_i i \to \neg\rproposal(i,\text{r},v) \label{line:stoppable-fol-propose-available-3}$
    $\rproposal(\text{r}, \text{v})$ := true $\label{line:stoppable-fol-propose-send}$
}
\end{lstlisting}
\caption{%
\label{fig:stoppable-paxos-fol}%
Model of Stoppable Paxos as a transition system in many-sorted first-order
logic.%
}
\end{figure}

Our model of Stoppable Paxos in first-order logic appears in
\Cref{fig:stoppable-paxos-fol}. The only actions that differ from
Multi-Paxos (\Cref{fig:multi-paxos-fol}) are $\ainstate$ and
$\apropose$. The rule described above is implemented in the
$\ainstate$ action. \Cref{line:stoppable-fol-instate-max} computes $m$
as in Multi-Paxos. Then, \cref{line:stoppable-fol-instate-if-1} checks
if there are any $\rstop$ commands reported in $m$. If there are no
$\rstop$ commands, then \cref{line:stoppable-fol-instate-propose-1}
proposes exactly as in Multi-Paxos (\Cref{fig:multi-paxos-fol}
\cref{line:multi-fol-instate-propose}). If there are $\rstop$ commands
in $m$, then \cref{line:stoppable-fol-instate-max-stop} finds
$\text{i}_\text{s}$, the instance of the stop command with the highest
round present in $m$. \Cref{line:stoppable-fol-instate-if-2} then
checks if this $\rstop$ command is voided by a value present in $m$ at
a higher instance and with a higher round. If so, then
\cref{line:stoppable-fol-instate-propose-2} proposes all values from
$m$ that are not $\rstop$ commands, excluding those which are voided
by the $\rstop$ at $\text{i}_\text{s}$, i.e., those at a higher
instance and lower round. In case the $\rstop$ at $\text{i}_\text{s}$
is not voided, \cref{line:stoppable-fol-instate-propose-3} proposes
all non-stop commands until $\text{i}_\text{s}$, as well as a $\rstop$
at $\text{i}_\text{s}$, and does not propose anything for higher
instances.

The $\apropose$ action is similar to Multi-Paxos, but contains a few
changes as described in \cite{lamport_stoppable_2008}. First, we may
not propose any value if we have proposed a $\rstop$ command at a
lower instance
(\cref{line:stoppable-fol-propose-available-2}). Second, we may only
propose a stop value at an instance such that we have not already
proposed anything at higher instances
(\cref{line:stoppable-fol-propose-available-3}).

Our rule for Stoppable Paxos provides an overapproximation of the rule
of \cite{lamport_stoppable_2008}, in that it forces less proposals in
the $\ainstate$ action. Thus, any behaviour of
\cite{lamport_stoppable_2008} can be simulated by an $\ainstate$
action followed by several $\apropose$ actions to produce the missing
proposals. We have also verified this using IVy and Z3, albeit the
verification conditions were outside of EPR. Nevertheless, Z3 was able
to verify this in under 2 seconds. This is in contrast to verifying
the inductive invariant for the version of
\cite{lamport_stoppable_2008}, for which Z3 diverged.

\subsubsection{Inductive Invariant}

The inductive invariant for Stoppable Paxos contains the inductive
invariant for Multi-Paxos (\Cref{sec:multi-paxos-inv}), and in
addition includes conjuncts that capture the special meaning of
$\rstop$ commands, and ensure the additional safety property of
\cref{eq:stoppable-safety}. Below we list these additional conjuncts.

First, the inductive invariant asserts that a $\rstop$ command
proposed at some instance forbids proposals of other commands in higher instances \emph{in the same round}:
\begin{small}
\begin{equation}
\forall i_1,i_2:\sinst,r:\sround,v:\svalue. \;
\rproposal(i_1,r,\rstop) \land i_2 >_i i_1 \to \neg \rproposal(i_2,r,v)
\end{equation}
\end{small}

Next, the inductive invariant connects different rounds via the
\emph{choosable} concept (see \Cref{sec:paxos-fol-inv},
\cref{eq:choosable}). If any value is proposed, then $\rstop$ command
cannot be choosable at lower instances and lower rounds:

\begin{small}
\begin{equation} \label{eq:stoppable-choosable-1}
\begin{split}
& \forall i_1,i_2:\sinst,r_1,r_2:\sround,v:\svalue, q:\squorum. \;
\rproposal(i_2,r_2,v) \land r_1 <_r r_2 \land i_1 <_i i_2 \to \\
& \qquad \qquad \exists n:\snode, r', r'':\sround, v':\svalue. \; \\
& \qquad \qquad \qquad \qquad
\rmember(n,q) \land \neg \rvote(n,r_1,\rstop) \land r' > r_1 \land \ronebmaxvote(n,r',r'',v')
\end{split}
\end{equation}
\end{small}
And, in addition, if $\rstop$ is proposed, than nothing can be choosable at lower rounds and higher instances:
\begin{small}
\begin{equation} \label{eq:stoppable-choosable-2}
\begin{split}
& \forall i_1,i_2:\sinst,r_1,r_2:\sround,v:\svalue, q:\squorum. \;
\rproposal(i_1,r_2,\rstop) \land r_1 <_r r_2 \land i_1 <_i i_2 \to \\
& \qquad \qquad \exists n:\snode, r', r'':\sround, v':\svalue. \; \\
& \qquad \qquad \qquad \qquad
\rmember(n,q) \land \neg \rvote(n,r_1,v) \land r' > r_1 \land \ronebmaxvote(n,r',r'',v')
\end{split}
\end{equation}
\end{small}

\subsubsection{Transformation to EPR}

The quantifier alternation structure of \Cref{fig:stoppable-paxos-fol}
and the inductive invariant described above is the same as the one
obtained for Multi-Paxos. Notice that
\cref{eq:stoppable-choosable-1,eq:stoppable-choosable-2} introduce a
quantifier alternation cycle identical to the one introduced by
\cref{eq:multi-choosable} for Multi-Paxos (and by \cref{eq:choosable}
for single decree Paxos). Thus, the transformation of the Stoppable
Paxos model to EPR is identical to that of Multi-Paxos
(\Cref{sec:multi-paxos-epr}), using the same derived relations
$\ronebproj$ and $\rleftround$, when $\rleftround$ is used to rewrite
\cref{eq:stoppable-choosable-1,eq:stoppable-choosable-2} in the
inductive invariant in the same way it was used in
\Cref{sec:paxos-epr}. This again shows the reusability of the derived
relations across many variants of Paxos.

\section{Comparing the EPR Methodology with Interactive Theorem Proving}
\label{sec:isabelle}

To compare the methodolgy presented in this paper with one based on an interactive proof assistant, we proved the safety of a Paxos model almost identical to the IVy model in Isabelle/HOL (the proof figures in ~\cref{sec:paxos_isabelle} and is written for Isabelle-2016-1).
Isabelle/HOL~\cite{nipkow_isabelle/hol:_2002} is an interactive theorem prover with a small trusted kernel that checks all inferences, and can offload proof obligations to integrated automated provers or external SMT solvers whose proofs are then automatically reconstructed and checked in Isabelle/HOL.
The Isabelle/HOL model of Paxos uses the theory of I/O-Automata~\cite{lynch_hierarchical_1987} (we reuse an existing formalization) and closely follows the IVy model, except that no encoding of higher-order concepts is needed in Isabelle/HOL because its logic is a higher-order logic: we use natural numbers for rounds and sets of nodes with cardinality constraints for quorums.

The Isabelle/HOL proof is structured as follows:
We define an abstraction of the single-shot Paxos model which we prove prove safe.
This abstraction models the core of the Paxos algorithm, namely how a node proposing a new command computes which command is choosable at rounds lower than its round.
The proof is facilitated by the high level of abstraction.
Then, using a refinement mapping, we prove that the Paxos specification implements the abstraction.
The full proof of Paxos consists of roughly 200 lines.

The proof process in Isabelle/HOL involves two steps: first writing a high-level proof skeleton (including the definition of the abstraction, the refinement mapping, and all the invariants needed) in which some proof steps are missing (e.g. that an invariant is inductive and that the refinement mapping is correct).
This proof skeleton is then debugged with the Nipick~\cite{blanchette_nitpick:_2010} model-finder.
When we are satisfied that our proof skeleton is correct, we turn to proving all the steps left out.

The first step is similar to proving in IVy, except that instead of rewriting the model to make it fall in EPR we built the abstraction and the proof skeleton.
On our model of Paxos, the Nitpick model-finder exhibits reliable performance and thus this step is comparable in difficulty to proving with IVy, although the graphical user interface of IVy helps understand counterexample faster.
However, thanks to decidability of EPR, with IVy we are done at the end of this step.

In Isabelle/HOL, proving that the steps left out are correct is still a challenge.
To do so, we may have to refine our proof skeleton until all steps can be proved by the automated provers available in Isabelle/HOL.
This task can be time consuming and requires an experienced user who can gauge what can be proved automatically and what not.
Moreover, in refining the proof skeleton, one may try to isolate facts relevant for a sub-proof, and mistakenly omit necessary assumptions.
Therefore the proof skeleton must continuously be debugged with the model-finder to ensure that the user does not take a dead-end.
The Isabelle/HOL proof is completed in about 200 lines of proof script.

One tricky aspect of the Isabelle/HOL proof is the proof that the propose action preserves~\cref{eq:choosable}, which is automatic with IVy.
We reason as follows.
To show that~\cref{eq:choosable} is maintained by the \apropose\ action, assume it is not, i.e. that the owner $p$ of $r_2$ proposes $v_2$ in $r_2$, and that $v_1\neq v_2$ is choosable at $r_1< r_2$.
Consider the highest round $r_3$ reported by the quorum $q'$ of nodes that acknowledged $p$'s start-round message, and the associated value $v_3$.
If $r_3=\bot$, then no node in $q'$ voted before $r_2$. Therefore, for any quorum $q$, since $q$ intersects $q'$, one node of $q$ joined $r_2>r_1$ and did not vote in $r_1$. Therefore $v_1$ cannot be choosable in $r_1$, a contradiction.
Now assume $r_3\neq \bot$. This implies that $v_2=v_3$.
Consider the following three cases: either (a) $r_1<r_3$, (b) $r_1=r_3$, or (c) $r_1>r_3$.
Assume (a), $r_1<r_3$.  \Cref{eq:vote-proposal} implies that $v_3$ was proposed at $r_3$.
Then by~\cref{eq:choosable} and the fact that $v_1$ is choosable at $r_1<r_3$, we get that $v_1=v_3$. With $v_2=v_3$ we get that $v_1=v_2$, a contradiction.
Assume (b), $r_1=r_3$. Since $v_3$ has been voted for at $r_3$, the proposal at $r_3$ is $v_3$ and only $v_3$ can be choosable at $r_3$. With $v_2=v_3$ we get that $v_1=v_2$, a contradiction.
Finally, assume (c), $r_1>r_3$. We know that all node in $q'$ did not vote between $r_3$ and $r_2$. Therefore, for any quorum $q$, since $q$ intersects $q'$, one node of $q$ joined $r_2>r_1$ and did not vote in $r_1$. Therefore $v_1$ cannot be choosable in $r_1$, a contradiction.
We have reached a contradiction in all cases.

Finally, we note that the EPR methodology could be implemented in a proof assistant like Isabelle/HOL thanks to its ``smt'' method, which discharges proof obligations to SMT solvers and reconstructs SMT proof in Isabelle/HOL's logic.



\section*{Supplementary material}

## B.1 Safety Proof of Paxos in Isabelle/HOL

# 1 The Paxos I/O-Automaton

**theory** *Paxos*
  **imports** *Main Simulations ~~/src/HOL/Eisbach/Eisbach-Tools*
**begin**

The theory Simulations is part of the I/O-Automata formalization found at
https://github.com/nano-o/IO-Automata

**interpretation** *IOA* .

**datatype** $'v$ *paxos-action* =
  *Internal* | *Decision* $'v$

We specify Paxos following the IVy FOL specification. Rounds are natural numbers, and We use the round *0* to represent the round $\bot$. We use the type variable $'n$ for nodes and $'v$ for values.

**definition** *paxos-asig* **where**
  *paxos-asig* $\equiv$ (| *inputs* = {}, *outputs* = {*Decision v* | *v* . *True*}, *internals* = {*Internal*} |)

**record** $('v,'n)$ *state* =
  *start-round-msg* :: *nat set*
  *join-ack-msg* :: $('n \times nat \times nat \times 'v)$ *set*
  *propose-msg* :: $(nat \times 'v)$ *set*
  *vote-msg* :: $('n \times nat \times 'v)$ *set*
  *decision* :: $(nat \times 'v)$ *set*

**locale** *paxos-ioa* =
  **fixes** *quorums*::$'n$ *set set*
  — We will make assumptions on quorums only in the correctness proof.
**begin**

**definition** *start* **where**
  — The initial state
  *start* $\equiv$ {(| *start-round-msg* = {}, *join-ack-msg* = {}, *propose-msg* = {},
    *vote-msg* = {}, *decision* = {} |)}

**definition** *start-round* **where**
  *start-round r s s'* $\equiv$
    $r \neq 0 \wedge s' = s$(|*start-round-msg* := *start-round-msg s* $\cup$ {$r$}|)

**definition** *join-round* **where**
  *join-round n r maxr maxv s s'* $\equiv$
    $r \neq 0 \wedge r \in$ *start-round-msg s* $\wedge$
    $(\forall\ r' \geq r\ .\ \forall\ r''\ v\ .\ (n,\ r',\ r'',\ v) \notin$ *join-ack-msg s*) $\wedge$
    $((maxr = 0 \wedge (\forall\ r'\ v\ .\ r' < r \longrightarrow (n,\ r',\ v) \notin$ *vote-msg s*)) $\vee$
      $(maxr \neq 0 \wedge maxr < r \wedge (n,\ maxr,\ maxv) \in$ *vote-msg s*
        $\wedge\ (\forall\ r'\ v\ .\ r' < r \wedge (n,\ r',\ v) \in$ *vote-msg s* $\longrightarrow r' \leq maxr$))) $\wedge$
    $s' = s$(|*join-ack-msg* := *join-ack-msg s* $\cup$ {$(n,\ r,\ maxr,\ maxv)$}|)

**definition** *propose* **where**
  *propose r q v s s'* ≡
    $r \neq 0 \land (\forall\ v\ .\ (r,\ v) \notin propose\text{-}msg\ s) \land q \in quorums\ \land$
    $(\forall\ n \in q\ .\ \exists\ maxr\ maxv\ .\ (n,\ r,\ maxr,\ maxv) \in join\text{-}ack\text{-}msg\ s)\ \land$
    $(\exists\ n\ maxr\ maxv\ .\ (n,\ r,\ maxr,\ maxv) \in join\text{-}ack\text{-}msg\ s\ \land$
      $(\forall\ n'\ r\ maxr'\ maxv'\ .\ n' \in q \land (n',\ r,\ maxr',\ maxv') \in join\text{-}ack\text{-}msg\ s \longrightarrow maxr' \leq maxr)\ \land$
      $(maxr = 0 \lor v = maxv)\ \land$
      $s' = s(\!|propose\text{-}msg := propose\text{-}msg\ s \cup \{(r,\ v)\}|\!))$

**definition** *vote* **where**
  *vote n r v s s'* ≡
    $r \neq 0 \land (r,\ v) \in propose\text{-}msg\ s\ \land$
    $(\forall\ r'\ maxr\ maxv\ .\ r' > r \longrightarrow (n,\ r',\ maxr,\ maxv) \notin join\text{-}ack\text{-}msg\ s)\ \land$
    $(\exists\ maxr\ maxv\ .\ (n,r,maxr,maxv) \in join\text{-}ack\text{-}msg\ s)\ \land$
    $s' = s(\!|vote\text{-}msg := vote\text{-}msg\ s \cup \{(n,\ r,\ v)\}|\!)$

**definition** *decide* **where**
  *decide r v q s s'* ≡
    $r \neq 0 \land q \in quorums \land (\forall\ n \in q\ .\ (n,\ r,\ v) \in vote\text{-}msg\ s)\ \land$
    $s' = s(\!|decision := decision\ s \cup \{(r,\ v)\}|\!)$

**fun** *trans-rel* :: $('v,\ 'n)\ state \Rightarrow 'v\ paxos\text{-}action \Rightarrow ('v,\ 'n)\ state \Rightarrow bool$ **where**
  *trans-rel s Internal s'* = (
    $(\exists\ r\ .\ start\text{-}round\ r\ s\ s') \lor$
    $(\exists\ n\ r\ maxr\ maxv\ .\ join\text{-}round\ n\ r\ maxr\ maxv\ s\ s') \lor$
    $(\exists\ r\ q\ v\ .\ propose\ r\ q\ v\ s\ s') \lor$
    $(\exists\ n\ r\ v\ .\ vote\ n\ r\ v\ s\ s')\ )\ |$
  *trans-rel s (Decision v) s'* = $(\exists\ r\ q\ .\ decide\ r\ v\ q\ s\ s')$

**definition** *trans* **where**
  $trans \equiv \{\ (s,a,s')\ .\ trans\text{-}rel\ s\ a\ s'\}$

**definition** *paxos-ioa* **where**
  $paxos\text{-}ioa \equiv (\!|ioa.asig = paxos\text{-}asig,\ start = start,\ trans = trans|\!)$

**lemmas** *simps = paxos-ioa-def paxos-asig-def start-def start-round-def trans-def propose-def join-round-def*
  *vote-def decide-def*

**end**

# 2 The Abstract Paxos I/O-automaton

**record** $('v,'n)$ *abstract-state* =
  *votes* :: $('n \times nat \times 'v)\ set$
  *joined-round* :: $('n \times nat)\ set$
  *proposal* :: $(nat \times 'v)\ set$

$decision :: (nat \times {'}v)\ set$

**locale** $abstract\text{-}paxos\text{-}ioa =$
  **fixes** $quorums::{'}n\ set\ set$
**begin**

**definition** *start* **where**
  — The initial state
  $start \equiv \{(\!|\ votes = \{\},\ joined\text{-}round = \{\},\ proposal = \{\},\ decision = \{\}\ |\!)\}$

**definition** *vote* **where**
  $vote\ n\ r\ v\ s\ s' \equiv$
    $r \neq 0 \wedge (r,\ v) \in proposal\ s \wedge (\forall\ r' > r\ .\ (n,\ r') \notin joined\text{-}round\ s) \wedge$
    $(n,\ r) \in joined\text{-}round\ s \wedge s' = s(\!|votes := (votes\ s) \cup \{(n,\ r,\ v)\}|\!)$

**definition** *join-round* **where**
  $join\text{-}round\ n\ r\ s\ s' \equiv$
    $(\forall\ r' \geq r\ .\ (n,\ r') \notin joined\text{-}round\ s) \wedge s' = s(\!|joined\text{-}round := (joined\text{-}round\ s) \cup \{(n,\ r)\}|\!)$

**definition** *propose* **where**
  $propose\ r\ q\ v\ s\ s' \equiv$
    $r \neq 0 \wedge q \in quorums \wedge$
    $(\forall\ v\ .\ (r,\ v) \notin proposal\ s) \wedge$
    $(\forall\ n \in q\ .\ (n,\ r) \in joined\text{-}round\ s) \wedge$
    $(\exists\ maxr\ maxn\ .\ maxr < r\ \wedge$
      $((maxr = 0 \wedge (\forall\ n \in q\ .\ \forall\ r' < r\ .\ \forall\ v\ .\ (n,\ r',\ v) \notin votes\ s)) \vee$
      $(maxr \neq 0 \wedge (maxn,\ maxr,\ v) \in votes\ s \wedge (\forall\ n \in q\ .\ \forall\ r' < r\ .\ \forall\ v\ .$
        $(n,\ r',\ v) \in votes\ s \longrightarrow r' \leq maxr)))) \wedge$
    $s' = s(\!|proposal := (proposal\ s) \cup \{(r,v)\}|\!)$

**definition** *decide* **where**
  $decide\ r\ v\ s\ s' \equiv$
    $r \neq 0 \wedge (\exists\ q \in quorums\ .\ \forall\ n \in q\ .\ (n,\ r,\ v) \in votes\ s) \wedge$
    $s' = s(\!|decision := decision\ s \cup \{(r,\ v)\}|\!)$

**fun** $trans\text{-}rel :: ({'}v,\ {'}n)\ abstract\text{-}state \Rightarrow {'}v\ paxos\text{-}action \Rightarrow ({'}v,\ {'}n)\ abstract\text{-}state$
$\Rightarrow bool$ **where**
  $trans\text{-}rel\ s\ Internal\ s' = ($
    $(\exists\ n\ r\ .\ join\text{-}round\ n\ r\ s\ s') \vee$
    $(\exists\ r\ q\ v\ .\ propose\ r\ q\ v\ s\ s') \vee$
    $(\exists\ n\ r\ v\ .\ vote\ n\ r\ v\ s\ s'))\ |$
  $trans\text{-}rel\ s\ (Decision\ v)\ s' = (\exists\ r\ .\ decide\ r\ v\ s\ s')$

**definition** *trans* **where**
  $trans \equiv \{\ (s,a,s')\ .\ trans\text{-}rel\ s\ a\ s'\}$

**definition** $abstract\text{-}paxos\text{-}ioa :: (({'}v,\ {'}n)\ abstract\text{-}state,\ {'}v\ paxos\text{-}action)\ ioa$ **where**
  $abstract\text{-}paxos\text{-}ioa \equiv (\!|ioa.asig = paxos\text{-}asig,\ start = start,\ trans = trans|\!)$

**end**

# 3 Proof Automation setup

A very simple verification-condition generator for proving invariants, that unfolds definitions and inserts proved invariants as premises.

**named-theorems** *invs* **named-theorems** *inv-defs* **named-theorems** *ioa-defs*

**method** *rm-reachable* = (*match* **premises in** *R*[*thin*]:*reachable ?ioa ?s* ⇒ ⟨−⟩)

**lemma** *reach-and-inv-imp-p*:⟦*reachable ioa s*; *invariant ioa i*⟧ ⟹ *i s*
 **by** (*auto simp add*:*invariant-def*)

**method** *instantiate-invs* **declares** *invs* =
 (*match* **premises in** *I*[*thin*]:*invariant ?ioa ?inv* **and** *R*:*reachable ?ioa ?s* ⇒ ⟨
  *print-fact I*, *insert reach-and-inv-imp-p*[*OF R I*]⟩)+

**method** *inv-vcg* **declares** *invs ioa-defs inv-defs* = (
 *rule invariantI*,
 *force simp add*:*ioa-defs inv-defs*,
 (*insert invs*, *instantiate-invs*)?, *rm-reachable*, *simp add*:*ioa-defs*)

**locale** *paxos-proof* = *abs*:*abstract-paxos-ioa quorums* + *paxos-ioa quorums*
  **for** *quorums* :: ′*n set set* +
  **fixes** *abs-ioa* :: ((′*v*, ′*n*) *abstract-state*, ′*v paxos-action*) *ioa*
   **and** *con-ioa* :: ((′*v*, ′*n*) *state*, ′*v paxos-action*) *ioa*
  **defines** *abs-ioa* ≡ *abs.abstract-paxos-ioa* **and** *con-ioa* ≡ *paxos-ioa*
  **assumes** ⋀ *q1 q2* . ⟦*q1* ∈ *quorums*; *q2* ∈ *quorums*⟧ ⟹ *q1* ∩ *q2* ≠ {}
  **and** *quorums* ≠ {} **and** ⋀ *q* . *q* ∈ *quorums* ⟹ *finite q*
  **and** *finite quorums*
**begin**

**lemma** *quorum-inter-witness*[*elim*]:
 **assumes** *q1* ∈ *quorums* **and** *q2* ∈ *quorums*
 **obtains** *a* **where** *a* ∈ *q1* **and** *a* ∈ *q2*
 **using** *assms paxos-proof-axioms*
 **by** (*metis disjoint-iff-not-equal paxos-proof-def*)

**lemmas** *asig-simps* = *externals-def abs.abstract-paxos-ioa-def paxos-asig-def con-ioa-def abs-ioa-def paxos-ioa-def*
**lemmas** *trans-simps* = *con-ioa-def paxos-ioa-def start-def abs-ioa-def abs.abstract-paxos-ioa-def abs.start-def*
 *is-trans-def trans-def abs.trans-def*
**lemmas** *act-defs* = *abs.propose-def abs.join-round-def abs.vote-def abs.decide-def*
 *vote-def start-round-def join-round-def decide-def propose-def*

**declare** *trans-simps*[*ioa-defs*]

# 4 Proof of Abstract Paxos

**definition** *safety* **where** *safety-def*[*inv-defs*]:
  *safety s* ≡ ∀ *v v′ r r′* . (*r*,*v*) ∈ *decision s* ∧ (*r′*,*v′*) ∈ *decision s* ⟶ *v* = *v′*

**definition** *invariant-1* **where** *invariant-1-def*[*inv-defs*]:
  *invariant-1 s* ≡ (∀ *n r v* . (*n, r, v*) ∈ *votes s* ⟶ (*r, v*) ∈ *proposal s*)
    ∧ (∀ *r v v′* . (*r, v*) ∈ *proposal s* ∧ (*r, v′*) ∈ *proposal s* ⟶ *v* = *v′*)

**definition** *choosable* **where**
  *choosable s r v* ≡ ∃ *q* ∈ *quorums* . ∀ *n* ∈ *q* .
    (∃ *r′* > *r* . (*n, r′*) ∈ *joined-round s*) ⟶ (*n, r, v*) ∈ *votes s*

**definition** *invariant-2* **where** *invariant-2-def*[*inv-defs*]:
  *invariant-2 s* ≡ ∀ *r r′ v v′* . *r′* < *r* ∧ (*r, v*) ∈ *proposal s* ∧
    *choosable s r′ v′* ⟶ *v* = *v′*

**definition** *invariant-3* **where** *invariant-3-def*[*inv-defs*]:
  *invariant-3 s* ≡ ∀ *v r* . (*r*,*v*) ∈ *decision s* ⟶ (∃ *q* ∈ *quorums* . ∀ *n* ∈ *q* . (*n, r, v*) ∈ *votes s*)

**declare** *act-defs*[*simp*]

**lemma** *abs-trans-cases*:
  — A case split rule for analyzing the transition relation by cases
  **assumes** *abs.trans-rel s a s′*
  **obtains**
  (*abs-join-round*) *n r* **where** *abs.join-round n r s s′* **and** *a* = *Internal* |
  (*abs-propose*) *r q v* **where** *abs.propose r q v s s′* **and** *a* = *Internal* |
  (*abs-vote*) *n r v* **where** *abs.vote n r v s s′* **and** *a* = *Internal* |
  (*abs-decide*) *r v* **where** *abs.decide r v s s′* **and** *a* = *Decision v*
  **by** (*meson abstract-paxos-ioa.trans-rel.elims*(*2*) *assms*)

**lemma** *invariant-1*:
  *invariant abs-ioa invariant-1*
  **apply** (*inv-vcg, elim abs-trans-cases*)
  **apply** (*auto simp add*:*inv-defs*)
  **done**
**declare** *invariant-1*[*invs*]

**lemma** *invariant-2*:
  *invariant abs-ioa invariant-2*
  **apply** (*inv-vcg*)
  **subgoal premises** *prems* **for** *s s′ a* **using** *prems*(*2*)
  **proof** (*cases rule*:*abs-trans-cases*)
    **case** (*abs-join-round n r*)
    **have** *votes s′* = *votes s*
      **and** ⋀ *n r* . (*n, r*) ∈ *joined-round s* ⟹ (*n*,*r*) ∈ *joined-round s′*
      **and** *proposal s′* = *proposal s* **using** *abs-join-round* **by** *simp-all*

55

```
    thus ?thesis using prems(1) apply (auto simp add:invariant-2-def choosable-def)
      by meson
  next
    case (abs-propose r q v)
    have votes s′ = votes s and joined-round s′ = joined-round s
      and proposal s′ = (proposal s) ∪ {(r,v)}
      using abs-propose by auto
    moreover
    have v′ = v if r′ < r and choosable s r′ v′ for v′ r′
    proof −
      have q-joined:∀ n ∈ q . (n,r) ∈ joined-round s and q ∈ quorums using
abs-propose
        by simp-all
      obtain rmax n where q ∈ quorums and rmax < r
        (rmax = 0 ∧ (∀ n ∈ q . ∀ r′ < r . ∀ v . (n, r′, v) ∉ votes s))
        ∨ (rmax ≠ 0 ∧ (n, rmax, v) ∈ votes s
          ∧ (∀ n′ ∈ q . ∀ r′ < r . ∀ v . (n′, r′, v) ∈ votes s ⟶ r′ ≤ rmax))
        using abs-propose by auto
      with this consider
        (no-votes) rmax = 0 and ∀ n ∈ q . ∀ r′ < r . ∀ v . (n, r′, v) ∉ votes s |
        (votes) rmax ≠ 0 and (n, rmax, v) ∈ votes s
          and ∀ n′ ∈ q . ∀ r′ < r . ∀ v . (n′, r′, v) ∈ votes s ⟶ r′ ≤ rmax by
auto
      thus ?thesis
      proof (cases)
        case no-votes
        hence False using q-joined that ⟨q ∈ quorums⟩ apply (auto simp add:choosable-def)
          by (meson quorum-inter-witness)
        thus ?thesis by auto
      next
        case votes
        from ⟨choosable s r′ v′⟩ obtain n′ where (n′, r′, v′) ∈ votes s
          using ⟨q ∈ quorums⟩ ⟨r′ < r⟩ q-joined by (force simp add:choosable-def)
        consider (a) r′ < rmax | (b) r′ = rmax | (c) r′ > rmax
          using nat-neq-iff by blast
        then show ?thesis
        proof (cases)
          case a
          with prems(1,3) votes(2) that(2) ⟨rmax < r⟩ show ?thesis
            by (auto simp add:invariant-2-def invariant-1-def, blast)
        next
          case b
          from ⟨choosable s r′ v′⟩ obtain n′ where ( n′, r′, v′) ∈ votes s
            using ⟨q ∈ quorums⟩ ⟨r′ < r⟩ q-joined by (force simp add:choosable-def)
          with prems(3) votes(2) b show ?thesis by (force simp add:invariant-1-def
split:option.splits)
        next
          case c
          hence ∀ v . (n, r′, v) ∉ votes s if n ∈ q for n
```

```
          using votes(3) that ‹r' < r› by force
        with ‹choosable s r' v'› ‹q ∈ quorums› q-joined ‹r' < r› have False by
(force simp add:choosable-def)
        thus ?thesis by auto
      qed
    qed
  qed
  ultimately show ?thesis using prems(1) apply (auto simp add:invariant-2-def
choosable-def)
    by (blast, meson)
next
  case (abs-vote n r v)
  have joined-round s' = joined-round s and proposal s' = proposal s
    and ⋀ n r v . (n, r, v) ∈ votes s' ∧ (n, r, v) ∉ votes s ⟶
      ((n, r) ∈ joined-round s ∧ (∀ r' > r . (n, r') ∉ joined-round s)) using
abs-vote by auto
  thus ?thesis using prems(1) prems(2) apply (auto simp add:invariant-2-def
invariant-1-def choosable-def)
      by meson
next
  case (abs-decide r v)
  then show ?thesis using prems(1) by (auto simp add:invariant-2-def choosable-def)
qed
done
declare invariant-2[invs]

lemma invariant-3:
  invariant abs-ioa invariant-3
  apply (inv-vcg, elim abs-trans-cases)
    apply (simp-all add:inv-defs)
  apply metis
  done
declare invariant-3[invs]

theorem safety:
  invariant abs-ioa safety
proof −
  have safety s if invariant-3 s and invariant-2 s and invariant-1 s for s using
that
    apply (auto simp add:inv-defs choosable-def)
    by (metis not-less-iff-gr-or-eq quorum-inter-witness)
  thus ?thesis
    using IOA.invariant-def invariant-1 invariant-2 invariant-3 by blast
qed

declare invariant-1[invs del] invariant-2[invs del] invariant-3[invs del]
```

57

# 5 Proof of Paxos by refinement to abstract paxos.

**definition** *invariant-4* **where** *invariant-4-def* [*inv-defs*]:
  *invariant-4 s* ≡ ∀ *n r maxr maxv* . (*n, r, maxr, maxv*) ∈ *join-ack-msg s* ∧ *maxr* ≠ *0*
    ⟶ (*maxr* < *r* ∧ (*n, maxr, maxv*) ∈ *vote-msg s* ∧ (∀ *r′* < *r* . ∀ *v* . (*n, r′, v*) ∈ *vote-msg s* ⟶ *r′* ≤ *maxr*))

**definition** *invariant-5* **where** *invariant-5-def* [*inv-defs*]:
  *invariant-5 s* ≡ ∀ *n r maxr maxv* . (*n, r, maxr, maxv*) ∈ *join-ack-msg s* ∧ *maxr* = *0*
    ⟶ (∀ *r′ v* . *r′* < *r* ⟶ (*n, r′, v*) ∉ *vote-msg s*)

**lemma** *trans-cases*:
  — A case split rule for analyzing the transition relation by cases
  **assumes** *trans-rel s a s′*
  **obtains**
  (*start-round*) *r* **where** *start-round r s s′* **and** *a* = *Internal* |
  (*join-round*) *n r maxr maxv* **where** *join-round n r maxr maxv s s′* **and** *a* = *Internal* |
  (*propose*) *r q v* **where** *propose r q v s s′* **and** *a* = *Internal* |
  (*vote*) *n r v* **where** *vote n r v s s′* **and** *a* = *Internal* |
  (*decide*) *r v q* **where** *decide r v q s s′* **and** *a* = *Decision v*
  **by** (*meson assms trans-rel.elims(2)*)

**lemma** *invariant-4*:
  *invariant con-ioa invariant-4*
  **apply** (*inv-vcg, elim trans-cases*)
      **apply** (*force simp add:inv-defs*)
     **apply** (*simp add:inv-defs*) **apply** *blast*
    **apply** (*simp add:inv-defs*) **apply** *force*
   **apply** (*simp add:inv-defs*) **apply** *blast*
  **apply** (*simp add:inv-defs*)
  **done**
**declare** *invariant-4* [*invs*]

**lemma** *invariant-5*:
  *invariant con-ioa invariant-5*
  **apply** (*inv-vcg, elim trans-cases*)
      **apply** (*force simp add:inv-defs*) **defer**
     **apply** (*simp add:inv-defs*)
   **apply** (*metis state.select-convs(2) state.select-convs(4) state.surjective state.update-convs(3)*)
    **apply** (*force simp add:inv-defs*)
   **apply** (*force simp add:inv-defs*)
  **apply** (*simp add:inv-defs*) **apply** *blast*
  **done**
**declare** *invariant-5* [*invs*]

**declare** *act-defs* [*simp del*]

**definition** *ref-map* **where**
  *ref-map s ≡* (|*votes = vote-msg s, joined-round = {(n,r) . ∃ r′ v . (n,r,r′,v) ∈ join-ack-msg s},*
    *proposal = propose-msg s, decision = state.decision s*|)

**lemmas** *ref-proof-simps = ref-map-def asig-simps trans-simps trace-def schedule-def filter-act-def*

**lemma** *refinement*:
  **shows** *is-ref-map ref-map con-ioa abs-ioa*
  **apply** (*auto simp add:is-ref-map-def*)
   **apply** (*simp add:ref-map-def ioa-defs*)
  **apply** (*insert invs*)
  **apply** (*instantiate-invs*)
  **apply** (*simp add:trans-simps refines-def trace-match-def*)
  **subgoal premises** *prems* **for** *s t a* **using** *prems*(*2*)
  **proof** (*induct rule:trans-cases*)
    **case** (*start-round r*)
    **then show** *?case*
      **apply** (*intro exI*[**where** *x*=[]])
      **apply** (*auto simp add:ref-proof-simps start-round-def*)
      **done**
  **next**
    **case** (*join-round n r maxr maxv*)
    **let** *?s′ = ref-map s*
    **let** *?t′ = (ref-map s)*(|*joined-round := joined-round (ref-map s) ∪ {(n,r)}*|)
    **have** *abs.join-round n r ?s′ ?t′* **using** *join-round*(*1*)
      **by** (*auto simp add:join-round-def abs.join-round-def ref-map-def*)
    **moreover**
     **have** *ref-map t = ?t′* **using** *join-round*(*1*) **by** (*auto simp add:join-round-def ref-map-def*)
    **ultimately show** *?case* **using** *join-round*(*2*)
      **apply** (*intro exI*[**where** *x*=[(*Internal, ?t′*)]])
      **apply** (*auto simp add:ref-proof-simps*)
      **done**
  **next**
    **case** (*propose r q v*)
    **let** *?s′ = ref-map s*
    **let** *?t′ = (ref-map s)*(|*proposal := proposal (ref-map s) ∪ {(r,v)}*|)
    **have** *abs.propose r q v ?s′ ?t′*
    **proof** −
      **obtain** *n maxr maxv* **where** *con-pre*:*r ≠ 0 ∧ (∀ v . (r, v) ∉ propose-msg s) ∧ q ∈ quorums ∧*
          (*∀ n′ ∈ q . ∃ maxr maxv . (n′,r,maxr,maxv) ∈ join-ack-msg s*)
          *∧ (n, r, maxr, maxv) ∈ join-ack-msg s ∧*
          (*∀ n′ r maxr′ maxv′ . n′ ∈ q ∧ (n′, r, maxr′, maxv′) ∈ join-ack-msg s ⟶ maxr′ ≤ maxr*) *∧*
          (*maxr = 0 ∨ v = maxv*) **using** *propose*(*1*) **by** (*auto simp add:propose-def*)

```
        have r ≠ 0 ∧ q ∈ quorums ∧ (∀ v . (r, v) ∉ proposal ?s') ∧
            (∀ n ∈ q . (n, r) ∈ joined-round ?s') using con-pre by (auto simp
add:ref-map-def)
        moreover have
          maxr < r ∧ ((maxr = 0 ∧ (∀ n ∈ q . ∀ r' < r . ∀ v . (n, r', v) ∉ votes
?s')) ∨
            (maxr ≠ 0 ∧ (n, maxr, v) ∈ votes ?s' ∧ (∀ n ∈ q . ∀ r' < r . ∀ v .
            (n, r', v) ∈ votes ?s' ⟶ r' ≤ maxr)))
        proof –
          show ?thesis
          proof (cases maxr = 0)
            case True
            hence ∀ n ∈ q . ∀ r' < r . ∀ v . (n, r', v) ∉ votes ?s' and r>0 using
prems(4) con-pre
              apply (simp-all add:ref-map-def invariant-5-def) by blast
            thus ?thesis using True by auto
          next
            case False
            then show ?thesis using prems(3,4) con-pre
              apply (auto simp add:ref-map-def invariant-5-def invariant-4-def)
            subgoal — Proof automatically generated by Sledgehammer      done
          qed
        qed
        moreover have ?t' = ?s'(|proposal := (proposal ?s') ∪ {(r,v)}|)
          by (auto simp add:ref-map-def)
        ultimately show ?thesis by (auto simp add:abs.propose-def)
      qed
      moreover
    have ref-map t = ?t' using propose(1) by (auto simp add:propose-def ref-map-def)
      ultimately show ?case using propose(2)
        apply (intro exI[where x=[(Internal, ?t')]])
        apply (auto simp add: ref-proof-simps)
        done
    next
      case (vote n r v)
      let ?s' = ref-map s
      let ?t' = (ref-map s)(|votes := votes (ref-map s) ∪ {(n,r,v)}|)
      have abs.vote n r v ?s' ?t'
      proof –
        from vote(1) have r ≠ 0 ∧ (r, v) ∈ proposal ?s' ∧ (∀ r' > r . (n, r') ∉
joined-round ?s') ∧
          (n, r) ∈ joined-round ?s'
          by (auto simp add:vote-def ref-map-def)
        moreover
        from vote(1) have ?t' = ?s'(|votes := (votes ?s') ∪ {(n, r, v)}|)
          by (auto simp add:vote-def ref-map-def)
        ultimately show ?thesis by (auto simp add:abs.vote-def)
      qed
      moreover have ref-map t = ?t' using vote(1) by (auto simp add:vote-def
```

```
    ref-map-def)
      ultimately show ?case using vote(2)
        apply (intro exI[where x=[(Internal,
          (ref-map s)(|votes := votes (ref-map s) ∪ {(n,r,v)}|))]])
        apply (auto simp add:ref-proof-simps)
        done
    next
      case (decide r v q)
      then show ?case
        apply (intro exI[where x=[(Decision v,
          (ref-map s)(|decision := decision (ref-map s) ∪ {(r,v)}|))]])
        apply (auto simp add:ref-proof-simps act-defs)
        done
    qed
  done

theorem trace-inclusion:
  traces con-ioa ⊆ traces abs-ioa
proof −
  have externals (ioa.asig con-ioa) = externals (ioa.asig abs-ioa)
    by (simp add:asig-simps)
  thus ?thesis using refinement ref-map-soundness
    by blast
qed

end

end
```



\end{document}